%% file: arxiv.tex
\documentclass[acmsmall,screen,nonacm]{acmart}
\usepackage[T1]{fontenc}
\usepackage{graphicx}
\usepackage{braket}
\usepackage{amsmath}
\allowdisplaybreaks
\usepackage{physics}
\usepackage{nicefrac}
\usepackage{hyperref}
\usepackage{stmaryrd}
\usepackage{color}
\usepackage{todonotes}
\usepackage{booktabs}
\usepackage[most]{tcolorbox}
\usepackage{enumitem}
\usepackage{bm}
\usepackage{mathtools}

\usepackage{attachfile2}
\usepackage{embedfile}
\usepackage{fontawesome5} 

\usepackage{thmtools}
\usetikzlibrary{automata, arrows.meta,positioning,shapes,shapes.geometric, shapes.multipart}
\begin{document}
\include{macro}
\title[Quantum Pre-expectations for Expected Runtime]{Quantum Weakest Preconditions Revisited: Pre-expectations for Expected Runtime Analysis}
\author{Christina Gehnen}
\orcid{0000-0002-6548-3432}
\affiliation{ \institution{RWTH Aachen University} \country{Germany}}
\email{christina.gehnen@cs.rwth-aachen.de}

\author{Dominique Unruh}
\orcid{0000-0001-8965-1931}
\affiliation{ \institution{RWTH Aachen University} \country{Germany}}
\email{unruh@cs.rwth-aachen.de}

\author{Joost-Pieter Katoen}
\orcid{0000-0002-6143-1926}
\affiliation{ \institution{RWTH Aachen University} \country{Germany}}
\email{katoen@cs.rwth-aachen.de}

\begin{abstract}
Quantum weakest preconditions are a fundamental tool for program verification of quantum programs. Many variations have been reported in the literature. We revisit quantum weakest preconditions from the perspective of expected runtime analysis of quantum programs and introduce a novel pre-expectation framework that enables to reason about the preconditions of quantum programs without the need of an upper bound. This is particularly interesting for quantum programs involving reward statements. The overall goal is to analyze runtime behavior even in the case of programs with potentially infinite expected runtime. This paper presents several ways to do so, e.g., a program transformation such that the expected runtime of a quantum program can be expressed using the weakest pre-expectation calculus with rewards.

\end{abstract}
\maketitle

\input{1_introduction}
\input{2_preliminaries}
\input{3_expectations.tex}

\input{4_programs.tex}
\input{5_wp.tex}
\input{6_ert.tex}

\input{7_examples.tex}
\input{8_conclusion.tex}

\begin{acks}
    This work is supported by the European Union's Horizon 2020 research and innovation programme under the Marie Skłodowska-Curie grant agreement \href{https://cordis.europa.eu/project/id/101008233}{No. 101008233} (MISSION), and the Interdisciplinary Doctoral Program in Quantum Systems Integration funded by the BMW group.
\end{acks}

\bibliographystyle{splncs04}
\bibliography{references}
\appendix

\input{appendix.tex}
\end{document}

%% file: macro.tex
\newcommand*{\density}{\densitygen{\hilbert} }
\newcommand*{\densitygen}[1]{\mathcal{D}(#1)}
\newcommand*{\traceclass}{\mathcal{T}(\hilbert)}
\newcommand*{\domain}[1]{\mathfrak{D}(#1)}
\newcommand*{\expect}[2]{\mathbb{E}_{#1}\left(#2\right)}
\newcommand*{\loewner}{\sqsubseteq}

\newcommand*{\hilbert}{\mathcal{H}}

\newcommand{\programs}{Prog(\hilbert)}

\newcommand*{\predicate}{\mathcal{E}(\hilbert)}
\newcommand*{\formPredicate}{\mathcal{F}(\hilbert)}

\newcommand*{\hoare}[3]{\{#1\} #2 \{#3\}}

\newcommand*{\semantics}[1]{\llbracket #1 \rrbracket}
\newcommand{\semanticsOrig}[1]{\semantics{#1}_{\text{orig}}}
\newcommand*{\semanticsR}[1]{\semantics{#1}_{R}}
\newcommand*{\semanticsD}[1]{\semantics{#1}_D}

\newcommand*{\R}{\mathbb{R}}
\newcommand*{\C}{\mathbb{C}}
\newcommand*{\extR}{\R_{\geq 0}^\infty}

\newcommand*{\skipbf}{\textbf{skip}}
\newcommand*{\qzero}{q:=0}
\newcommand*{\Uq}{\overline{q}:= U\overline{q}}
\newcommand*{\measure}{\textbf{measure } M[\overline{q}]:\overline{S}}
\newcommand*{\measurePrime}{\textbf{measure } M[\overline{q}]:\overline{S'}}
\newcommand*{\concat}{S_1;S_2}
\newcommand*{\while}{\textbf{while } M[\overline{q}]=1 \textbf{ do } S}
\newcommand*{\reward}[1]{\textbf{reward} (#1)}
\newcommand*{\rewardR}{\reward{ c}}

\newcommand*{\qwp}[2]{wp\llbracket #1 \rrbracket (#2)}
\newcommand*{\qwlp}[2]{wlp\llbracket #1 \rrbracket (#2)}
\newcommand*{\ert}[2]{ert\llbracket #1 \rrbracket (#2)}
\newcommand*{\qwpForm}[2]{wp_f\llbracket #1 \rrbracket (#2)}
\newcommand*{\ERT}[2]{ERT\llbracket #1 \rrbracket (#2)}
\newcommand*{\expReward}[1]{\llbracket #1 \rrbracket_R}

\newcommand*{\intType}{Int}
\newcommand*{\boolType}{Bool}

\newcommand*{\transform}[1]{T({#1})}

\newcommand*{\df}{:=}

\newcommand*{\identityOp}{\textbf{I}}
\newcommand*{\identityOpVar}[1]{\identityOp_{#1}}
\newcommand*{\zeroOp}{\textbf{0}}
\newcommand*{\zeroOpVar}[1]{\zeroOp_{#1}}
\newcommand*{\inftyOp}{\bm{\infty}}

\newcommand*{\half}[1]{#1^{\nicefrac{1}{2}}}
\newcommand*{\sumform}{\oplus}
\newcommand{\formsandwich}[2]{#1^\dagger \odot #2 \odot #1}

\tcbset{
  theoremstyle/.style={
    enhanced jigsaw,
    colback=white,           
    colframe=#1,   
    borderline west={3pt}{0pt}{#1},
    boxrule=0pt,
    sharp corners,
    left=10pt,
    right=10pt,
    top=6pt,
    bottom=6pt,
    before skip=10pt,
    after skip=10pt,
  }
}

\definecolor{rwthblue}{RGB}{0,84,159}
\definecolor{rwthyellow}{RGB}{255,237,0}
\definecolor{rwthgreen}{RGB}{87,171,39}
\definecolor{rwthorange}{RGB}{246,168,0}
\definecolor{rwthred}{RGB}{204,7,30}
\definecolor{rwthpetrol}{RGB}{0,97,101}
\definecolor{rwthbordeaux}{RGB}{161,16,53}
\definecolor{rwthlila}{RGB}{122,111,172}
\definecolor{rwthturquoise}{RGB}{0,152,161}
\definecolor{rwthmaigreen}{RGB}{189,205,0}
\definecolor{rwthmagenta}{RGB}{227,0,102}

\theoremstyle{remark}
\newtheorem{remark}{Remark}

%% file: 1_introduction.tex
\section{Introduction}
\label{sec:introduction}
Quantum programming is an emerging field that combines the principles of quantum mechanics with computer science to develop algorithms and software for quantum computers. As quantum computing continues to advance, there is a growing need for rigorous methods to analyze and verify quantum programs. Besides of correctness, runtime analysis is an important factor in quantum algorithms, as many of them are designed to solve problems more efficiently than classical algorithms. One fundamental principle in this area is deductive verification, in particular the concept of weakest preconditions, which allows to reason about the behavior of quantum programs and their expected outcomes.
The notion of (quantitative) weakest preconditions in the quantum setting was first introduced by D'Hondt and Panangaden in \cite{DHondtWeakestPreconditions} and since then there have been many variations and extensions, but all of them require the weakest preconditions to be bounded. In this paper, we challenge this requirement and introduce a novel framework for reasoning about quantum weakest preconditions without the need for an upper bound. Our approach is motivated by the analysis of expected runtimes of quantum programs as in \cite{LiuRuntime} but we do not want to restrict ourselves to (1) finite dimensional Hilbert spaces and (2) almost surely terminating (AST) programs.
By allowing infinite-dimensions and unbounded operators, we can analyze a wider range of quantum programs and provide insights into their expected runtimes. An interesting phenomenon that arises in the infinite-dimensional context only is that there can be programs which are almost surely terminating but that still have an infinite expected runtime (an example can be found in Section \ref{sec:examples} and the equivalence for finite-dimensions in \cite[Thm. 2]{LiuRuntime}).
The expected runtime of a (non-AST) program cannot be expressed naively using weakest preconditions and a counter-variable as this small example illustrates: $q:= \ket{0}; \textbf{while } \identityOp[q] \textbf{ do }\{ \skipbf \}$. This program has an expected runtime of $\infty$ but the weakest precondition of the loop is $\zeroOp$ (as it never terminates) and thus of the whole program as well. This shows the need for a more general framework to analyze expected runtimes of quantum programs that are not almost surely terminating.

A more interesting example is the following variation of the quantum walk which cannot be analyzed using existing work like \cite{LiuRuntime, olmedoRuntime}. We have a quantum integer variable $q$ which tracks the current position and a qubit $c$ which is used to control the speed of the walk.
\begin{align*}
    & c:= \ket{0}; \\
    & \textbf{while } M[q] \textbf{ do }\{ \\
    & \hspace*{1cm} c := H c; \\
    & \hspace*{1cm} q c := S q c; \\
    & \}
\end{align*}
with $M_0 = \ket{0}\bra{0}$, $M_1 = \identityOp - \ket{0}\bra{0}$ and the (controlled) shift operator $S\ket{n}\ket{0} = \ket{n-1}\ket{0}, S\ket{n}\ket{1}  =\ket{n}\ket{1}$. In comparison to the quantum walk, the shift operator $S$ does not shift in both directions but instead either shifts to the left (decrease the value of $q$) or stays in the same position depending on the value of $c$. The program terminates as soon as $q$ reaches $\ket{0}$.
For states $\ket{n}$ with $n\geq 0$ this program is almost surely terminating, as it will eventually reach the state $\ket{0}$ with probability 1. In Section \ref{sec:examples} we will see that the expected runtime for $\ket{n}$ is finite for all $n\geq 0$. In contrast to the probabilistic setting, where this implies that the expected runtime is finite for all possible states, we have to be more careful here: for superpositions the expected runtime can still be infinite. This shows that we cannot carry over all methods and results from the probabilistic setting to the quantum setting and that we need to generalize the framework to analyze expected runtimes of quantum programs like this.

As we want to provide a comprehensive and expressive framework for analyzing the expected runtime of quantum programs, we do this in terms of weakest pre-expectations using reward statements. This allows not only to express expected runtimes but also to reason about expected values of other observables.

Our main contributions can be summarized as follows:
\begin{itemize}
    \item We define semantic properties for quantum programs with rewards in infinite-dimensional spaces and give a concrete syntax and denotational semantics for qrWhile programs with rewards.
    \item We define weakest pre-expectations both semantically and syntactically for qrWhile programs. We drop the boundedness condition of predicates which was required in previous work and thus, our expectations are strictly more expressive.
    \item As we consider unbounded expectations, we can analyze expected runtimes without restricting to almost surely terminating programs. We show different ways to compute the expected runtime using syntactical transformations and weakest pre-expectations, giving a concrete forward and backward reasoning approach and show that all of them coincide.
    \item We give some simple rules to bound both the weakest preconditions and the expected runtimes of loops that apply in the bounded case as well.
    \item We consider example programs and analyze their expected runtime and termination probability using the tools we developed which would not have been possible using existing work.
\end{itemize}
We allow infinite-dimensional Hilbert spaces to enable the analysis of a wider range of quantum programs, e.g., we also allow quantum integers and not only qubits. As mentioned before, this is particularly interesting for expected runtime analysis, as there are programs which are almost surely terminating but have an infinite expected runtime.

\subsection*{Challenges}
Extending the previous work to unbounded operators is not straightforward and requires several technical challenges to be overcome. One of the most important differences between bounded and unbounded operators is that the latter may not be defined on the entire Hilbert space, i.e., one has to deal with states outside of the domain.
Properties like Hermitian or self-adjoint are equivalent for bounded operators and are used as observables / predicates. However, in the unbounded case, these properties do not coincide and one has to carefully choose the right operator class to ensure desired properties like convergence and expected values. The biggest technical challenge is to find a set of operators that is expressive enough, in the best case includes all bounded predicates and is $\omega$-complete together with a natural order. The last property is needed to express the weakest pre-expectation of a while-loop. This set of operators should be closed under the weakest pre-expectation of all program statements, which is not trivial to show, as one needs, e.g., a special kind of addition to ensure that self-adjointness is preserved. Additionally, those operators should induce a natural expected value which satisfies our requirements, i.e., its value should be in $\extR$.

\subsection*{Related Work}
Weakest preconditions for classical programs were first developed by Dijkstra \cite{Dijkstra76,Dijkstra75}, then extended to probabilistic programs by \cite{KOZEN1985162,McIverWpProb} and later to quantum programs by \cite{DHondtWeakestPreconditions}. As mentioned above, D'Hondt and Panangaden \cite{DHondtWeakestPreconditions} considered quantitative predicates bounded by $\identityOp$. 
Ying \cite{floydHoareLogic} extended this approach to capture partial correctness and gave an explicit representation of the predicate transformer for the quantum qWhile language.
Alternatively, a qualitative approach that allows simpler reasoning is given by \cite{ZhouAppliedQHL}. Several extensions of weakest preconditions like adding classical variables \cite{DENG202273, FengQHLClassicalVars}, non-determinism \cite{FengNondeterministicQuantumVerification} or conditional statements \cite{bayesianInf} have been considered.

Expected runtime analysis using weakest precondition style was developed first in the probabilistic setting \cite{probRuntime, ProbERTJournal} and later extended to quantum programs \cite{LiuRuntime, olmedoRuntime}. Liu et al. \cite{LiuRuntime} defined a real-valued function $ERT$ that determines the expected runtime of a program on input state $\rho$.
They also gave a weakest precondition style calculus $ert$ to compute $ERT(\rho)$ and showed that it coincides with the definition of $ERT$.
Their approach is limited to finite-dimensional Hilbert spaces and $ert$ is only defined for almost surely terminating programs. In contrast, our approach allows us to analyze expected runtimes in infinite-dimensional Hilbert spaces and for non-terminating programs, which can give rise to infinite expected runtimes. This was already mentioned in \cite{LiuRuntime} as an interesting topic for future research: "the infinite-dimensional case is certainly an interesting (and challenging, we believe) topic for future research".
Similarly, Olmedo and Díaz-Caro \cite{olmedoRuntime} define expected runtimes for quantum programs but consider only finite-dimensional Hilbert spaces. They do not define weakest preconditions in an observable way but rather use a more probabilistic approach where observables are functions from states to real numbers.

Another different way of reasoning about expected runtimes and more general expected values is given in \cite{expectationTransformer}. They define a quantum expectation transformer $qet$ for classical-quantum programs that does not depend on the denotational semantics. Their expectations do not operate on density operators but instead on the cost structure. However, their approach is also limited to qubits, i.e., to finite-dimensional Hilbert spaces.

Barthe et al. \cite{relationalHoareLogicTransport,dualityTheorem} define a relational Hoare logic for quantum programs. They also allow for infinite-valued predicates but only on finite-dimensional Hilbert spaces. Their approach is restricted to AST programs. They mention that the approach can be extended to the infinite-dimensional case but further details are not provided. At the same time as we finished this paper, \cite{LicsNewPaper} was published which captures the infinite-dimensional case. However, they still restrict to AST programs ("The restriction to AST programs applies to all logics in the paper") and they consider neither expected runtimes nor reward statements. The weakest precondition calculus is only stated for the bounded case and they decided to not use unbounded operators as predicates but instead relations on the state spaces.

\subsection*{Outline}
The rest of the paper is structured as follows. In Section \ref{sec:preliminaries}, we give some preliminaries on quantum computing and the mathematical background.
Section \ref{sec:expectations} introduces the notion of expectations, the expected value and their properties.
In Section \ref{sec:programs}, we define the syntax and semantics of quantum programs with rewards. In Section \ref{sec:wp}, we define the weakest pre-expectation semantics for quantum programs with rewards and give a syntactic weakest precondition calculus for our concrete qrWhile language. In Section \ref{sec:ertRewards}, we show how to analyze expected runtimes of quantum programs by transforming the program syntactically and using our weakest pre-expectation framework. Section \ref{sec:ertCalculus} continues with runtime analysis, here we give concrete ERT transformers. In Section \ref{sec:examples}, we give two example programs and analyze their expected runtimes and termination probabilities using the tools we developed. Finally, in Section \ref{sec:conclusion}, we summarize our contributions and discuss future work.

%% file: 2_preliminaries.tex
\section{Preliminaries}
\label{sec:preliminaries}
This section introduces the mathematical background and notation used in the paper. We advise the reader to skip this section on the first read and refer back to it when necessary.

We use $\mathbb{R}$ for the real numbers and $\mathbb{R}_{\geq 0}^\infty$ for the set of non-negative extended real numbers, i.e., including infinity.
$\mathbb{C}$ denotes the set of complex numbers and we write $\overline{x+yi}=x-yi$ for the complex conjugate of $x+yi \in \mathbb{C}$.
The infinite sum $\sum_{i \in I} u_i$ in a topological additive semigroup $V$ is $u \in V$ iff for every open neighborhood $U$ of $V$ there exists a finite $J \subseteq I$ such that $\sum_{i \in K} u_i \in U$ for every finite $K \subseteq I$ with $J \subseteq K$ \cite[p. 15/16]{conway1994}. Then we call $I$ \emph{summable}. (This includes $\R, \extR$, and all vector spaces we consider in this work.)

\subsection{Hilbert Spaces}
The \emph{inner product} over a vector space $\mathcal{V}$ is denoted by $\langle \cdot \mid \cdot \rangle$ and
the \emph{norm (or length) of a vector} $u \in \mathcal{V}$ is defined by $\Vert u \Vert:= \sqrt{\langle u \mid u \rangle}$.
When $\Vert u \Vert = 1$ holds, $u$ is a \emph{unit vector}. Vectors $u,v \in \mathcal{V}$ are \emph{orthogonal} (denoted $u \bot v$) if $\langle u \mid v \rangle = 0$.

If for any $\epsilon>0$, there exists a positive integer $N$ such that $\Vert u_n - u_m \Vert < \epsilon$ for all $n,m\geq N$ then the sequence $\{u_i\}_{i \in \mathbb{N}} \subseteq \mathcal{V}$ is a \emph{Cauchy sequence}.
The limit $\lim_{i \to \infty} u_i$ of such a sequence is $u$ if for any $\epsilon >0$, there exists a positive integer $N$ such that $\Vert u_n - u \Vert < \epsilon$ for all $n \geq N$.

A complete inner product space $\hilbert$ is a \emph{Hilbert space}, i.e., a vector space with inner product where every Cauchy sequence of vectors in $\hilbert$ has a limit \cite{floydHoareLogic}.

An orthonormal \emph{basis} of a Hilbert space $\hilbert$ is a (possibly infinite) family $\{u_i\}_{i \in I}$ of pairwise orthogonal unit vectors such that every $v\in \hilbert$ can be written as $v = \sum_{i \in I} \langle u_i \mid v\rangle \cdot u_i$. By $\abs{I}$ (the cardinality of $I$) we denote the dimension of $\hilbert$. We consider only separable Hilbert spaces, i.e., those with a countable (but potentially infinite) orthonormal basis. Hilbert spaces and their elements can be combined using the \emph{tensor product} $\otimes$ \cite[Def. IV.1.2]{takesaki1979theory}.

\begin{example}
    \label{example:HS}
    For a set $X$, the space
    \begin{align*}
        l^2(X) = \{\sum_{n\in X} \alpha_n \ket{n}\mid \alpha_n \in \mathbb{C} \text{ for all }n\in X\text{ and } \sum_{n\in X} \abs{\alpha_n}^2 < \infty\}
    \end{align*}
    is a Hilbert space with the inner product defined as $
        \langle\sum_{n\in X} \alpha_n \ket{n} \mid \sum_{n\in X}\alpha'_n \ket{n}\rangle := \sum_{n\in X} \overline{\alpha_n} \alpha'_n$.
    An orthonormal basis, the \emph{computational basis}, is $\{\ket{n}\mid n\in X\}$.
    The infinite-dimensional space $l^2(\mathbb{Z})$ can be used for quantum integers and is also denoted $\hilbert_\infty$.
    For qubits, we use the finite-dimensional space $l^2(\{0,1\})$ and denote it as $\hilbert_2$.
\end{example}

\subsection{Operators}
In the following, all vector spaces will be over $\mathbb{C}$.
A function $f:\mathcal{V}\to \mathcal{W}$ between vector spaces $\mathcal{V}$ and $\mathcal{W}$ is called \emph{linear} if $f(ax+y)=af(x)+f(y)$ for $x,y\in \mathcal{V}$ and $a\in \mathbb{C}$.

\subsubsection{Bounded Operators}
\label{sec:PrelimsBoundedOps}
Let $\mathcal{V}, \mathcal{W}$ be normed vector spaces. Then $f$ is called \emph{bounded linear} if $f$ is linear and $\norm{f(x)} \leq c \cdot \norm{x}$ for some constant $c \geq 0$ for all $x \in \mathcal{V}$.
If $\mathcal{V} = \mathcal{W}$ is a Hilbert space, we call $f$ a \emph{bounded operator}. We use $A\ket{\phi}$ to denote the result of applying operator $A$ to $\ket{\phi}\in \hilbert$.

The \emph{identity operator} $\identityOpVar{\hilbert}$ on $\hilbert$ is defined by $\identityOpVar{\hilbert} \ket{\phi}=\ket{\phi}$. The \emph{zero operator} on $\hilbert$, denoted by $\zeroOpVar{\hilbert}$, maps every vector to the zero vector. We omit $\hilbert$ if it is clear from the context.

A bounded operator $A$ on $\hilbert$ is called \emph{positive} if $\bra{\psi}A\ket{\psi}\geq 0$ for all $\ket{\psi}\in \hilbert$ and \emph{completely positive} if for every Hilbert space $\hilbert'$, the operator $\identityOpVar{\hilbert'}\otimes A$ is positive on $\hilbert' \otimes \hilbert$ \cite{Nielsen_Chuang_2010}.

Now we consider special kinds of bounded operators: A \textit{unitary operator} $U$ is an operator where $U^\dagger U = \identityOp$ and $U U^\dagger = \identityOp$.
$A$ is a \emph{trace class} operator if there exists an orthonormal basis $\{\ket{\psi_i}\}_{i \in I}$ such that $\{\bra{\psi_i} \cdot \abs{A} \cdot \ket{\psi_i} \}_{i \in I}$ is summable where $\abs{A}$ is the unique positive operator $B$ with $B^\dagger B = A^\dagger A$. The set of all trace class operators on $\hilbert$ is denoted by $\traceclass$. Then the trace of $A$ is defined as $tr(A) = \sum_{i \in I} \bra{\psi_i} \cdot A \cdot  \ket{\psi_i}$ where $\{\ket{\psi_i}\}_{i \in I}$ is an orthonormal basis.
It can be shown that $tr(A)$ is independent of the chosen basis \cite{floydHoareLogic}. The trace is cyclic, i.e., $tr(AB)=tr(BA)$ \cite{heisenbergdualityUnruh}, linear, i.e., $tr(A+B) = tr(A) + tr(B)$, scalar, i.e., $tr(cA) = c\cdot tr(A)$ for a constant $c$ \cite{conway2000courseOperator}, and multiplicative, i.e., $tr(A \otimes B)= tr(A)tr(B)$ \cite{heisenbergdualityUnruh}.
A positive trace class operator $\rho$ with $tr(\rho)\leq 1$ is called \textit{density operator} (sometimes also referred to as \textit{partial density operator}). The set of density operators is denoted by $\density$. Every trace class operator can be written as a linear combination of density operators \cite{isabelleproofTraceclassSum}. In the following we define semantics etc. on density operators but we implicitly mean the extensions of them to trace class operators by linearity.
Any density operator can be written as a weighted sum, i.e. $\rho = \sum_i p_i \ket{\psi_i}\bra{\psi_i}$ by the spectral decomposition (where the sum is usually meant w.r.t. (trace) norm topology but the topologies coincide in this case \cite[Lem. 30]{heisenbergdualityUnruh}). The decomposition is not unique in general.

To order bounded operators, the \textit{Loewner order} is used: for operators $A,B$, $A\loewner B$ holds iff $B-A$ is a positive operator. This is equivalent to $tr(A\rho)\leq tr(B \rho)$ for all density operators $\rho \in \density$ \cite{floydHoareLogic}. It is a partial order.
The Loewner order is compatible w.r.t. addition (also known as monotonic), i.e., $A \loewner B$ implies $A+C \loewner B+C$ for any $C$, and w.r.t. multiplication of non-negative scalars, i.e., $A\loewner B$ implies $c A \loewner c B$ for $c \geq 0 $ \cite{boyd2004convex}.
The Loewner partial order is an $\omega$-complete partial order ($\omega$-cpo) on the set of partial density operators \cite[Prop. 8.2.1]{YingPredicateTranserformerSemantics}, i.e., each increasing sequence of partial density operators has a least upper bound.

\subsubsection{Unbounded Operators}
An \emph{(unbounded) operator} $A:\domain{A} \to \hilbert$ is a linear function from a linear subspace $\domain{A}\subseteq \hilbert$ to Hilbert space $\hilbert$. We call $\domain{A}$ the domain of $A$. The range of an operator $A$ is defined as $ran(A) = \{A\ket{\psi} \mid \ket{\psi} \in \domain{A}\}$.

Let $A: \domain{A} \to \hilbert, B: \domain{B} \to \hilbert$. We can define the (operator) sum, the multiplication with a scalar and the product of $A$ and $B$ as follows \cite{kato}:
\begin{align*}
    (A+B)\ket{\psi} &= A\ket{\psi} + B\ket{\psi} \text{ for }\ket{\psi} \in \domain{A+B} = \domain{A} \cap \domain{B},\\
    (cA)\ket{\psi} &= c \cdot A\ket{\psi} \text{ for }\ket{\psi} \in \domain{cA} = \domain{A},\\
    (AB)\ket{\psi} &= A(B\ket{\psi}) \text{ for }\ket{\psi} \in \domain{AB} = \{\ket{\psi} \in \domain{B} \mid B\ket{\psi} \in \domain{A}\}.
\end{align*}

An operator $A: \domain{A} \to \hilbert$ is called \emph{densely defined} if $\domain{A}$ is dense in $\hilbert$, i.e., for every $\ket{\psi} \in \hilbert$ there exists a sequence $\{\ket{\psi_n}\}_{n\in \mathbb{N}}$ in $\domain{A}$ such that $\lim_{n\to \infty} \ket{\psi_n} = \ket{\psi}$. This is equivalent to say that the closure of $\domain{A}$ is equal to $\hilbert$ where
the \emph{closure} of a subset $X \subseteq \hilbert$ is the smallest closed set $\overline{X}$ with $X \subseteq \overline{X} \subseteq \hilbert$.

If $A: \domain{A} \to \hilbert, B: \domain{B} \to \hilbert$ are densely defined, then $A\otimes B: \domain{A \otimes B} \to \hilbert$ is dense in $\hilbert \otimes \hilbert$ where $\domain{A \otimes B}$ contains all finite linear combinations of vectors $\psi \otimes \phi$ with $\psi \in \domain{A}$ and $\phi \in \domain{B}$. It is $(A\otimes B)(\psi \otimes \phi) = A\psi \otimes B\phi$, extended by linearity. If $A$ and $B$ are closeable, then $A\otimes B$ is closable as well \cite[VIII. 10]{reed1980methods}.

Let $A: \domain{A} \to \hilbert$ be a densely defined operator. Then we set $\domain{A^\dagger} = \{\psi \in \hilbert \mid \exists u \in \hilbert: \braket{A\phi}{\psi} = \braket{\phi}{u} \text{ for }\phi \in \domain{A}\}$. As $\domain{A}$ is dense, $u$ is uniquely defined by $\psi$. Then we define the \emph{adjoint operator} $A^\dagger: \domain{A^\dagger} \to \hilbert$ by $A^\dagger \psi = u$. $A^\dagger$ is a linear operator \cite[Def. 1.6]{unboundedSelfAdjointBook}.

A densely defined operator $A$ is called \emph{symmetric} if $A \subseteq A^\dagger$, that means $\langle A\psi \mid \phi \rangle = \langle \psi \mid A \phi\rangle$ for all $\psi,\phi \in \domain{A}$ \cite[Def. 3.1]{unboundedSelfAdjointBook}. A symmetric operator $A$ is called \emph{positive} (written $A \geq 0$) if $\langle \psi \mid A \mid \psi \rangle \geq 0$ for all $\psi \in \domain{A}$ \cite[Def. 3.2]{unboundedSelfAdjointBook}.
A densely defined operator $A$ is called \emph{self-adjoint} if $A = A^\dagger$ \cite[Sec. 1.2]{unboundedSelfAdjointBook}. In the bounded case, self-adjointness and symmetry coincide, in the unbounded case self-adjointness implies symmetry.

According to the spectral theorem for unbounded self-adjoint operators \cite[Theorem 5.7]{unboundedSelfAdjointBook},
for every self-adjoint operator $A: \domain{A} \to \hilbert$ there exists a unique spectral measure $E$ on the Borel $\sigma$-algebra $B(\mathbb{R})$ such that $
    A = \int_{\mathbb{R}}\lambda d E_\lambda$.
By \cite[Theorem 10.4]{hall2013}, the spectral measure $E$ is a projection-valued measure, i.e. $E_\lambda$ are projections by definition of the spectral measure.

For a positive self-adjoint operator on $\hilbert$, there exists a unique positive self-adjoint operator $\half{A}$ on $\hilbert$ such that $\half{A} \half{A} = A$ \cite[Prop 5.13]{unboundedSelfAdjointBook} and $\domain{A} \subseteq \domain{\half{A}}$ \cite[Section 10.2]{unboundedSelfAdjointBook}. $\half{A}$ is called the \emph{square root} of $A$.

Similar as we introduced $\zeroOp$ and $\identityOp$ as special bounded operators, we use $\inftyOp$ to denote the unbounded operator defined by $\domain{\inftyOp} = \{0\}$ and $\inftyOp 0 = 0$.

\subsection{Sesquilinear Forms}
\label{sec:PrelimsForms}
Let $\hilbert$ be a Hilbert space. Then a \emph{sesquilinear form} (or \emph{form}) $a: \domain{a} \times \domain{a} \to \mathbb{C}$ with $\domain{a} \subseteq \hilbert$ is a function that is linear in the second and conjugate linear in the first argument \cite[Def. 10.1]{unboundedSelfAdjointBook}.

$a[u] = a[u,u]$ is called the \emph{quadratic form} associated with $a$. $a[u]$ determines $a[u,v]$ uniquely by the polarization identity \cite[Eq. 6.1.1]{kato} and \cite[Eq. 10.1]{unboundedSelfAdjointBook}.

For forms $a,b$ and $c \in \mathbb{C}$ the sum and multiplication with a scalar are defined as in \cite[p. 222]{unboundedSelfAdjointBook}:
\begin{align*}
    (a+b)[u,v] &= a[u,v] + b[u,v] \text{ for all }u,v \in \domain{a+b} =\domain{a} \cap \domain{b},\\
    (ca)[u,v] &= c \cdot a[u,v] \text{ for all }u,v \in \domain{ca} = \domain{a}.
\end{align*}

A form $a$ is called \emph{densely defined} if $\domain{a}$ is dense in $\hilbert$ \cite[Chapter 6]{kato}.

The adjoint $a^\dagger$ of a form $a$ is defined by $a^\dagger[u,v] = \overline{a[v,u]}$ for all $u,v \in \domain{a}$ \cite[Chapter 6]{kato}.
$a$ is called \emph{symmetric} if $a=a^\dagger$, i.e. $a[u,v] = a^\dagger[v,u]$ for all $u,v \in \domain{a}=\domain{a^\dagger}$ \cite[Chapter 6]{kato}.
A symmetric form $a$ is called \emph{lower semi-bounded} if $a[u] \geq c\cdot \norm{u}^2$ for all $u \in \domain{a}$ with $\norm{u} = 1$ \cite[Chapter 6]{kato} with $c\in \mathbb{R}$. If $c=0$, we call the form $a$ \emph{non-negative} or \emph{positive} \cite[Chapter 6]{kato}.

A form $a$ is called \emph{closed} if for every sequence $\{u_n\}_{n\in \mathbb{N}}$ in $\domain{a}$ with $u_n \to u$ and $\lim_{n,m \to \infty} a[u_n-u_m] = 0$, we have $u \in \domain{a}$ and $\lim_{n\to\infty} a[u_n-u] = 0$ \cite[Prop. 10.1]{unboundedSelfAdjointBook}.

For closed symmetric forms $a,b$ on the same Hilbert space, $a\leq b$ is defined as $\domain{b}\subseteq \domain{a}$ and $\forall x \in \domain{b}: a[x] \leq b[x]$ \cite[VIII, Chap. 3]{kato}.

We define a special case, the zero form $0$ with $0[u,v] = 0$ for all $u,v \in \domain{0} = \hilbert$.

\subsection{Quantum Specific Math}
The state space of a quantum system can be described as a Hilbert space where states correspond to unit vectors \cite{floydHoareLogic}. We consider only separable Hilbert spaces.
A quantum state is called \textit{pure} if it can be described by a vector in the Hilbert space; otherwise \textit{mixed}, i.e., it is a probability distribution over pure states.
We use density operators (see Section \ref{sec:PrelimsBoundedOps}) to describe (mixed) states, in particular to capture the current state of a program.
If a quantum system is in a pure state $\ket{\psi_i}$ with probability $p_i$ (with $\sum_i p_i \leq 1$), then this is represented by the density operator $\rho=\sum_i p_i \ket{\psi_i}\bra{\psi_i}$\footnote[1]{As in \cite{heisenbergdualityUnruh}, we mean convergence of sums with respect to SOT if not stated otherwise.\label{fn_1}}.

In quantum mechanics, a measurement can impact the current state of a quantum system.
Formally, a \emph{measurement} is a (possibly countably infinite) family of operators $\{M_m\}_{m\in I}$ where $m$ is the measurement outcome and $\sum_{m\in I} M_m^\dagger M_m = \identityOp$ \textsuperscript{\ref{fn_1}}.
If the current state of a quantum system is $\rho$ before the measurement $\{M_m\}$, then the probability for obtaining result $m$ is $p(m)=tr(M_m \rho M^\dagger_m)$ and the post-measurement state is $\rho_m = \frac{M_m\rho M^\dagger_m}{p(m)}$.

We call $f: \density \to \density$ a quantum subchannel if it is a (linear, bounded) completely positive and trace non-increasing map. Completely positive is defined according to \cite[Def. 34.2]{takesaki1979theory}.
Trace non-increasing means that $tr(f(\rho)) \leq tr(\rho)$ for all $\rho \in \density$.
According to \cite[Theorem 9]{krausChoi}, an equivalent definition (for separable Hilbert spaces) of a quantum subchannel is that it can be represented using Kraus operators, i.e., there exists a family of bounded operators $\{E_j\}$ such that $\sum_{j = 1}^\infty E_j^\dagger E_j \loewner \identityOp$ and $f(\rho) = \sum_{j = 1}^\infty E_j \rho E_j^\dagger$ for all $\rho \in \density$\footnote[2]{Convergence w.r.t. SOT (not WOT as in \cite{krausChoi}), see footnote 36 in \cite{heisenbergdualityUnruh}.}.

%% file: 3_expectations.tex
\section{Expectations}
\label{sec:expectations}
One of the important notions in deductive verification is the concept of a predicate or an expectation \cite{hoare69, McIverWpProb,DHondtWeakestPreconditions}. Those predicates are used to express properties that we are interested in, i.e., that we want to verify. Allowing a broader class of predicates allows us to express more properties. In this paper we want to express \emph{unbounded} expectation values, thus we need to change the definition from predicates to expectations. However, this leads to some technical challenges that we will address in this section.

A predicate is a $1$-bounded positive (hence self-adjoint) operator \cite{DHondtWeakestPreconditions}. Simply dropping the bound-condition does not work as the set of positive self-adjoint operators is not an $\omega$-cpo with respect to $\loewner$ (meaning the order for unbounded operators) which we need for defining the weakest precondition of the while-loop. This was already noted in \cite{DHondtWeakestPreconditions}: "We need to restrict to positive operators and -- in order to obtain least upper bounds for increasing sequences -- we need to bound them." To overcome this issue, we introduce a novel definition of expectations that enables to capture unbounded expectation values and show that the set of expectations forms an $\omega$-cpo with respect to $\loewner$.

There are equivalent definitions in the bounded case that describe the same class of operators (self-adjoint, symmetric) which are distinct for unbounded operators. We thus first need to figure out which of those definitions lead to the desired properties (e.g., we always want the expected value to be a positive real number or infinite).

Omitted proofs of this section can be found in Appendix \ref{sec:appendix_expectations}.

\subsection{Definitions}

A first naive attempt would be to require the unbounded operator to be self-adjoint. That implies a dense domain and we actually want to allow operators with non-dense domains as well. This is why we introduce the notion of $\infty$-self-adjointness. This is a weaker condition than self-adjointness but still ensures that the expected value is a real number or infinite:
\begin{definition}
    Let $\hilbert$ be a Hilbert space and $A: \domain{A} \to \hilbert$ an operator.
    \begin{itemize}
        \item $A$ is \emph{$\infty$-square} if $ran(A)\subseteq \overline{\domain{A}}$.
        \item If $A$ is $\infty$-square, its \emph{$\infty$-restriction} is the operator $A_\infty: \domain{A_\infty} \to \overline{\domain{A}}$ with $\domain{A_\infty} = \domain{A}$ and $A_\infty \psi = A\psi$ for all $\psi \in \domain{A_\infty}$.
        \item $A$ is called \emph{$\infty$-self-adjoint} if its $\infty$-restriction $A_\infty$ is self-adjoint.
    \end{itemize}
\end{definition}
Intuitively speaking, we require $\infty$-square and the $\infty$-restriction to properly define self-adjointness without requiring the operator to be dense. The $\infty$-restriction of an $\infty$-square operator is the part of the operator that is well-defined and thus can be used to define the expected value. $\infty$-self-adjointness ensures that the expected value is a real number or infinite. To define expectations, we add the condition that the operator is positive to ensure that the expected value is non-negative:
\begin{definition}
    The set of \emph{expectations} over a Hilbert space $\hilbert$ is the set of all positive, $\infty$-self-adjoint operators on $\hilbert$, denoted by $\predicate$.
\end{definition}

Expectations together with a suitable order need to form an $\omega$-cpo to handle while-loops correctly. To establish this, we define an order for expectations and show that $\predicate$ forms an $\omega$-cpo with respect to this order. The order is strongly inspired by \cite[Def. 10.5]{unboundedSelfAdjointBook} and is a naive extension of the Loewner order for bounded operators.

The Loewner order for bounded operators is defined by $A \loewner B$ if $\bra{\psi}A\ket{\psi} \leq \bra{\psi}B\ket{\psi}$ for all $\ket{\psi} \in \hilbert$. This definition is not applicable to unbounded operators as $\bra{\psi}A\ket{\psi}$ is not well-defined for all $\ket{\psi} \in \hilbert$ if $\ket{\psi} \notin \domain{A}$.

For unbounded self-adjoint operators, a constraint regarding the domain must be added. We adapt the order from \cite[Def. 10.5]{unboundedSelfAdjointBook} to order expectations:
\begin{definition} \label{def:expectationOrder}
    Let $A, B \in \predicate$ be expectations with self-adjoint $\infty$-restrictions $A_\infty, B_\infty$.
    Then $A \loewner B$ iff $\domain{\half{B_\infty}} \subseteq \domain{\half{A_\infty}}$ and $\norm{\half{A_\infty}\ket{\psi}}^2 \leq \norm{\half{B_\infty}\ket{\psi}}^2$ for all $\psi \in \domain{\half{B_\infty}}$.
\end{definition}
Note that $\half{A_\infty}$ means $\half{(A_\infty)}$ and not $(\half{A})_\infty$ and that the domain set-inclusion holds in the opposite direction.
In the bounded case, this order coincides with the standard Loewner order.

\subsection{Convergence of Expectations}
The aim is to show that the set of expectations over $\hilbert$ forms an $\omega$-cpo with respect to $\loewner$.
To that end, we introduce form-expectations, show their equivalence to expectations and then use Kato's theorem \cite[VIII, Thm. 3.13a]{kato} to show that the limit of an increasing sequence of form-expectations exists and is a form-expectation.
We start by defining form-expectations:
\begin{definition} \label{def:formExpectation}
    A \emph{form-expectation} is a closed, symmetric and positive form $a: \domain{a} \times \domain{a} \to \mathbb{C}$. The set of form-expectations over $\hilbert$ is denoted by $\formPredicate$.
\end{definition}

Every expectation is equivalent to a form-expectation:
\begin{proposition}
    \label{prop:OpstoForms}
    For every expectation $A\in \predicate$, there is a corresponding form-expectation $a \in \formPredicate$ with $a[x,y] = \langle \half{A_\infty}x, \half{A_\infty}y\rangle$ for all $x,y \in \domain{a} = \domain{\half{A_\infty}}$ and the other way round.
\end{proposition}
\begin{proof}[Proof Sketch]
    The main ingredient of the proof is the following theorem (\cite[VI, Thm. 2.1, 2.7]{kato} and \cite[Corr. 10.8, Eq. 10.12]{unboundedSelfAdjointBook}):
    "For every positive self-adjoint operator $A$, there is a corresponding densely defined, closed, symmetric, positive form $a: \domain{a} \times \domain{a} \to \mathbb{C}$ with $a[x,y] = \langle \half{A}x, \half{A}y\rangle$ for all $x,y \in \domain{a} = \domain{\half{A}}$ and the other way round".
\end{proof}

The next step is to show that the order of expectations as defined in Definition \ref{def:expectationOrder} corresponds to the order for form-expectations. Recall that forms are ordered by $a \leq b$ if $\domain{b} \subseteq \domain{a}$ and $a[x] \leq b[x]$ for all $x \in \domain{b}$ where $a[x] = a[x,x]$.
\begin{proposition}
    \label{prop:orderEquivalence}
    Let $A,B\in \predicate$ and $a,b$ be their corresponding form-expectations. Then $a \leq b$ if and only if $A \loewner B$.
\end{proposition}
\begin{proof} Follows immediately from the definitions:
    \begin{align*}
            a \leq b
             \Leftrightarrow & \text{ }\forall x \in \domain{b}: a[x] \leq b[x] \land \domain{b} \subseteq \domain{a} \\
             \Leftrightarrow & \text{ }\forall \psi \in \domain{\half{B_\infty}}: \norm{\half{A_\infty}\ket{\psi}}^2 \leq \norm{\half{B_\infty}\ket{\psi}}^2 \land \domain{\half{B_\infty}} \subseteq \domain{\half{A_\infty}}
             \Leftrightarrow  A \loewner B \qedhere
         \end{align*}
\end{proof}

We now establish that $\predicate$ constitutes an $\omega$-complete partial order w.r.t. $\loewner$:
\begin{theorem}
    \label{thm:omegaCompleteness}
    An increasing sequence of expectations $\{A_n\}_n$ with $A_n: \domain{A_n} \to \hilbert$ converges to an expectation $A: \domain{A} \to \hilbert$ which is also the least upper bound of the sequence, i.e., $A = \bigvee_{n} A_n = \lim_{n \to \infty} A_n$\footnote{Convergence in WOT.}. Thus $\predicate$ is an $\omega$-cpo w.r.t. $\loewner$.
\end{theorem}
\begin{proof}[Proof Sketch]
    The increasing sequence of expectations $\{A_n\}_n$ corresponds to an increasing sequence of form-expectations $\{a_n\}_n$ by Propositions \ref{prop:OpstoForms} and \ref{prop:orderEquivalence}. By Kato's theorem, the limit form $a$ of the sequence $\{a_n\}_n$ exists and is a form-expectation. Then we can construct the corresponding expectation $A$ from the limit form $a$ by Proposition \ref{prop:OpstoForms}. Finally, we show that $A$ is indeed the least upper bound of the sequence $\{A_n\}_n$.
\end{proof}
Thus, our definition of expectations provides the foundation for defining the semantics of while-loops in the presence of unbounded expectation values.

\subsection{Addition, Scalar Multiplication, Operator Application via Forms}
In this section we give some technical results about operations on expectations as, for example, the (operator) sum of two self-adjoint operators is not self-adjoint again, so we cannot define the sum of two expectations as usual.
Instead we define addition, scalar multiplication and bounded operator application of expectations via their form-expectations. Addition and scalar multiplication of forms are defined in Section \ref{sec:PrelimsForms} (\cite[p. 222]{unboundedSelfAdjointBook}) and application of bounded operator $M$ on form $a$ we define as
\begin{equation*}
a_M[u, v] = a[M u, M v] \text{ for all }u,v \in \domain{a_M} = \{u \mid M u \in \domain{a}\}.
\end{equation*}
Note that it suffices to consider the quadratic form $a_M[u] = a[M u, M u] = a[Mu]$.
We start by showing that form-expectations are closed under those operations:

\begin{proposition}
    \label{prop:FormAddMult}
    For $a,b \in \formPredicate,\alpha \geq 0$ and bounded operator $M$ is $a+b, \alpha \cdot a, a_M \in \formPredicate$.
\end{proposition}
\begin{proof}[Proof Sketch]
    For each operation, we show that the resulting form is closed, symmetric and positive, thus a form-expectation.
\end{proof}

Using these operations we can define addition, scalar multiplication and bounded operator application on expectations such that the resulting operators are still expectations:
\begin{definition}
    \label{def:FormAddMult}
    Let $A,B \in \predicate$ with corresponding form-expectations $a,b$. Let:
    \begin{itemize}
        \item $A \sumform B \in \predicate$ be the expectation that corresponds to form $a+b$,

        \item $\alpha \cdot A \in \predicate$ (with $\alpha \geq 0$) be the expectation that corresponds to the form $\alpha \cdot a$, and

        \item $\formsandwich{M}{A} \in \predicate$ (with bounded operator $M$ on $\hilbert$) be the expectation that corresponds to the form $a_M$.
    \end{itemize}
\end{definition}
We use $\sum A_n$ or similar for $A_n \in \predicate$ to refer to the sum $\sumform$ as defined in Definition \ref{def:FormAddMult}.
The fact that these operations are well-defined, i.e., that the resulting forms correspond to expectations, follows immediately from Propositions \ref{prop:OpstoForms} and \ref{prop:FormAddMult}.
This generalizes to infinite sums:
\begin{lemma} \label{lem:infSumExpect}
    Let $A_n \in \predicate$ for all $n$. Then $\sum_{n=0}^\infty A_n \in \predicate$\footnote{Convergence in WOT.}.
\end{lemma}
This also carries over to infinite sums of form-expectations, i.e., $\sum_{n=0}^\infty a_n = a \in \formPredicate$ for $a_n \in \formPredicate$.

We can also show continuity and distributivity of our newly defined operations:
\begin{proposition}
    \label{prop:ContinuityDistributivityFormSum}
    Let $A, A_n, B, B_n \in \predicate$ with $A_{n-1}\loewner A_n, B_{n-1} \loewner B_n$ and $M$ be a bounded operator on $\hilbert$. Then
    \begin{align}
        & \left(\bigvee_n A_n \right)\sumform B = \bigvee_n \left(A_n \sumform B\right) \\
        & B \sumform \left(\bigvee_n A_n \right) = \bigvee_n \left(B \sumform A_n \right) \\
        & \left(\bigvee_n A_n \right)\sumform \left( \bigvee_n B_n \right) = \bigvee_n \left(A_n \sumform B_n \right) \\
        & \formsandwich{M}{\left(\bigvee_n A_n \right)} = \bigvee_n \left(\formsandwich{M}{A_n} \right)\\
        & \formsandwich{M}{(A \sumform B)} = \formsandwich{M}{A} \sumform \formsandwich{M}{B}
    \end{align}
\end{proposition}
Note that for an increasing sequence $A_n$ of expectations, the supremum $\bigvee_n A_n$ is the limit of the sequence and an expectation according to Theorem \ref{thm:omegaCompleteness}.

We also have that distributivity applies to infinite sums:
\begin{lemma}\label{lem:distrInfSum}
    Let $ A_n\in \predicate$. Then $
        \formsandwich{M}{(\sum_{n=0}^\infty A_n)} = \sum_{n=0}^\infty \formsandwich{M}{A_n}$.
\end{lemma}

All definitions above coincide with the usual notions for bounded expectations:
\begin{proposition}
    \label{prop:BoundedAddition}
    Addition, scalar multiplication and application of bounded operators as defined in Definition \ref{def:FormAddMult} coincide with the usual operations for bounded expectations.
\end{proposition}

Now we show some technical properties of the order $\loewner$ that we defined on expectations (Definition \ref{def:expectationOrder}). In particular, we show that it is a partial order and that it is preserved under addition and the application of bounded operators.
\begin{proposition}
    \label{prop:orderproperties}
    Let $A,B \in \predicate$.
    \begin{itemize}
        \item $\loewner$ is a partial order, i.e., antisymmetric, reflexive and transitive.
        \item $A \loewner B \Rightarrow A\sumform C \loewner B \sumform C$ for any $C \in \predicate$.
        \item $A \loewner B \Rightarrow \formsandwich{M}{A} \loewner \formsandwich{M}{B}$ for any bounded operator $M$ on $\hilbert$.
    \end{itemize}
\end{proposition}

\subsection{Expectation Value}
So far, we considered expectation values in an abstract way, i.e., we only required them to be positive, real and possibly infinite.
Now we give a concrete definition of the expectation value for a given state based on the spectral measure of the corresponding operator.
We adapt the definition from \cite{ExpectationValuesKraus}:
\begin{definition} \label{def:expectationValue}
    Let $A \in \predicate$ and $A_\infty$ be its $\infty$-restriction. We define the \emph{expectation value} $\mathbb{E}_A: \density \to \mathbb{R}_{\geq 0}^\infty$ for $\rho \in \density$ as
    \begin{equation*}
        \expect{A}{\rho} = \begin{cases}
            \int \lambda d (tr(E_\lambda \rho)) & \text{if } \int \abs{\lambda} d (tr(E_\lambda \rho)) < \infty \\
            \infty & \text{otherwise}
            \end{cases}
    \end{equation*}
    where $E_\lambda$ is the spectral measure of $A_\infty$.
\end{definition}

For (bounded) expectation $A \loewner \identityOp$, the integral always exists and the expectation value is equivalent to $tr(A\rho)$ \cite[Remark 2]{ExpectationValuesKraus}, the satisfaction of a predicate as in \cite{DHondtWeakestPreconditions}.

As this is a quite technical definition, we present two equivalent definitions of $\mathbb{E}_A$ which are more intuitive and easier to use in proofs (e.g., in the proof of Lemma \ref{lem:orderExpectation}). The first one is based on $\half{A}$ and aligns with the definition of the order $\loewner$; the second one is based on the corresponding form-expectation.
\begin{lemma}
    \label{lem:expectAlternativeDef}
    Let $A \in \predicate$ with form $a$. Let $\rho = \sum_i p_i \ket{\psi_i}\bra{\psi_i}\in\density$ be a decomposition with $p_i > 0$. Then $\mathbb{E}_A(\rho) = \sum_i p_i \expect{A}{\ket{\psi_i}\bra{\psi_i}}$ where either
    \begin{enumerate}
        \item \begin{equation*}
        \expect{A}{\ket{\psi}\bra{\psi}} = \begin{cases}
            \norm{\half{A_\infty} \ket{\psi}}^2 & \text{if } \ket{\psi} \in \domain{\half{A_\infty}} \\
            \infty & \text{otherwise}
            \end{cases} \text{ or }
        \end{equation*}
        \item \begin{equation*}
            \expect{A}{\ket{\psi}\bra{\psi}} = \begin{cases}
                a[\ket{\psi}] & \text{if } \ket{\psi} \in \domain{a} \\
                \infty & \text{otherwise.}
            \end{cases}
        \end{equation*}
    \end{enumerate}
\end{lemma}

Note that whenever $\ket{\psi} \in \domain{A}$, then $\expect{A}{\ket{\psi}\bra{\psi}} = \bra{\psi} A \ket{\psi}$ \cite[Lemma 1.11]{ExpectationValuesKraus}.

The Loewner order on 1-bounded expectations can be characterized using expected values, i.e., $A \loewner B$ if $tr(A\rho) \leq tr(B\rho)$ for all $\rho$ \cite[Prop. 2.5]{DHondtWeakestPreconditions}. This also holds for our definition of expectation values:
\begin{lemma}
    \label{lem:orderExpectation}
     For $A , B \in \predicate$:
     \begin{align*}
         A \loewner B \Leftrightarrow \text{ for all } \rho \in \density: \expect{A}{\rho} \leq \expect{B}{\rho}
    \end{align*}
\end{lemma}
The correspondence to expected values is a healthiness condition indicating that our order $\loewner$ is indeed the right order to use for expectations.

Using this result, it follows that $\zeroOp$ and $\inftyOp$ are the least and greatest element of $\predicate$ with respect to $\loewner$ as $\expect{\zeroOp}{\rho} = 0 \leq \expect{A}{\rho} \leq \expect{\inftyOp}{\rho} = \infty$ for all $A \in \predicate, \rho \in \density$ ($\rho \neq 0$, otherwise $\expect{\inftyOp}{\rho} = 0 = \expect{A}{\rho}$).

Now that we have established expectation values we can show that they are linear and continuous in both the expectation and the state:
\begin{proposition}
    \label{prop:expectationProperties}
    \text{  }
    \begin{itemize}
        \item Expectation values are linear (in expectation and state), i.e., $\expect{A\sumform B}{\rho} = \expect{A}{\rho} + \expect{B}{\rho}$ and $\expect{A}{\rho + \sigma} = \expect{A}{\rho} + \expect{A}{\sigma}$
         for $A, B \in \predicate$ and $\rho, \sigma \in \density$ such that $\rho+\sigma \in \density$.
        \item Expectation values are scalar multiplicative (in expectation and state), i.e.,  $\expect{cA}{\rho} = c \cdot \expect{A}{\rho}$ and $\expect{A}{c\rho} = c \cdot \expect{A}{\rho}$ for $c \geq 0, \rho, \sigma \in \density$ such that $c \rho, \rho+\sigma \in \density$.
        \item Expectation values are continuous (in expectation and state), i.e., $\expect{\bigvee_i A_i}{\rho} = \bigvee_i \expect{A_i}{\rho}$ and $\expect{A}{\bigvee_i \rho_i} = \bigvee_i \expect{A}{\rho_i}$ for an increasing chain of expectations $\{A_i\}_i \subseteq \predicate$ and states $\{\rho_i\}_i \subseteq \density$ respectively.
    \end{itemize}
\end{proposition}

Linearity can also be extended to infinite sums analogously to before:
\begin{lemma} \label{lem:infSumLinExpectation}
    Let $A_n \in \predicate$ and $\rho_n \in \density$, for all $n$. Then $\expect{\sum_n A_n}{\rho} = \sum_n \expect{A_n}{\rho}$ for all $\rho \in \density$ and $\expect{A}{\sum_n \rho_n} = \sum_n \expect{A}{\rho_n}$ for all $A \in \predicate$ and $\rho_n, \sum_n \rho_n \in \density$.
\end{lemma}

We conclude with a result that will be useful in the following sections:
\begin{proposition}
    \label{prop:boundedOpinExpectation}
    Let $A \in \predicate$ and $M$ a bounded operator on $\hilbert$. Then $\expect{\formsandwich{M}{A}}{\rho} = \expect{A}{M \rho M^\dagger}$ for all $\rho \in \density$.
\end{proposition}

%% file: 4_programs.tex
\section{A Quantum Programming Language with Rewards}
\label{sec:programs}

Our aim is to extend existing quantum while languages to include rewards to enable reasoning about costs such as runtime. In this section, we first present some desirable properties about the semantics of such a programming language. Then we define the programming language qrWhile and show that it satisfies these properties. Omitted proofs can be found in Appendix \ref{sec:appendix_pl}.

\subsection{Quantum Programs with Rewards}
\label{sec:programGeneral}

We consider a separable Hilbert space $\hilbert$ as the state space of our quantum programs and positive (real) rewards only. Our aim is to define a semantics for quantum programs $S$ as two separate maps $\semanticsD{S}: \density \to \density$ and $\semanticsR{S}: \density \to \extR$, where $\density$ is the set of density operators on $\hilbert$.
The map $\semanticsD{S}$ describes how the quantum state changes during the execution of the program (the usual denotational semantics), while $\semanticsR{S}$ describes the expected reward collected during the program execution. The use of two separate maps is similar to \cite{bayesianInf} and allows us to reason about the expected reward without modifying the quantum state. If we would integrate the expected reward as a variable into the quantum state, we have to be careful about diverging loops similar to the first example from the introduction (these problems occur in the probabilistic setting as well \cite[Fig. 5]{philippHigherMoments}). Desirable properties of these two maps are:
\begin{definition}\label{def:generalSemantics}
    A semantics for quantum programs with rewards is a pair of maps $\semanticsD{S}: \density \to \density$ and $\semanticsR{S}: \density \to \extR$ such that
    \begin{itemize}
        \item $\semanticsD{S}$ is linear, trace non-increasing and completely positive
        \item if $S$ is reward-free, then $\semanticsR{S}(\rho) = 0$ for all $\rho \in \density$
        \item $\semanticsR{S}$ is linear
        \item the form $a$ defined by $a[\ket{\psi}] = \semanticsR{S}(\ket{\psi}\bra{\psi})$ with $\domain{a} =\{\ket{\psi} \mid \semanticsR{S}(\ket{\psi}\bra{\psi}) < \infty\}$ is a positive, symmetric, closed form.
    \end{itemize}
    Note that we implicitly mean the extensions to trace class operators $\traceclass$ and $\C$.
    The set of all quantum programs over $\hilbert$ with such a semantics is denoted by $\programs$.
\end{definition}
Requiring $\semanticsD{S}$ to be linear, trace non-increasing and completely positive are the common properties that we expect from the denotational semantics of quantum programs / a quantum channel (see, e.g., \cite{floydHoareLogic} and \cite{bayesianInf}). This also implies that $\semanticsD{S}$ can be represented as a Kraus decomposition \cite[Theorem 9]{krausChoi}.
The second constraint $\semanticsR{S}(\rho) = 0$ for reward-free $S$ asserts that the expected reward is zero if the program does not contain any reward statements, and is quite natural. Similarly, the third constraint, that $\semanticsR{S}$ is linear, is also natural as the expected reward should be linear in the input state. The last constraint probably needs more motivation. For finite programs (e.g., without loops), the expected reward $\semanticsR{S}^{fin}: \density \to \extR$ can be computed as $\semanticsR{S}^{fin}(\rho) = tr(A_R \rho)$ which is a positive functional where $A_R$ corresponds to the observable that keeps track of the earned reward\footnote{As mentioned earlier, we implicitly mean the extension to $\traceclass \to \C$.}.
The expected reward for infinite programs (e.g., with loops) can then be defined as a least upper bound of the expected rewards of all its finite approximations (i.e., abort the program after a finite number of steps).
Such a least upper bound corresponds to a positive, symmetric and closed form, thus we require this property for $\semanticsR{S}$:
\begin{proposition} \label{prop:supCPmaps}
    Let $g_i: \density \to \mathbb{R}_{\geq 0}$ be a (pointwise) increasing sequence of maps. Then $g: \density \to \extR$ defined by $g(\rho) = \sup_i g_i(\rho)$ corresponds to a positive, symmetric and closed form.
\end{proposition}

\subsection{QrWhile - A Concrete Syntax and Semantics}
\label{sec:qrWhile}
In this section we present a concrete programming language, called qrWhile, and show that it satisfies the properties defined in the previous section (Definition \ref{def:generalSemantics}). The syntax and semantics of qrWhile is based on \cite{floydHoareLogic} but we extend it by reward statements inspired by \cite{philippHigherMoments, expectationTransformer, lenaPOPL2023}.

We assume a countable number of variables that are either qubits or quantum integers. Qubits are represented by the 2-dimensional Hilbert space $\hilbert_2$, quantum integers by infinite-dimensional Hilbert spaces $\hilbert_\infty$ as defined in Example \ref{example:HS}.
The state space $\mathcal{H}$ of a quantum program is the tensor product of the Hilbert spaces of all its variables.

Additionally to the usual program statements as in \cite{floydHoareLogic} we introduce a reward statement that collects rewards during the execution of the program and allows us to reason about expected rewards and expected runtime, similar to the consume statement in \cite{expectationTransformer}.
Note that in comparison to \cite{expectationTransformer}, we do not have a separate (classical) variable to store the reward and we also allow infinite values. We allow rewards from a countable set $R \subseteq \extR$.

The $\rewardR$ statement does not change the quantum state (as it does not correspond to any operation on the quantum computer), but adds $c \cdot p$ to the expected reward where $p$ is the probability of the current execution.

\begin{definition}
    The syntax of qrWhile programs $S$ over the Hilbert space $\hilbert$ is defined as follows:
    \[S ::= \skipbf \mid \qzero\mid \Uq \mid \rewardR \mid \measure \mid \concat \mid \while \]
    where
    \begin{itemize}
        \item $\skipbf$ is the skip statement
        \item $\qzero$ initializes the quantum variable $q$ to the basis state $\ket{0}$
        \item $\Uq$ applies the unitary operator $U$ to the quantum variables $\overline{q}$
        \item $\rewardR$ adds a reward $c \in R$
        \item $\measure$ performs a measurement $M = \{M_m\}$ on the quantum variables $\overline{q}$ and executes the subprogram $S_m \in \overline{S}$ according to the measurement outcome $m$
        \item $\concat$ is the sequential composition of programs $S_1$ and $S_2$
        \item $\while$ is a while-loop that measures the quantum variables $\overline{q}$ with the measurement $M = \{M_0, M_1\}$ and executes the body $S$ if the outcome is $1$, otherwise it terminates.
    \end{itemize}
\end{definition}

The denotational semantics of qrWhile is defined as follows. We need to handle the rewards separately, yielding pairs (similar as in \cite{bayesianInf}). The first element is the density operator (as in \cite[Prop. 5.1]{floydHoareLogic}, \cite{bayesianInf}), the second element is the expected reward collected during the program execution.
We denote $\semantics{S}(\rho) = (\semanticsD{S}(\rho), \semanticsR{S}(\rho))$:
\begin{definition}
    \label{def:semantics}
    The \emph{denotational semantics} of a qrWhile program $S$ is given by a map $\semantics{S}: \density \to \density \times \extR$:
    \begin{itemize}
        \item $\semantics{\skipbf}(\rho) = (\rho,0)$
        \item $\semantics{\qzero}(\rho) = (\sum_{n=0}^\infty \ket{0}\bra{n} \rho \ket{n}\bra{0}, 0)$
        \item $\semantics{\Uq}(\rho) = (U \rho U^\dagger,0)$
        \item $\semantics{\rewardR}(\rho) = (\rho, c \cdot tr(\rho))$
        \item $\semantics{S_1; S_2}(\rho) = (\semanticsD{S_2}(\semanticsD{S_1}(\rho)), \semanticsR{S_2}(\semanticsD{S_1}(\rho)) + \semanticsR{S_1}(\rho))$
        \item $\semantics{\measure}(\rho) = \sum_m \semantics{S_m}(M_m \rho M_m^\dagger)$
        \item $\semantics{\while}(\rho) =$\\
        $(\bigvee_{i=0}^\infty \semanticsD{\while^i}(\rho), \lim_{i \to \infty} \semanticsR{\while^i}(\rho))$ where \newpage
        \begin{align*}
        \while^0 &= \Omega \text{ with } \semantics{\Omega}(\rho) = (\zeroOp,0)\\
        \while^{i+1} &= \measure \\
        & \text{ with } S_0 = \skipbf\\
        & \text{ and } S_1 = S; \while^i
        \end{align*}
    \end{itemize}
\end{definition}
where addition is entry-wise.
Well-definedness follows from \cite{floydHoareLogic} and the fact that $\semanticsR{\while^i}(\rho)$ is increasing.

When considering programs without reward statements, the semantics is the same as in \cite{floydHoareLogic}. This follows immediately from the definition.
\begin{proposition} \label{prop:rewardFreeOriginalSem}
    For all qrWhile programs $S$ without reward statements and all $\rho \in \density$, it holds that $\semantics{S}(\rho) = (\semanticsOrig{S}(\rho), 0)$ where $\semanticsOrig{S}$ is the semantics defined in \cite{floydHoareLogic}.
\end{proposition}

The following healthiness properties hold:
\begin{proposition}
    \label{prop:semanticsProperties} Let $S$ be a qrWhile program. Then:
    \begin{enumerate}
        \item $\semanticsD{S}$ is linear
        \item $\semanticsD{S}$ is trace non-increasing, i.e., $tr(\semanticsD{S}(\rho)) \leq tr(\rho)$ for all $\rho \in \density$
        \item $\semanticsD{S}$ is completely positive
        \item $\semanticsR{S}$ is linear
        \item $a[\ket{\psi}] := \semanticsR{S}(\ket{\psi} \bra{\psi})$ for all $\ket{\psi} \in \domain{a} = \{\ket{\psi} \in \hilbert \mid \semanticsR{S}(\ket{\psi} \bra{\psi}) < \infty \}$ is a symmetric, positive and closed form.
    \end{enumerate}
\end{proposition}

Thus our concrete semantics defined in Definition \ref{def:semantics} satisfies the requirements of Definition \ref{def:generalSemantics} and thus it is a semantics for quantum programs with rewards.
\begin{theorem}
    The semantics defined in Definition \ref{def:semantics} satisfies the properties in Definition \ref{def:generalSemantics} and thus every qrWhile program is an element of $\programs$.
\end{theorem}
\begin{proof}
    Follows immediately from Propositions \ref{prop:rewardFreeOriginalSem} and \ref{prop:semanticsProperties}.
\end{proof}

Another choice for the reward statement would have been to allow expectations (or other real valued measurements) in the reward statement, but they induce an implicit measurement which changes the quantum state and thus we only consider constants. Note that this is not a restriction as we could measure and then add the reward based on the measurement outcome (using the measurement statement) and thus we can also add rewards that depend on the program state.

%% file: 5_wp.tex
\section{Weakest Pre-Expectations}
\label{sec:wp}

In this section we first define Hoare triples and the weakest pre-expectation in terms of general quantum programs and give a concrete syntactic way of computing them for qrWhile programs. Omitted proofs can be found in Appendix \ref{sec:appendix_wp}.

\subsection{Correctness}
The correctness of Hoare triples with respect to a program $S\in \programs$ and expectations is defined as follows:
\begin{definition} Let $S \in \programs$ and $A,B \in \predicate$.
    A Hoare triple $\{A\} S \{B\}$ is
    \begin{itemize}
        \item \emph{totally correct} if for all $\rho \in \density$:
        \begin{equation*}
        \expect{A}{\rho} \leq \expect{B}{\semanticsD{S}(\rho)} + \semanticsR{S}(\rho)
        \end{equation*}
        \item \emph{partially correct} if $A,B \loewner \identityOp$, $S$ is reward-free and for all $\rho \in \density $:
        \begin{equation*}
            \expect{A}{\rho} \leq \expect{B}{\semanticsD{S}(\rho)} + \underbrace{[tr(\rho) - tr(\semanticsD{S}(\rho))]}_{\substack{\text{probability of non-termination}}}
        \end{equation*}
    \end{itemize}
\end{definition}
\noindent Note that the last definition (partial correctness) requires $1-$bounded expectations $A,B \loewner \identityOp$ as otherwise adding the term $[tr(\rho) - tr(\semanticsD{S}(\rho))]$ does not make sense. Similarly, for partial correctness we only consider reward-free programs (and not add the term $\semanticsR{S}(\rho)$). Otherwise we add probabilities and expected values (possibly $>1$). Thus, partial correctness coincides with the definition in \cite{floydHoareLogic} and is therefore not further considered in this paper.

We obtain similar results as \cite[Prop. 6.1]{floydHoareLogic} but often the restriction to reward-free programs is necessary:
\begin{proposition}\label{prop:hoareProperties} For $S \in \programs$ and $A,B \in \predicate$:
    \begin{enumerate}
    \item If a Hoare triple is totally correct, $S$ reward-free and $A,B \loewner \identityOp$, it is also partially correct.
    \item The Hoare triple $\{\zeroOp\} S \{A\}$ is totally correct.
    \item For all reward-free $S$ and $A \loewner \identityOp$, the Hoare triple $\{A\} S \{ \identityOp\}$ is partially correct.
\end{enumerate}
\end{proposition}

\subsection{Existence of WP} \label{sec:wpExistence}
As our aim is to keep the wp-calculus as close as possible to \cite{DHondtWeakestPreconditions}, we define the weakest pre-expectation as follows:
\begin{definition} \label{def:generalwp}
    Let $S \in \programs$.
    The \emph{weakest pre-expectation} with respect to a post-expectation $B \in \predicate$ is the expectation $\qwp{S}{B} = \sup\{A \mid \{A\} S \{B\}\}$ (if it exists).
\end{definition}
The existence of the weakest pre-expectation is shown in Theorems \ref{thm:existenceFormsWp} and \ref{thm:existenceWp} below. An alternative condition to check for weakest pre-expectations is:
\begin{lemma}
    \label{lem:ExpectationValue}
    Let $B, C \in \predicate$ and $S \in \programs$.
    If $\text{ }\expect{C}{\rho} = \expect{B}{\semanticsD{S}(\rho)} + \semanticsR{S}(\rho)$ holds for all $\rho\in \density$, then $C$ is the weakest pre-expectation.
\end{lemma}
\begin{proof}
    $\expect{C}{\rho} = \expect{B}{\semanticsD{S}(\rho)} + \semanticsR{S}(\rho)$ implies that $C$ is a pre-expectation of $B$ and if there was another pre-expectation $A$ with $C \loewner A$, it would hold that $\expect{A}{\rho} > \expect{C}{\rho} = \expect{B}{\semanticsD{S}(\rho)} + \semanticsR{S}(\rho)$ for some $\rho$ which contradicts that $A$ is a pre-expectation.
\end{proof}
In the case of bounded operators and reward-free programs, these definitions coincide with \cite{DHondtWeakestPreconditions,floydHoareLogic}.

We can also directly state some properties of weakest pre-expectations that hold whenever the premise in Lemma \ref{lem:ExpectationValue} holds. Note that the first two properties only hold for reward-free programs.
\begin{proposition}
    \label{prop:wpProperties}
    Let $S \in \programs$ and $B \in \predicate$ such that $\expect{\qwp{S}{B}}{\rho} = \expect{B}{\semanticsD{S}(\rho)} + \semanticsR{S}(\rho)$.
    Then
    \begin{itemize}
        \item $\qwp{S}{\zeroOpVar{\hilbert}} = \zeroOpVar{\hilbert}$ for reward-free programs $S$
        \item Linearity: For reward-free programs, $c \geq 0$ with $cB \in \predicate$ it is $\qwp{S}{c B} = c \cdot \qwp{S}{B}$ and for $B_1, B_2 \in \predicate$ it is $\qwp{S}{B_1 \sumform B_2} = \qwp{S}{B_1} \sumform \qwp{S}{B_2}$
        \item Monotonicity: If $B_1, B_2 \in \predicate$ with $B_1 \loewner B_2$, then $\qwp{S}{B_1} \loewner \qwp{S}{B_2}$
        \item Order-continuity: For an increasing chain of expectations $\{B_n\}_n \subseteq \predicate$ is $\qwp{S}{\bigvee_n B_n} = \bigvee_n \qwp{S}{B_n}$
    \end{itemize}
\end{proposition}

Now we can prove the existence of the weakest pre-expectation (the supremum) for all quantum programs in $\programs$ based on forms. It can be seen as a kind of Schrödinger-Heisenberg duality for unbounded operators, as we show the expectation value equivalence in the proof. In fact, this proposition not only holds for our qrWhile semantics but for all quantum programs in $\programs$, i.e., all quantum programs with semantics that satisfy the properties of Definition \ref{def:generalSemantics}.
\begin{theorem}
    \label{thm:existenceFormsWp}
    Let $S \in \programs$ and $b\in \formPredicate$ with corresponding expectation $B \in \predicate$.
    Then:
    \begin{align*}
        \qwpForm{S}{b} &:= \lambda \psi. b[\semanticsD{S}(\ket{\psi} \bra{\psi})] + \semanticsR{S}(\ket{\psi}\bra{\psi}) \in \formPredicate
    \end{align*}
    and $\qwp{S}{B}$ (as defined in Definition \ref{def:generalwp}) is its corresponding expectation.
    Note that we abuse notation and write $b[\rho]$ to mean $\sum_i p_i b[\ket{\phi_i}]$ for $\rho = \sum_i p_i \ket{\phi_i}\bra{\phi_i}$.
\end{theorem}
As mentioned above, this result for forms gives us also the existence of the weakest pre-expectation in the operator form:
\begin{theorem} \label{thm:existenceWp}
    Let $S\in \programs$ and $B \in \predicate$. Then $\qwp{S}{B} \in \predicate$ exists with $\expect{B}{\semanticsD{S}(\rho)} + \semanticsR{S}(\rho)= \expect{\qwp{S}{B}}{\rho}$ for all $\rho \in \density$.
\end{theorem}
\begin{proof}
    Follows from Theorem \ref{thm:existenceFormsWp} and Lemma \ref{lem:ExpectationValue}.
\end{proof}

\subsection{WP for qrWhile Programs}

For qrWhile programs, we give a concrete representation in Table \ref{tab:wp} which shows how to compute the weakest pre-expectation for a given program and post-expectation. The correctness of this representation is given by:
\begin{proposition}
    \label{prop:wpDefinition}
    Let $S$ be a qrWhile program.
    The transformer $wp \llbracket S \rrbracket$ as given in Table \ref{tab:wp} is well-defined, i.e., $\qwp{S}{B} \in \predicate$ for post-expectation $B \in \predicate$ and $\expect{\qwp{S}{B}}{\rho} = \expect{B}{\semanticsD{S}(\rho)} + \semanticsR{S}(\rho)$ for all $\rho \in \density$. That is, $\qwp{S}{B}$ is the weakest pre-expectation of $S$ with respect to $B$.
\end{proposition}
\begin{proof}[Proof Sketch]
    We show by simultaneous induction on the structure of $S$:
    \begin{enumerate}
        \item $\qwp{S}{B}$ is an expectation, and
        \item $\expect{\qwp{S}{B}}{\rho} = \expect{B}{\semanticsD{S}(\rho)} + \semanticsR{S}(\rho)$ for all $\rho \in \density$.
    \end{enumerate}
\end{proof}

\begin{table}
        \centering
        \renewcommand{\arraystretch}{1.3}
        \setlength{\tabcolsep}{8pt}
        \begin{tabular}{@{}ll@{}}
        \toprule
        \text{Program $S$} & \text{$\qwp{S}{B}$} \\
        \midrule
        $\skipbf$ & $B$ \\
        $\qzero$ & $\displaystyle \sum_{n=0}^\infty \ket{n}\bra{0} \odot B \odot \ket{0}\bra{n}$ \\
        $\Uq$ & $\formsandwich{U}{B}$ \\
        $\rewardR$ & $B \sumform c \cdot \identityOp$ \\
        $\concat$ & $\qwp{S_1}{\qwp{S_2}{B}}$ \\
        $\measurePrime$ & $\displaystyle \sum_{m} \formsandwich{M_m}{\qwp{S_m'}{B}} $ \\
        $\while'$ & $\displaystyle \bigvee_n A_n \quad \text{where }
            A_0 = \zeroOp,\; A_{n+1} = \formsandwich{M_0}{B} \sumform \formsandwich{M_1}{\qwp{S'}{A_n}}$ \\
        \bottomrule\\
    \end{tabular}
        \caption{Weakest pre-expectation transformer $wp \llbracket S \rrbracket:\predicate \to \predicate$ for a qrWhile program $S$ and post-expectation $B\in \predicate$.}
        \label{tab:wp}
    \end{table}

Finally, if we remove the reward from the qrWhile programs and consider only $1-$bounded expectations, the calculus presented above coincides with the standard one \cite{floydHoareLogic}:
\begin{proposition} \label{prop:wpBoundedEquiv}
    If $A \loewner \identityOp$ is $1-$bounded and qrWhile program $S$ is reward-free, then our weakest precondition calculus and definitions are equivalent to \cite{floydHoareLogic,DHondtWeakestPreconditions} as in this case, $\expect{A}{\rho} = tr(A\rho)$ for all partial density operators $\rho \in \density$.
\end{proposition}

\subsection{Park Induction}
The reasoning about loops is typically facilitated by proof rules. These rules provide sufficient conditions under which bounds for qrWhile programs can be proven. Those rules also hold in the bounded case \cite{floydHoareLogic}.
First, we define the characteristic function of a loop:
\begin{definition}
    The \emph{characteristic function} $\psi_B: \predicate \to \predicate$ of $\while$ w.r.t. post-expectation $B \in \predicate$ is defined for $A \in \predicate$ by:
    \begin{align*}
        \psi_B(A) = \formsandwich{M_0}{B} \sumform \formsandwich{M_1}{\qwp{S}{A}}
    \end{align*}
\end{definition}

Based on this, we obtain the quantum version of Park induction \cite{park}, which is a well-known proof rule for upper bounds of wp in the probabilistic case:
\begin{proposition} \label{prop:park}
    Let $\while$ be a loop in qrWhile, $A,B \in \predicate$ and $\psi_B$ its characteristic function then
    \begin{equation*}
        \psi_B(A)\loewner A \Rightarrow \qwp{\while}{B} \loewner A
    \end{equation*}
\end{proposition}

Analogously we obtain a dual version for lower bounds of wlp as defined in \cite{floydHoareLogic}:
\begin{proposition} \label{prop:parkWLP}
    Let $\while$ be a loop in qrWhile, $A,B \in \predicate$ with $A,B \loewner \identityOp$ and $\phi_B(A) = \formsandwich{M_0}{B} + \formsandwich{M_1}{\qwlp{S}{A}}$, then
    \begin{equation*}
        A \loewner \phi_B(A) \Rightarrow A \loewner \qwlp{\while}{B}
    \end{equation*}
\end{proposition}

%% file: 6_ert.tex
\section{Expected Runtime using Rewards}
\label{sec:ertRewards}
In this section we show how to use the weakest pre-expectation to obtain the expected runtime -- the number of executed statements as in \cite{LiuRuntime} -- of quantum programs. The main idea is to use expected rewards. Thus this section only handles programs without rewards, as we want to use the reward statements to track the runtime.

\begin{remark}
There are many different ways to track the expected runtime in a program (e.g., we could also add a tick variable that is incremented in each statement or define a forward transformer that simply counts the number of statements as done in \cite{LiuRuntime}), but we choose this transformation using reward statements to make our approach flexible: it is easy to change the notion of runtime, e.g., count the number of iterations in a while-loop instead (then only the program transformation must be changed). More reasons for using reward statements (in the probabilistic case) can be found in \cite{philippHigherMoments}, they also give examples why using counter variables do not work as expected.
\end{remark}

\cite{LiuRuntime} presents an expected runtime transformer for quantum programs with bounded operators only (and finite dimensional spaces), the extension to unbounded operators is one instance of our weakest pre-expectation framework (see Theorems \ref{thm:equivErtQwp} and \ref{thm:ertEquivERT}).
To do so, we transform (reward-free) programs syntactically by adding $\reward{1}$ statements as shown in Table \ref{tab:programTransformation}. Intuitively, we add a $\reward{1}$ statement in front of each statement to count the number of executed statements. For the while-loop, we add a reward statement for each measurement, that means we do it once in the beginning and then after the loop body to account for the next measurement.
Omitted proofs can be found in Appendix \ref{sec:appendix_ert_rewards}.

\begin{table}
        \centering
        \renewcommand{\arraystretch}{1.3}
        \setlength{\tabcolsep}{8pt}
        \begin{tabular}{@{}ll@{}}
        \toprule
        \text{Program $S$} & \text{Transformed Program $\transform{S}$} \\
        \midrule
        $\skipbf$ & $\skipbf$ \\
        $\qzero$ & $\reward{1}; \qzero $ \\
        $\Uq$ & $\reward{1}; \Uq $ \\
        $\concat$ & $\transform{S_1}; \transform{S_2}$ \\
        $\measure$ & $\reward{1}; \textbf{measure } M[\overline{q}]:\overline{\transform{S}}$ \\
        $\while$ & $\reward{1}; \textbf{while } M[\overline{q}]=1 \textbf{ do } \{\transform{S}; \reward{1}\}$\\
        \bottomrule\\
    \end{tabular}
        \caption{The syntactic transformation of a reward-free program $S$ to a program $\transform{S}$ that allows us to track the expected runtime of $S$.}
        \label{tab:programTransformation}
\end{table}

\begin{definition} \label{def:transformedProgram}
    For a reward-free qrWhile program $S$, the program $\transform{S}$ is inductively defined in Table \ref{tab:programTransformation}.
\end{definition}

The added reward statements do not change the quantum state, i.e., $\transform{S}$ is semantically equivalent to $S$ in terms of the quantum state. They enable to track the expected reward which will be used to define the expected runtime of $S$.
\begin{lemma} \label{lem:rewardTransformQuantumState}
    Let $S$ be a reward-free qrWhile program. Then $\semanticsD{\transform{S}}(\rho) =\semanticsD{S}(\rho)$.
\end{lemma}
\begin{proof}
    Follows immediately from the definition of $\semanticsD{\reward{r}}$ and we only add finite many reward statements to each statement of $S$, i.e., we do not create non-termination by adding those reward statements.
\end{proof}

Finally we define the expected runtime of a reward-free program $S$ using the expected reward of the transformed program $\transform{S}$:
\begin{definition} \label{def:ert}
    Let $S$ be a reward-free qrWhile program. Then the \emph{expected runtime} of $S$ on input state $\rho$ is defined as $\semanticsR{\transform{S}}(\rho)$.
\end{definition}

One of the main results of this section is, that we can now characterize the expected runtime using weakest preconditions. This brings the usual advantages of weakest preconditions, e.g., we can determine the expected runtime for all input states at once.
\begin{theorem}\label{thm:equivWpSemanticsErt}
    The expected runtime of a reward-free qrWhile program $S$ on input state $\rho$ is given by $\expect{\qwp{\transform{S}}{\zeroOp}}{\rho}$.
\end{theorem}
\begin{proof}
    It is $\expect{\zeroOp}{\semanticsD{\transform{S}}(\rho)} = 0$, thus \[
        \expect{\qwp{\transform{S}}{\zeroOp}}{\rho}
        = \expect{\zeroOp}{\semanticsD{\transform{S}}(\rho)} + \semanticsR{\transform{S}}(\rho)
        = \semanticsR{\transform{S}}(\rho). \qedhere
    \]
\end{proof}

In runtime analysis, we are often interested in determining whether a program terminates with probability 1. This property is called \emph{almost sure termination} \cite{AST78} and can be defined for quantum programs as in \cite{LiuRuntime}:
\begin{definition} \label{def:ast}
    The set of states on which a program $S$ is \emph{almost surely terminating} (AST) is $\{\rho \in \density \mid tr(\semanticsD{S}(\rho)) = tr(\rho)\}$.
    If $S$ is AST on all input states, we say that $S$ is AST.
\end{definition}

One main difference to \cite{LiuRuntime} is, that we consider infinite-dimensional Hilbert spaces. Thus, a program can be AST but has an infinite expected runtime. As in the finite-dimensional case, AST and finite expected runtime coincide \cite[Thm. 2]{LiuRuntime}, the restriction to AST in \cite{LiuRuntime} has the limitation to be unable to deal with infinite runtimes.
In presence of possible infinite expected runtimes, we consider \emph{positive almost sure termination} analogous to \cite[Def. 5]{hardnessQuantitative}:
\begin{definition} \label{def:past}
    The set of states on which a reward-free program $S$ is \emph{positive almost surely terminating} (PAST) is defined as $\{\rho \in \density \mid \semanticsR{\transform{S}}(\rho) < \infty\}$.
    If $S$ is PAST on all input states, we say that $S$ is PAST.
\end{definition}
We can also give alternative characterizations for reward-free programs of AST and PAST using the weakest pre-expectation similar to \cite{probRuntime}:
\begin{proposition}\label{prop:AstPastUsingWp}
    Let $S$ be a reward-free qrWhile program. Then
    \begin{itemize}
        \item $S$ is AST if and only if $\qwp{S}{\identityOp} = \identityOp$.
        \item $S$ is PAST if and only if $\domain{\qwp{\transform{S}}{\zeroOp}} = \hilbert$.
    \end{itemize}
\end{proposition}

These two characterizations are especially interesting in the quantum case. We know that the space of states for which a program is AST is a closed subspace of the Hilbert space as it can be characterized by the 1-Eigenspace of the operator $\qwp{S}{\identityOp}$. In contrast, the space of states for which a program is PAST is a subspace but not necessarily closed (as it can be characterized by the domain of the bounded operator $\half{\qwp{\transform{S}}{\zeroOp}}_\infty$). An example of a non-closed subspace is given in Section \ref{sec:examples}.

\section{ERT Calculus}
\label{sec:ertCalculus}
This section generalizes the transformers in \cite{LiuRuntime} by lifting the restriction to finite dimensional Hilbert spaces, AST programs and bounded operators.
First, we give a forward transformer $ERT\llbracket S \rrbracket$ of the expected runtime. Then we give an alternative backward transformer of the expected runtime $ert\llbracket S \rrbracket$ that maps an operator $A$ to an operator that gives the expected runtime plus the "post" $A$ after termination. Finally, we show that this backward definition is equivalent to the forward definition and also to the weakest pre-expectation of the transformed program $\transform{S}$ defined in Table \ref{tab:programTransformation}.
Proofs can be found in Appendix \ref{sec:appendix_ert_calculus}.

First, we adapt Definition 3.1 from \cite{LiuRuntime} to our infinite-dimensional setting:
\begin{lemma}
    \label{lem:ERTForeward}
    The forward expected runtime of reward-free qrWhile program $S$ on input $\rho \in \density$ is given by $\ERT{S}{\rho}$ where $\ERT{S}{\cdot}: \density \to \extR$ is defined by:
    \begin{itemize}
        \item $\ERT{\skipbf}{\rho} = 0$
        \item $\ERT{\qzero}{\rho} = tr(\rho)$
        \item $\ERT{\Uq}{\rho} = tr(\rho)$
        \item $\ERT{\concat}{\rho} = \ERT{S_1}{\rho} + \ERT{S_2}{\semanticsD{S_1}(\rho)}$
        \item $\ERT{\measure}{\rho} = tr(\rho) + \sum_{m} \ERT{S_m}{M_m \rho M_m^\dagger}$
        \item $\ERT{\while}{\rho} = \lim_{n \to \infty} \ERT{while^{[n]}}{\rho}$
        where
        \begin{align*}
            while^{[0]} &= \skipbf \\
            while^{[n+1]} &= \measure \text{ with } S_0 = \skipbf \text{ and } S_1 = S;while^{[n]}
        \end{align*}
    \end{itemize}
\end{lemma}

Based on this forward transformer, we can also give a backward definition of the expected runtime using the same idea as used in weakest preconditions analogous to probabilistic expected runtime transformers \cite{probRuntime}\footnote{Note that in comparison to the probabilistic ert calculus \cite{probRuntime}, we do not count the skip statement as a step, as it does not correspond to any operation on the quantum computer and we want to stay comparable to \cite{LiuRuntime}.}. The transformer $ert\llbracket S \rrbracket$ as defined in Table \ref{tab:ert} maps an operator $A$ to an operator that gives the expected runtime plus the "post" $A$ after termination, in this case we set $A=\zeroOp$ initially.

\begin{table}
        \centering
        \renewcommand{\arraystretch}{1.3}
        \setlength{\tabcolsep}{8pt}
        \begin{tabular}{@{}ll@{}}
        \toprule
        \text{Program $S$} & \text{$\ert{S}{B}$} \\
        \midrule
        $\skipbf$ & $B$ \\
        $\qzero$ & $\displaystyle \identityOp \sumform \sum_{n=0}^\infty \ket{n}\bra{0} \odot B \odot \ket{0}\bra{n} $ \\
        $\Uq$ & $\identityOp\sumform \formsandwich{U}{B} $ \\
        $\concat$ & $\ert{S_1}{\ert{S_2}{B}}$ \\
        $\measurePrime$ & $\displaystyle \identityOp \sumform \left(\sum_{m} \formsandwich{M_m}{\ert{S_m'}{B}} \right)$ \\
        $\while'$ & $\displaystyle \bigvee_n F_n \quad \text{where }$\\
        & $F_0 = \zeroOp,\; F_{n+1} = \formsandwich{M_0}{B} \sumform \formsandwich{M_1}{\ert{S'}{F_n}} \sumform \identityOp$ \\
        \bottomrule\\
    \end{tabular}
        \caption{The backward \emph{expected runtime transformer} $ert\llbracket S \rrbracket: \predicate \to \predicate$ for a reward-free qrWhile program $S$. }
        \label{tab:ert}
    \end{table}

Now we show healthiness properties of the expected runtime transformer:
\begin{proposition}
    \label{prop:ERTProperties}
    Let $S$ be a reward-free qrWhile program, $A,B \in \predicate$ and $c\geq 0$.
    The expected runtime transformer $ert$ satisfies:
    \begin{enumerate}
        \item Well-definedness: $\ert{S}{A} \in \predicate$
        \item Connection to wp: $\ert{S}{A} = \qwp{S}{A} \sumform \ert{S}{\zeroOp}$
        \item Monotonicity: $A \loewner B$ implies $\ert{S}{A} \loewner \ert{S}{B}$
        \item Continuity: $\ert{S}{\bigvee_i A_i} = \bigvee_i \ert{S}{A_i}$ for any increasing chain $\{A_i\}_i$ of expectations.
        \item Constant propagation: $\ert{S}{c\cdot A} = c\cdot A \sumform \ert{S}{\zeroOp}$
        \item Preservation of infinity: $\ert{S}{\inftyOp} = \inftyOp$
    \end{enumerate}
\end{proposition}

Now we can show that the expected runtime transformer actually computes the expected runtime of a program as defined in Definition \ref{def:ert} (i.e., that it is equivalent to the weakest pre-expectation of the transformed program $\transform{S}$  as in Theorem \ref{thm:equivWpSemanticsErt}):
\begin{theorem}
    \label{thm:equivErtQwp}
    Let $ert\llbracket S \rrbracket: \predicate \to \predicate$ be defined as in Table \ref{tab:ert} and $A \in \predicate$. Then $\ert{S}{A}= \qwp{\transform{S}}{A}$ for any reward-free qrWhile program $S$.
\end{theorem}

We can also show the consistency between Table \ref{tab:ert} and Lemma \ref{lem:ERTForeward} (analogous to \cite[Thm. 1]{LiuRuntime}):
\begin{theorem} \label{thm:ertEquivERT}
    Let $S$ be a reward-free qrWhile program and $\rho \in \density$. Then $\expect{\ert{S}{\zeroOp}}{\rho} = \ERT{S}{\rho}$.
\end{theorem}

This result together with \cite[Thm. 1]{LiuRuntime} shows that our operator-based definition of the expected runtime transformer coincides with the operator-based definition \cite[Def. 4.1]{LiuRuntime} for finite-dimensional, reward-free AST programs.

Using the equivalences above, we can show that whenever the expected runtime of a program is finite (i.e., the program is PAST), it terminates with probability 1 and thus is AST as well:
\begin{lemma} \label{lem:PastImpliesAst}
    Let $S$ be a reward-free qrWhile program and $\rho \in \density$. If $\text{ }S$ is PAST on $\rho$, then $S$ is AST on $\rho$.
\end{lemma}
The other direction does not hold, as there are programs that are AST on a state but not PAST on the same state (e.g., the examples presented in Section \ref{sec:examples}).

\paragraph*{Park Induction for Expected Runtime}
Similar as in \cite{probRuntime} for the probabilistic case, we can use the idea of Park Induction (Proposition \ref{prop:park}) for bounds of the expected runtime.
First we define the characteristic functional for the while-loop $\while$:
\begin{equation*}
    \Psi_B (A) = \identityOp \sumform \formsandwich{M_0}{B} \sumform \formsandwich{M_1}{\ert{S}{A}}
\end{equation*}

Based on this, we can consider Park induction for expected runtime. Intuitively, it means that if we can find an operator $A$ such that $\Psi_B(A) \loewner A$ holds, then $A$ is an upper bound on the expected runtime of the while-loop. This can be used to give upper bounds on the expected runtime without having to compute the exact least upper bound, which is often more difficult.
\begin{proposition} \label{prop:ertPark}
    Let $\while$ be a reward-free while-loop in qrWhile and $A, B \in \predicate$. If \text{ }$\Psi_B(A) \loewner A$, then $\ert{\while}{B} \loewner A$.
\end{proposition}

%% file: 7_examples.tex
\section{Examples}
\label{sec:examples}

In this section we present two example programs and analyze their expected runtime and termination behavior. The first example is a program with infinite expected runtime but that terminates with probability 1 which consists of two while-loops executed after each other. We will also analyze the runtime of both while-loops separately.

The second example is a variant of the infinite quantum walk where the walker can only move in one direction or stay at its position. We will analyze the expected runtime of this program and show that it is partially non-AST, partially AST but not PAST and partially PAST, i.e., there are input states where the probability of termination is smaller than $1$, there are other states that terminate with probability $1$ but have infinite expected runtime and there are states that terminate with probability $1$ and have finite expected runtime.

Omitted calculations can be found in Appendix \ref{sec:appendix_examples}.
Note that in this section we write $+$ instead of $\sumform$ and leave $\odot$ out for better readability.

\subsection{Probabilistic PAST vs. Quantum PAST}
As mentioned earlier, there are programs in the infinite-dimensional setting with infinite expected runtime but that terminate with probability 1. In this section, we present an example of such a program and analyze its expected runtime and termination behavior.

\subsubsection{The Quantum Program}
We adapt the probabilistic program from \cite{probRuntime} to the quantum setting. It is used to illustrate that probabilistic PAST is not preserved under sequential composition.
In the quantum analog we have two variables, one is a qubit $q_2$ that we use to simulate the coin flip and the other one is a quantum integer $q_1$ that we use as integer variable (allowing only non-negative values).
The quantum program can be found in Figure \ref{fig:PASTcompositional}.
\begin{figure}
    \begin{align*}
    &  q_1:= \ket{1}; \\
    &  q_2:= \ket{1}; \\
    &\colorbox{rwthblue!30}{\parbox{3.5cm}{$\begin{aligned} \text{while} (q_2 = \ket{1}) \{ \\
    & \hspace*{-1.8cm} q_2:= H q_2; \\
    & \hspace*{-1.8cm} q_1 := U_{2} q_1; \\
    & \hspace*{-2.2cm}\}\end{aligned}$}} \hspace*{-0.5cm}\textcolor{white}{L_1}\\
    & \colorbox{rwthblue!30}{\parbox{3.5cm}{$\begin{aligned} \text{while} (q_1 \neq \ket{0}) \{ \\
    & \hspace*{-1.8cm} q_1 := U_{-1} q_1; \\
    & \hspace*{-2.2cm}\}\end{aligned}$}} \hspace*{-0.5cm}\textcolor{white}{L_2}\\
    \end{align*}

    \caption{A quantum program where $U_2$ maps $\ket{n}$ to $\ket{2n}$ and $U_{-1}$ maps $\ket{n}$ to $\ket{n-1}$ (to make them unitary, one can for example add ancilla qubits which we omit to not complicate the analysis). The measurement operators are $M_0 = \ket{0}\bra{0}$, $M_1 = \ket{1}\bra{1}$ for the first loop on variable $q_2$ and $M_0 = \ket{0}\bra{0}$, $M_1 = \identityOp - \ket{0}\bra{0}$ for the second loop on variable $q_1$.}\label{fig:PASTcompositional}
\end{figure}

\subsubsection{Park Induction for the Loops}
First we consider the sub programs and analyze their expected runtime. Computing the concrete runtime of a while-loop is quite complex as it involves computing the least upper bound of a sequence, so we use Park induction.

We start with the first while-loop $L_1$. The characteristic functional is given by
\begin{align*}
    \Psi_\zeroOp(A) =&\text{ } \identityOp + M_0^\dagger \zeroOp M_0 + M_1^\dagger \ert{S}{A} M_1 \\
    =&\text{ } \identityOp + 2 \identityOp \otimes \ket{1}\bra{1} + (\identityOp \otimes \ket{1}\bra{1})(U_2^\dagger \otimes H^\dagger) A (U_2 \otimes H) (\identityOp \otimes\ket{1}\bra{1})
\end{align*}

For the loop invariant $A= \identityOp + 6(\identityOp \otimes \ket{1}\bra{1})$, we have
\begin{align*}
    \Psi_\zeroOp(A) = &\text{ } \identityOp + 2 \identityOp \otimes \ket{1}\bra{1} + (\identityOp \otimes \ket{1}\bra{1})(U_2^\dagger \otimes H^\dagger) (\identityOp + 6(\identityOp \otimes \ket{1}\bra{1})) (U_2 \otimes H) (\identityOp \otimes\ket{1}\bra{1}) \\
     =&\text{ } \identityOp + 6(\identityOp \otimes \ket{1}\bra{1}) = A
\end{align*}
So we have $\Psi_\zeroOp(A) \loewner A$ and thus by Park induction $\ert{L_1}{\zeroOp} \sqsubseteq A$. That means for every initial state with $\ket{1}$ on the second variable, the expected runtime of the first while-loop is at most 7, with $\ket{0}$ on the second variable it is at most $1$. If we assume the initialization to be part of this program, we are in the first case. In any case, the expected runtime is finite, i.e., the subprogram is PAST.

Similarly, we can analyze the second while-loop $L_2$ where we only need to consider one variable. The characteristic functional is given by
\begin{align*}
    \Psi_\zeroOp(A) =&\text{ } \identityOp + M_0^\dagger \zeroOp M_0 + M_1^\dagger \ert{S}{A} M_1 = \identityOp + (\identityOp - \ket{0}\bra{0}) + (\identityOp - \ket{0}\bra{0}) U_{-1}^\dagger A U_{-1} (\identityOp - \ket{0}\bra{0})
\end{align*}
Using the loop invariant $A = \sum_{k=0}^{\infty} (2k+1) \ket{k}\bra{k}$, we have
\begin{align*}
    \Psi_\zeroOp(A) =&\text{ } \identityOp + (\identityOp - \ket{0}\bra{0}) + (\identityOp - \ket{0}\bra{0}) U_{-1}^\dagger \left(\sum_{k=0}^{\infty} (2k+1) \ket{k}\bra{k}\right) U_{-1} (\identityOp - \ket{0}\bra{0})\\
    =& \sum_{k=0}^{\infty} \ket{k}\bra{k} + \sum_{k=1}^{\infty} \ket{k}\bra{k} + \sum_{k=1}^{\infty} (2k - 1) \ket{k}\bra{k}
    = \sum_{k=0}^{\infty} (2k+1) \ket{k}\bra{k} = A
\end{align*}
By Park induction $\ert{L_2}{\zeroOp} \sqsubseteq A$.
This operator is unbounded in general, that means the program is not PAST (or at least we cannot show it with this invariant).
In the probabilistic setting we get a finite runtime because for every (classical) basis state $\ket{n}$, we have $\expect{A}{\ket{n}\bra{n}} = 2n+1$ which is a finite upper bound on the runtime.
As in this algorithm, $q_1$ will always have a classical value (something like $\ket{n}$), we can conclude that the expected runtime of the second while-loop is finite for every reachable state.
The difference is that in the quantum setting, the expected runtime can still be infinite for certain superpositions, e.g., for a state $\sum_{n=1}^{\infty} \frac{\alpha}{n^2} \ket{n}\bra{n}$ (with suitable $\alpha$), the expected runtime could be infinite as shown in the next example.

\subsubsection{Infinite Expected Runtime}
We now show that the whole program has infinite expected runtime. Park induction would only give us an upper bound on the expected runtime, so we have to compute the expected runtime of the whole program directly to show that it is infinite.
To start the analysis, we determine $\ert{L_2}{\zeroOp} = \bigvee_n P_n$. Intuitively, $L_2$ decrements the quantum integer $q_1$ until it reaches $\ket{0}$, i.e. its runtime is the current value of $q_1$. As shown in Appendix \ref{sec:appendix_example1}, we have
\begin{equation*}
    \ert{L_2}{\zeroOp} = \bigvee_{n=0}^\infty P_n = \lim_{n\to\infty} \left((n+1)\identityOp - \sum_{k=0}^{n-1} (n-k)\ket{k}\bra{k}\right):= W_1
\end{equation*}
To compute this limit, we can look at the expectation value on a basis state $\ket{m}$:
\begin{align*}
    \expect{\ert{L_2}{\zeroOp}}{\ket{m}\bra{m}} &= \lim_{n\to\infty} \left((n+1) - (n-m)\right) = \lim_{n\to\infty} (m+1) = m+1
\end{align*}
This coincides with the result one would expect intuitively, as the runtime of the loop is exactly the current value of $q_1$ plus one for the final check. Note that so far we only considered the variable $q_1$, the other variable $q_2$ did not matter for this loop, so we have $\identityOp$ on the other variable.

Now we can analyze the first while-loop $L_1$. As shown in Appendix \ref{sec:appendix_example1}, we have
\begin{align*}
    &\ert{L_1}{W_1 \otimes \identityOp} = \bigvee_{n=0}^\infty Q_n \\
    =& \lim_{n\to\infty} \left( \left(\left(6 - \frac{8}{2^n} \right) \identityOp + \sum_{k=1}^{n-1} (\frac{1}{2} U_2^\dagger)^k W_1 U_2^k \right)\otimes \ket{1}\bra{1} + \identityOp +  W_1 \otimes (\identityOp - \ket{1}\bra{1}) \right) := W_2
\end{align*}

Before this loop, the program only sets $q_1$ and $q_2$ to $\ket{1}$. To get the expected runtime of the entire program, we determine $\ert{q_1:= \ket{1}; q_2:= \ket{1}}{W_2}$ or, equivalently, the expected runtime of the program afterwards with fixed input $\ket{1}$ for both variables which is simpler in this case. That means the expected runtime of the full program is given by $W_2$ evaluated on the state $\ket{1}\otimes \ket{1}$. To compute this, we can look at the expectation value:
\begin{align*}
    &\expect{\ert{L_1}{\zeroOp}}{\ket{1}\bra{1} \otimes \ket{1}\bra{1}} = \lim_{n\to\infty} 7 - \frac{8}{2^n} + n-2 + \frac{1}{2^{n-1}}  = \infty.
\end{align*}
Overall, the expected runtime of this program is infinite. If we would have other initial values of $q_1$ and $q_2$, computing the expected runtime would influence only the last equation in comparison to \cite[Example. 3.3.]{LiuRuntime} where the whole calculation would need to be redone.

\subsubsection{Almost Sure Termination}
To complete our analysis, we show the program to be AST by using the fact that $\qwp{S}{\identityOp} = \identityOp$ implies AST. We compute $\qwp{L_2}{\identityOp} = \bigvee_n P_n$ first only for $q_1$ with
\begin{align*}
    P_0 &= \zeroOp \\
    P_{n+1} &= M_0^\dagger \identityOp M_0 + M_1^\dagger \qwp{S}{P_n} M_1 \\
    &= \ket{0}\bra{0} + (\identityOp - \ket{0}\bra{0}) \qwp{q_1 := U_{-1} q_1}{P_n} (\identityOp - \ket{0}\bra{0}) \\
    &= \ket{0}\bra{0} + (\identityOp - \ket{0}\bra{0}) U_{-1}^\dagger P_n U_{-1} (\identityOp - \ket{0}\bra{0})
\end{align*}
Then we obtain $P_n = \ket{0}\bra{0} + \ket{1}\bra{1} + \dots + \ket{n-1}\bra{n-1}$ and overall we have $\qwp{L_2}{\identityOp} = \bigvee_n P_n = \identityOp$.

For the first while loop, we have $\qwp{L_1}{\identityOp \otimes \identityOp} = \bigvee_n P_n$ with
\begin{align*}
    P_0 &= \zeroOp \\
    P_{n+1} &= M_0^\dagger \identityOp M_0 + M_1^\dagger \qwp{S}{P_n} M_1 \\
    &= \identityOp \otimes \ket{0}\bra{0} + (\identityOp \otimes \ket{1}\bra{1} ) \qwp{q := H q; q_1 := U_2 q_1}{P_n} (\identityOp \otimes \ket{1}\bra{1}) \\
    &= \ket{0}\bra{0} \otimes \identityOp + (\identityOp \otimes \ket{1}\bra{1}) ( U_2^\dagger \otimes H^\dagger) P_n (U_2 \otimes H)  (\identityOp \otimes \ket{1}\bra{1})
\end{align*} with the closed form being
    $P_n = \identityOp \otimes \left(\ket{0}\bra{0} + \left(1 - \frac{1}{2^n}\right) \ket{1}\bra{1}\right)$
and overall $\qwp{L_1}{\identityOp} = \bigvee_n P_n = \identityOp$.
Together with the first two initializations, $\qwp{S}{\identityOp} = \identityOp$.

\subsection{One-Sided Quantum Walk}
In this section we present a variant of the quantum walk on an one-side bounded infinite line that, in comparison to the standard version, does not spread in both direction but only in one direction or stays at its position. The program consists of two variables, one is a quantum integer $q$ (that can also take negative values now) that we use to simulate the position of the walker and the other one is a qubit $c$ that we use to simulate the coin flip. We define the (controlled) shift operator $S$ as follows: $S\ket{n}\ket{0} = \ket{n-1}\ket{0}, S\ket{n}\ket{1}  =\ket{n}\ket{1}$. Then we can define the program $P_1$ as follows:
\begin{align*}
    & c:= \ket{0}; \\
    & \textbf{while } M[q] \textbf{ do }\{ \\
    & \hspace*{1cm} c := H c; \\
    & \hspace*{1cm} q c := S q c; \\
    & \}
\end{align*}
with $M_0 = \ket{0}\bra{0}$, $M_1 = \identityOp - \ket{0}\bra{0}$. In each iteration, we apply a Hadamard gate to the coin to create an equal superposition of $\ket{0}$ and $\ket{1}$ and then apply the shift operator to move left or do nothing according to the coin.

Intuitively, if the program starts in a state $\ket{n}$ with $n<0$, it will never terminate as the walker will always move left and never reach $\ket{0}$. This also holds for superpositions or mixed states that have a non-zero probability of being in a state $\ket{n}$ with $n<0$. We can show that the program is non-AST for those states and AST for all other states by calculating $\qwp{P_1}{\identityOp}$. The characteristic function is given by
\begin{align*}
    \psi_\identityOp(A) =&\text{ } M_0^\dagger \identityOp M_0 + M_1^\dagger (\identityOp \otimes H)S^\dagger A S (\identityOp \otimes H) M_1 \\
    =&\ket{0}\bra{0} \otimes \identityOp + ((\identityOp - \ket{0}\bra{0})\otimes H)S^\dagger A S  ((\identityOp - \ket{0}\bra{0}) \otimes H)
\end{align*}
We can show that $A = (\zeroOp_{<0} + \identityOp_{\geq 0}) \otimes \identityOp = \identityOp_{\geq 0} \otimes \identityOp$ is a fixed point of $\psi_\identityOp$ (see Appendix \ref{sec:appendix_example2}).
Thus $\qwp{P_1}{\identityOp} \sqsubseteq (\zeroOp_{<0} + \identityOp_{\geq 0}) \otimes \identityOp$ and thus it cannot be AST for states $\ket{n}$ with $n<0$. For states $\ket{n}$ with $n\geq 0$, we have just shown $\identityOp$ to be an upper bound, thus we cannot follow directly that the program is AST for those states. However, we show now that the expected runtime is finite for states $\ket{n}$ with $n\geq 0$ thus the program is AST for those states and automatically for all superpositions of those as well.
To do so, we use the ert-calculus with Park induction. The characteristic functional is given by
\begin{align*}
    \Psi_\zeroOp(A) =& \text{ }\identityOp + M_0^\dagger \zeroOp M_0 + M_1^\dagger \ert{S}{A} M_1 \\
    =& \text{ }3 \identityOp - 2 \ket{0}\bra{0} \otimes \identityOp + ((\identityOp - \ket{0}\bra{0})\otimes H)S^\dagger A S  ((\identityOp - \ket{0}\bra{0}) \otimes H)
\end{align*}

Now we want to find $A$ such that $\Psi_\zeroOp(A) \loewner A$. To do so, we consider the ansatz\footnote{Convergence w.r.t. WOT, thus $A \in \predicate$ by Definition \ref{def:FormAddMult} and Lemma \ref{lem:infSumExpect}.}
\begin{align*}
A = \sum_{n=0}^\infty a_n \ket{n}\bra{n} \otimes \ket{0}\bra{0} + \sum_{n=0}^\infty b_n \ket{n}\bra{n} \otimes \ket{1}\bra{1} + \zeroOp|_{\geq 0}
\end{align*} where $\zeroOp|_{\geq 0} = (\sum_{n=-\infty}^{-1} \infty \ket{n}\bra{n}+\sum_{n=0}^\infty 0\ket{n}\bra{n}) \otimes \identityOp$ is added to ensures that negative values are handled correctly as we know that their expected runtime is infinite. Besides of that part, we choose a diagonal ansatz to keep it simple and figured out that $\sum_{n=0}^\infty a_n \ket{n}\bra{n} \otimes \identityOp$ does not work, so we split the qubit into the two cases of $\ket{0}$ and $\ket{1}$, which makes sense as $S$ depends on the value of the qubit. We used the ansatz and solved $\Psi_\zeroOp(A) \loewner A$ for the minimal solution of $a_n \geq 0$ which is $2n(6+3\sqrt{2}) +1$. For this choice of $a_n$, $(2n-1)(6+3\sqrt{2}) +1$ is the only choice for $b_n$. Details can be found in Appendix \ref{sec:appendix_example2}. As these values are surprisingly high (i.e., $a_1>20$ which means that on average we need more than 20 steps to go from $\ket{1}$ to $\ket{0}$ where each loop-iteration takes $3$ steps), we suspect that there is a better invariant using a different ansatz. As for our purpose it is enough to show that the expected runtime is finite, we did not further optimize this invariant.

Using $A$ with the values from above and Park induction, we get $\expect{\ert{P_1}{\zeroOp}}{\ket{n}\bra{n} \otimes \ket{\psi}\bra{\psi}} \leq 2n(6+3\sqrt{2}) + 2$ for $n \geq 0$ and any state $\psi$, i.e., the expected runtime is finite for all states where $q = \ket{n}$ with $n \geq 0$.

It is $B = \sum_{n=0}^{\infty}n \ket{n}\bra{n} \otimes \identityOp \loewner \ert{P_1}{\zeroOp}$ (see Appendix \ref{sec:appendix_example2}), so the runtime (even besides the negative part) is not bounded.
In fact, the same phenomenon occurs as in the previous example: The expected runtime is finite for all basis states $\ket{n}$ with $n>0$ but still infinite for (some) superpositions of these. The expected runtime of state $\sum_{n=1}^{\infty} \frac{\alpha}{n^2} \ket{n}\bra{n}$ (with suitable $\alpha$) is lower bounded by
\begin{align*}
    \expect{B}{\sum_{n=0}^{\infty} \frac{\alpha}{n^2} \ket{n}\bra{n} \otimes \ket{0}\bra{0}} = \sum_{n=0}^{\infty} \frac{\alpha n}{n^2} \to \infty
\end{align*}

Thus we have shown that the program is non-AST for states $\ket{n}$ with $n<0$ (and superpositions including those states), PAST and AST for states $\ket{n}$ with $n\geq 0$ and AST (but not PAST) for all superpositions of those states.

If we consider an AST variant of the program where the loop is only entered with states $\ket{n}$ with $n> 0$, i.e., we change the behavior on the negative positions, we would change the loop guard to $M_0 = \sum_{n=-\infty}^0 \ket{n}\bra{n}$ and $M_1 = \sum_{n=1}^{\infty} \ket{n}\bra{n}$. Then we can still reuse the same invariant by removing $\zeroOp|_{\geq 0}$ and obtain similar results (except for the negative part where the runtime would be finite).

%% file: 8_conclusion.tex
\section{Conclusion}
\label{sec:conclusion}
We considered quantum programs with rewards in infinite-dimensional spaces and defined semantic properties that should hold in this setting. We gave a concrete syntax and denotational semantics for qrWhile programs with rewards. The main contribution of this paper is the definition of weakest pre-expectations both semantically for quantum programs with rewards and syntactically for qrWhile programs. We dropped the condition of boundedness of the expectation which was required in previous work and this more general definition allows us to reason about runtimes without restricting to finite expected runtimes. We gave some simple loop rules to bound both the weakest preconditions and the expected runtimes of loops to avoid heavy calculations. Finally, we considered two example programs and analyzed their expected runtime and termination probability using the tools we developed which would not have been possible using existing work.
Our goal was to keep the same style of weakest preconditions as in the finite-dimensional case \cite{DHondtWeakestPreconditions} and thus allow for unbounded operators as expectations in comparison to \cite{LicsNewPaper}. This allows us to capture expected runtimes and reward statements in a natural way.

For future work, it would be interesting to consider more general semantic properties of quantum programs with rewards, e.g., not only track the expected reward but instead the distribution of rewards. An approach as \cite{philippHigherMoments} allows to analyze higher moments of the reward distribution which gives more information about the behavior of the program.

Of course, it is also interesting to combine our weakest pre-expectations approach and the runtime analysis methods with other features of quantum programming languages, e.g., non-determinism \cite{FengNondeterministicQuantumVerification}, classical variables \cite{FengQHLClassicalVars}, observations \cite{bayesianInf}, etc. for which we already know how to define semantics and (bounded) weakest preconditions. Another interesting direction is to find a suitable representation language for expectations which is one of the first steps towards developing an automated tool.

%% file: appendix.tex
\section{Details of Section \ref{sec:expectations} - Expectations}
\label{sec:appendix_expectations}
In this section we give additional proof details about expectations and expectation values.

Proof of Proposition \ref{prop:OpstoForms}:
\begin{proof}
    We use the following theorem (\cite[VI, Thm. 2.1, 2.7]{kato} and \cite[Corr. 10.8, Eq. 10.12]{unboundedSelfAdjointBook}):
    "For every positive self-adjoint operator $A$, there is a corresponding densely defined, closed, symmetric, positive form $a: \domain{a} \times \domain{a} \to \mathbb{C}$ with $a[x,y] = \langle \half{A}x, \half{A}y\rangle$ for all $x,y \in \domain{a} = \domain{\half{A}}$ and the other way round."

    We show that this theorem implies the proposition:
    \begin{itemize}
        \item $\Rightarrow$: For every expectation $A$ there is the positive, self-adjoint $\infty$-restriction $A_\infty$. By the above theorem, there is a corresponding densely defined, closed, symmetric, positive form $a$, thus also a form-expectation.
        \item $\Leftarrow$: For every form-expectation $a$, we can construct a densely defined form by extending $a$ to the closure of its domain (that means $\domain{a}$ is dense in $\overline{\domain{a}}$). By the above theorem, there is a corresponding positive self-adjoint operator $A_\infty$. We obtain $A \in \predicate$ by extending the $\infty$-restriction $A_\infty$ to $\hilbert$. \qedhere
    \end{itemize}
\end{proof}

Proof of Theorem \ref{thm:omegaCompleteness}:
\begin{proof}
    Let $\{A_n\}_n$ be an increasing sequence of expectations over $\hilbert$.
    That means every $A_n$ is positive and $\infty$-self-adjoint and $A_n \loewner A_{n+1}$ holds for all $n$.

    \begin{enumerate}
        \item \label{enum.opsToForms} For each $A_n$, there is the corresponding form-expectation $a_n$ by Proposition \ref{prop:OpstoForms}.

        \item \label{enum.order} $a_n \leq a_{n+1}$ holds exactly if $A_n \loewner A_{n+1}$ by Proposition \ref{prop:orderEquivalence}. Thus $\{a_n\}_n$ is a monotone non-decreasing sequence of closed, symmetric forms bounded from below.

        \item \label{enum.conv} We use Kato's theorem \cite[VIII, Thm. 3.13a]{kato}: "Let $a_1 \leq a_2 \leq ...$ be a monotone non-decreasing sequence of closed symmetric forms bounded from below. Set $a[x] = \lim a_n[x]$ whenever the finite limit exists. $a$ is a closed symmetric form bounded from below." As each $a_n$ is positive, $a$ is positive as well. Thus $a \in \formPredicate$. Note that $\domain{a}$ might not be dense in $\hilbert$, but it is dense in $\overline{\domain{a}}$.

        \item \label{enum.formToOps} We use this limit form and define the corresponding expectation $A$ according to Proposition \ref{prop:OpstoForms}.

        Then $A_\infty$ is the self-adjoint operator that corresponds to the limit form \cite[VI, Thm. 2.1, 2.7]{kato} and as each $A_n$ is positive, it is also positive.

        Then $A$ is an expectation with $A x = A_\infty x$ for all $x \in \domain{A} = \domain{A_\infty}$.
    \end{enumerate}
    Now we show that $A$ is indeed the least upper bound of the sequence $A_n$:

        By \ref{enum.order}, we have $A_n \loewner A$ for all $n$ because $a_n \leq a$ for all $n$.
        That means $A$ is an upper bound of the sequence $A_n$.

        Now it remains to show that it is the least upper bound.
        Let $B: \domain{B} \to \hilbert$ be another expectation with $A_n \loewner B$ for all $n$ (a possible other smaller upper bound).
        Then by \ref{enum.order}, it is $a_n \leq b$ for all $n$ where $b$ is the form-expectation that corresponds to $B$.
        Thus, by the definition of the limit form, it is $a \leq b$.
        Again by \ref{enum.order}, this implies $A \loewner B$.

        Convergence on the form-level is pointwise convergence, i.e., $a[x] = \lim_n a_n[x]$ for all $x \in \domain{a}$. For the corresponding expectations, this implies convergence in WOT:
        \begin{align*}
            &\text{ }\forall x \in \domain{A}: a_n[x] \to a[x] \\
            \Rightarrow & \text{ }\forall x,y \in \domain{A}: a_n[x,y] \to a[x,y] \\
            \Rightarrow &\text{ }\forall x,y \in \domain{A}: \langle x, A_n y\rangle \to \langle x, A y\rangle \\
            \Rightarrow &\text{ }A_n \to A \text{ in WOT}
        \end{align*}

     Overall, the proof steps are visualized in the diagram in Figure \ref{fig:omegaCompletenessProof}.
     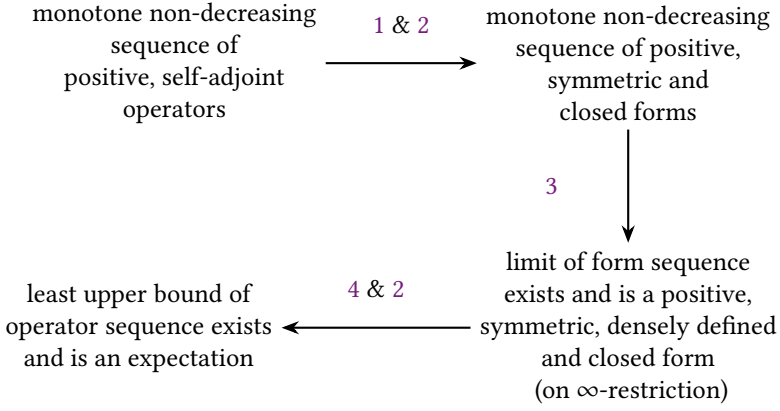
\begin{figure}
        \centering
         \begin{tikzpicture}[
     node distance=2cm,
     every node/.style={rounded corners, align=center, minimum width=2cm, minimum height=1cm},
    arrow/.style={-{Stealth}, thick}]

        \node (A) {monotone non-decreasing \\sequence of \\positive, self-adjoint \\ operators};
        \node (C) [right=of A] {monotone non-decreasing \\sequence of positive,\\symmetric and \\ closed forms};
         \node (G) [below= 1.5cm of C] {limit of form sequence \\ exists and is a positive, \\ symmetric, densely defined\\ and closed form \\ (on $\infty$-restriction)};
         \node (H) [left =2.5cm of G] {least upper bound of \\operator sequence exists\\ and is an expectation};

        \draw[arrow] (A) -- node[above]{\ref{enum.opsToForms} \& \ref{enum.order}}(C);
         \draw[arrow] (C) -- node[left]{\ref{enum.conv}} (G);
         \draw[arrow] (G) -- node[above]{\ref{enum.formToOps} \& \ref{enum.order}} (H);
     \end{tikzpicture}
     \caption{Proof steps for convergence of increasing sequences of expectations}
    \label{fig:omegaCompletenessProof}
     \end{figure}
\end{proof}

Proof of Proposition \ref{prop:FormAddMult}:
\begin{proof} \text{ }
    \begin{itemize}
        \item To show: $a+b$ is a form-expectation. Recall, $a+b$ is defined as the form $c: \domain{c} \times \domain{c} \to \mathbb{C}$ with
        \begin{align*}
            \domain{c} = \domain{a} \cap \domain{b}, c[x,y] = a[x,y] + b[x,y]
        \end{align*}
        The definition of the sum is the same as in \cite[VI, Chap. 2, Sect. 5]{kato} (called 'generalized sum'). By \cite[VI, Thm. 1.31]{kato} is $c$ closed. Symmetry follows from the symmetry of $a$ and $b$ and positivity from the positivity of $a$ and $b$.
        \item To show: $\alpha \cdot a$ is a form-expectation. Recall, $\alpha \cdot a$ is defined as the form $\alpha \cdot a: \domain{\alpha \cdot a} \times \domain{\alpha \cdot a} \to \mathbb{C}$ with
        \begin{align*}
            \domain{\alpha \cdot a} = \domain{a},(\alpha \cdot a)[x,y] = \alpha \cdot a[x,y]
        \end{align*}
        \begin{itemize}
            \item Symmetry follows from the symmetry of $a$, as $(\alpha \cdot a)[x,y] = \alpha \cdot a[x,y] = \alpha \cdot \overline{a[y,x]} = \overline{(\alpha \cdot a)[y,x]}$ and $\alpha$ is real.
            \item Positivity follows from the positivity of $a$, as $(\alpha \cdot a)[x] = \alpha \cdot a[x] \geq 0$ for all $x \in \domain{\alpha \cdot a}$ and $\alpha \geq 0$.
            \item Closeness follows from the closeness of $a$ as every sequence $\{x_n\}_n$ in $\domain{\alpha \cdot a}$ with $x_n \underset{t}{\to} x$ also is a sequence in $\domain{a}$ with $x_n \underset{t}{\to} x$. Thus we have $x\in \domain{\alpha \cdot a} = \domain{a}$ and $(\alpha \cdot a)[x_n - x_m] = \alpha \cdot a[x_n - x_m]\to 0$ for $n,m \to \infty, \alpha\geq 0$. Thus $\alpha \cdot a$ is closed.
        \end{itemize}

        \item To show: $a_M$ is a form-expectation. Recall, $a_M$ is defined as the form $a_M: \domain{a_M} \times \domain{a_M} \to \mathbb{C}$ with \begin{align*}
            \domain{a_M} = \{ \psi \in \hilbert \mid M\psi \in \domain{a}\}, a_M[x,y] = a[Mx, My]
        \end{align*}
        Note that $\bra{\psi} \formsandwich{M}{A}  \ket{\psi} = \bra{M\psi} A \ket{M\psi}$.
         \begin{itemize}
            \item Symmetry follows from the symmetry of $a$, as $a_M[x,y] = a[Mx, My] = \overline{a[My, Mx]} = \overline{a_M[y,x]}$.
            \item Positivity follows from the positivity of $a$, as $a_M[x] = a[Mx] \geq 0$ for all $x \in \domain{a_M}$.
           \item Closeness follows from the closeness of $a$:
        Let $\{x_n\}_n$ be a sequence in $\domain{a_M}$ with $x_n \underset{t}{\to} x$. That means $x_n \in \domain{a_M}$, $x_n \to x$ and $a_M[x_n - x_m] \to 0$ for $n,m \to \infty$. We need to show that $x \in \domain{a_M}$ and $a_M[x_n - x] \to 0$ for $n \to \infty$. It is $a_M[x_n - x_m]=a[M(x_n-x_m)]$, so $a[x_n-M x_m] \to 0$ and $Mx_n \in \domain{a}$ for all $n$. As $M$ is bounded (thus continuous), $Mx_n \to Mx$. By the closeness of $a$, it is $Mx \in \domain{a}$ and $a[Mx_n - Mx] \to 0$ for $n \to \infty$. Thus $x \in \domain{a_M}, Mx_n \to Mx$ and $a_M[x_n - x] \to 0$ for $n \to \infty$. By the closeness of $a$, we get $Mx \in \domain{a}$ and thus $x \in \domain{a_M}$. Thus $a_M$ is closed. \qedhere
        \end{itemize}
    \end{itemize}
\end{proof}

Proof of Lemma \ref{lem:infSumExpect}:

\begin{proof}
    As $A_n \in \predicate$, each $A_n$ is positive. Thus $\{\sum_{n=0}^N A_n \}_N$ with $\sum_{n=0}^N A_n = A_0 \sumform A_1 \sumform ... \sumform A_N$ is an increasing sequence. All sums are finite and thus in $\predicate$ by Definition \ref{def:FormAddMult}.
    By Theorem \ref{thm:omegaCompleteness}, this sequence has a least upper bound/limit $\lim_{N \to \infty} \sum_{n=0}^N A_n \in \predicate$. As convergence in WOT considers sequences of (extended) reals, this limit coincides with $\sum_{n=0}^\infty A_n$ and as the limit is an expectation, $\sum_{n=0}^\infty A_n \in \predicate$.
\end{proof}

Proof of Proposition \ref{prop:ContinuityDistributivityFormSum}:
\begin{proof}
    Let $a, a_n, b, b_n$ be the corresponding form-expectations of $A, A_n, B, B_n$.
    \begin{enumerate}
        \item $\left(\bigvee A_n \right)\sumform B = \bigvee \left(A_n \sumform B\right)$:

        Let $\bigvee A_n = A$.
        By Proposition \ref{prop:orderEquivalence}, $a_n \leq a_{n+1}$ for all $n$ and $\bigvee a_n = a$ by the definition of the limit form. It is $\domain{a} = \{x \in \bigcap_n \domain{a_n} \mid \bigvee a_n[x] < \infty\}$ and $a[x,y] = \bigvee a_n[x,y]$. We show that both sides correspond to the same form and thus to the same operator by Proposition \ref{prop:OpstoForms}:

        The form that corresponds to $\left(\bigvee A_n \right)\sumform B = A \sumform B$ is $c: \domain{c} \times \domain{c} \to \mathbb{C}$ with
        \begin{align*}
            \domain{c} &= \domain{a} \cap \domain{b}\\
             c[x,y] &= a[x,y] + b[x,y]
        \end{align*}for all $x,y \in \domain{c}$.

        The form that corresponds to $\bigvee \left(A_n \sumform B\right)$ is $d: \domain{d} \times \domain{d} \to \mathbb{C}$ with
        \begin{align*}
            \domain{d} &= \{x \in \bigcap_n \domain{a_n} \cap \domain{b} \mid \bigvee a_n[x] < \infty\}\\
             d[x,y] &= \bigvee a_n[x,y] + b[x,y]
        \end{align*}
        for all $x,y \in \domain{d}$.
        Then
        \begin{align*}
        \domain{c} &= \domain{a} \cap \domain{b} = \{x \in \bigcap_n \domain{a_n} \cap \domain{b} \mid \bigvee a_n[x] < \infty\} = \domain{d},\\
        c[x,y] &= a[x,y] + b[x,y] = \bigvee a_n[x,y] + b[x,y] = d[x,y]
        \end{align*} for all $x,y \in \domain{c} = \domain{d}$, so both sides correspond to the same form and thus to the same operator.

        \item $B \sumform \left(\bigvee A_n \right) = \bigvee \left(B \sumform A_n \right)$: Analogous to the previous case.

        \item $\left(\bigvee A_n \right)\sumform \left( \bigvee B_n \right) = \bigvee \left(A_n \sumform B_n \right)$:

        Let $\bigvee A_n = A$ and $\bigvee B_n = B$. It is $\domain{a} = \{x \in \bigcap_n \domain{a_n} \mid \bigvee a_n[x] < \infty\}$, $\domain{b} = \{x \in \bigcap_n \domain{b_n} \mid \bigvee b_n[x] < \infty\}$ and $a[x,y] = \bigvee a_n[x,y], b[x,y] = \bigvee b_n[x,y]$. We show that both sides correspond to the same form and thus to the same operator by Proposition \ref{prop:OpstoForms}:

        The form that corresponds to $\left(\bigvee A_n \right)\sumform \left( \bigvee B_n \right) = A \sumform B$ is $c: \domain{c} \times \domain{c} \to \mathbb{C}$ with
        \begin{align*}
            \domain{c} &= \domain{a} \cap \domain{b}\\
             c[x,y] &= a[x,y] + b[x,y]
        \end{align*}for all $x,y \in \domain{c}$.

        The form that corresponds to $\bigvee \left(A_n \sumform B_n \right)$ is $d: \domain{d} \times \domain{d} \to \mathbb{C}$ with
        \begin{align*}
            \domain{d} &= \{x \in \bigcap_n \domain{a_n}\cap \domain{b_n} \mid \bigvee a_n[x]+ b_n[x] < \infty\} \\
             d[x,y] &= \bigvee a_n[x,y] + b_n[x,y]
        \end{align*}
        for all $x,y \in \domain{d}$.

        As $a_n$ and $b_n$ are monotone non-decreasing sequences of forms, it is $\bigvee a_n[x] + b_n[x] = \bigvee a_n[x] + \bigvee b_n[x]$ for all $x \in \bigcap_n \domain{a_n} \cap \domain{b_n}$.

        Then
        \begin{align*}
            &\domain{c} = \domain{a} \cap \domain{b} = \{x \in \bigcap_n \domain{a_n} \cap \domain{b_n} \mid \bigvee a_n[x] < \infty \land \bigvee b_n[x] < \infty\} = \domain{d},\\
            & c[x,y] = a[x,y] + b[x,y] = \bigvee a_n[x,y] + \bigvee b_n[x,y] = \bigvee a_n[x,y] + b_n[x,y] = d[x,y]
        \end{align*} for all $x,y \in \domain{c} = \domain{d}$, so both sides correspond to the same form and thus to the same operator.

        \item $\formsandwich{M}{\left(\bigvee A_n \right)} = \bigvee \left(\formsandwich{M}{A_n} \right)$:
        Let $\bigvee A_n = A$.
        By Proposition \ref{prop:orderEquivalence}, $a_n \leq a_{n+1}$ for all $n$ and $\bigvee a_n = a$ by the definition of the limit form. It is $\domain{a} = \{x \in \bigcap_n \domain{a_n} \mid \bigvee a_n[x] < \infty\}$ and $a[x,y] = \bigvee a_n[x,y]$. We show that both sides correspond to the same form and thus to the same operator by Proposition \ref{prop:OpstoForms}:

        The form that corresponds to $\formsandwich{M}{\left(\bigvee A_n \right)} = \formsandwich{M}{A}$ is $c: \domain{c} \times \domain{c} \to \mathbb{C}$ with
        \begin{align*}
            \domain{c} &= \{x \in \hilbert \mid Mx \in \domain{a}\}\\
             c[x,y] &= a[Mx, My]
        \end{align*}for all $x,y \in \domain{c}$.

        The form that corresponds to $\bigvee \left(\formsandwich{M}{A_n} \right)$ is $d: \domain{d} \times \domain{d} \to \mathbb{C}$ with
        \begin{align*}
            \domain{d} &= \{x \in \hilbert \mid Mx \in \bigcap_n \domain{a_n} \text{ and } \bigvee a_n[Mx] < \infty\}\\
             d[x,y] &= \bigvee a_n[Mx, My]
        \end{align*}
        for all $x,y \in \domain{d}$.
        Then
        \begin{align*}
            &\domain{c} = \{x \in \hilbert \mid Mx \in \domain{a}\} = \{x \in \hilbert \mid Mx \in \bigcap_n \domain{a_n} \text{ and } \bigvee a_n[Mx] < \infty\} = \domain{d},\\
            & c[x,y] = a[Mx, My] = \bigvee a_n[Mx, My] = d[x,y]
        \end{align*} for all $x,y \in \domain{c} = \domain{d}$, so both sides correspond to the same form and thus to the same operator.
        \item $\formsandwich{M}{(A \sumform B)} = \formsandwich{M}{A} \sumform \formsandwich{M}{B}$:
        The form that corresponds to $\formsandwich{M}{(A \sumform B)}$ is $c: \domain{c} \times \domain{c} \to \mathbb{C}$ with
        \begin{align*}
            \domain{c} &= \{x \in \hilbert \mid Mx \in \domain{a} \cap \domain{b}\}\\
             c[x,y] &= a[Mx, My] + b[Mx, My]
        \end{align*}for all $x,y \in \domain{c}$.

        The form that corresponds to $\formsandwich{M}{A} \sumform \formsandwich{M}{B}$ is $d: \domain{d} \times \domain{d} \to \mathbb{C}$ with
        \begin{align*}
            \domain{d} &= \{x \in \hilbert \mid Mx \in \domain{a} \} \cap \{x \in \hilbert \mid Mx \in \domain{b}\}\\
             d[x,y] &= a[Mx, My] + b[Mx, My]
        \end{align*}
        for all $x,y \in \domain{d}$.
        Then
        \begin{align*}
            \domain{c} &= \{x \in \hilbert \mid Mx \in \domain{a} \cap \domain{b}\} = \domain{d},\\
            c[x,y] &= a[Mx, My] + b[Mx, My] = d[x,y]
        \end{align*} for all $x,y \in \domain{c} = \domain{d}$, so both sides correspond to the same form and thus to the same operator. \qedhere
    \end{enumerate}
\end{proof}

Proof of Lemma \ref{lem:distrInfSum}:
\begin{proof}
    Similar to the proof of Lemma \ref{lem:infSumExpect}, we consider finite sums first, i.e., $\sum_{n=0}^N A_n$. By Proposition \ref{prop:ContinuityDistributivityFormSum}, it is $\formsandwich{M}{\sum_{n=0}^N A_n} = \sum_{n=0}^N \formsandwich{M}{A_n}$ for all $N$.

    By Theorem \ref{thm:omegaCompleteness}, there is a least upper bound of the sequence $\{\sum_{n=0}^N A_n\}_{N}$, which is exactly $\sum_{n=0}^\infty A_n$. Using Proposition \ref{prop:ContinuityDistributivityFormSum}, it is \begin{align*}
        \formsandwich{M}{(\sum_{n=0}^\infty A_n)} =& \text{ }\formsandwich{M}{\left(\bigvee_{N} \sum_{n=0}^N A_n\right)} \\
        =& \bigvee_N \formsandwich{M}{\sum_{n=0}^N A_n} \\
        =& \bigvee_N \sum_{n=0}^N \formsandwich{M}{A_n} \\
        =& \sum_{n=0}^\infty \formsandwich{M}{A_n}\qedhere
    \end{align*}
\end{proof}

Proof of Proposition \ref{prop:BoundedAddition}:
\begin{proof}
    Let $A,B \in \predicate$ be bounded expectations over $\hilbert$, $\alpha\geq 0$ and $M$ be a bounded operator on $\hilbert$.
    As $A,B$ are bounded, their $\infty$-restrictions $A_\infty, B_\infty$ are bounded and thus defined on $\hilbert$. Thus the forms $a,b$ that correspond to $A_\infty, B_\infty$ are defined on $\hilbert$ as well (and bounded by \cite[VI, Thm. 2.7]{kato}).
    We have $a[x,y] = \bra{x} A \ket{y}$ and $b[x,y] = \bra{x} B \ket{y}$ for all $x,y \in \hilbert$. Then
    \begin{itemize}
        \item Addition:
         $\domain{c} = \hilbert$ with $c[x,y] = a[x,y] + b[x,y] = \bra{x} A+B \ket{y}$ for all $x,y \in \hilbert$.
        \item Scalar multiplication:
         $\domain{\alpha \cdot a} = \hilbert$ with $(\alpha \cdot a)[x,y] = \alpha \cdot a[x,y] = \bra{x} \alpha \cdot A \ket{y}$ for all $x,y \in \hilbert$.
        \item Application of bounded operator:
        $\domain{a_M} = \hilbert$ with $a_M[x,y] = a[Mx, My] = \bra{Mx} A \ket{My} = \bra{x} M^\dagger A M \ket{y}$ for all $x,y \in \hilbert$. \qedhere
    \end{itemize}
\end{proof}

Proof of Proposition \ref{prop:orderproperties}:
\begin{proof}
    Assume $A,B,C \in \predicate$ and $M$ bounded.
    \begin{itemize}
        \item \textbf{Antisymmetric:} $A \loewner B \land B \loewner A \Rightarrow A=B$.

        $A \loewner B$ implies $\domain{\half{B_\infty}} \subseteq \domain{\half{A_\infty}}$ and the other way round, so $\domain{\half{A_\infty}} = \domain{\half{B_\infty}}$. Also $\norm{\half{A_\infty} \psi}^2 = \norm{\half{B_\infty} \psi}^2$ for all $\ket{\psi} \in \domain{\half{B_\infty}}$, so $A=B$.

        \textbf{Reflexive:} It is $A \loewner A$.

        Clearly $\domain{\half{A_\infty}} \subseteq \domain{\half{A_\infty}}$ and $\norm{\half{A_\infty} \psi}^2 \leq \norm{\half{A_\infty} \psi}^2$, so $A \loewner A$.

        \textbf{Transitive:} $A \loewner B \land B \loewner C \Rightarrow A \loewner C$.

        $A \loewner B$ means $\domain{\half{B_\infty}} \subseteq \domain{\half{A_\infty}}$ and $\norm{\half{A_\infty} \psi}^2 \leq \norm{\half{B_\infty} \psi}^2$ for all $\ket{\psi} \in \domain{\half{B_\infty}}$.
        $B \loewner C$ means $\domain{\half{C_\infty}} \subseteq \domain{\half{B_\infty}}$ and $\norm{\half{B_\infty} \psi}^2 \leq \norm{\half{C_\infty} \psi}^2$ for all $\ket{\psi} \in \domain{\half{C_\infty}}$.

        So we have $\domain{\half{C_\infty}} \subseteq \domain{\half{B_\infty}} \subseteq \domain{\half{A_\infty}}$ (in particular $\domain{\half{C_\infty}}\subseteq \domain{\half{A_\infty}}$) and $\norm{\half{A_\infty} \psi}^2 \leq \norm{\half{B_\infty} \psi}^2 \leq \norm{\half{C_\infty} \psi}^2$ for all $\ket{\psi} \in \domain{\half{C_\infty}}$. Thus $A \loewner C$.

        \item Let $a,b,c$ be the forms that correspond to $A_\infty, B_\infty, C_\infty$ and $A \loewner B$. Then $\domain{b} = \domain{\half{B_\infty}} \subseteq \domain{\half{A_\infty}} = \domain{a}$ and $\norm{\half{A_\infty} \psi}^2 \leq \norm{\half{B_\infty} \psi}^2$ for all $\ket{\psi} \in \domain{\half{B_\infty}}$.

        It is $\domain{\half{(B \sumform C)_\infty}} = \domain{b} \cap \domain{c} \subseteq \domain{a} \cap \domain{c} = \domain{\half{(A \sumform C)_\infty}}$.
        For all $\ket{\psi} \in \domain{\half{(B \sumform C)_\infty}}$ it is
        \begin{align*}
            \norm{\half{(A \sumform C)_\infty} \psi}^2 &= a[\psi]+ c[\psi] \leq b[\psi] + c[\psi] = \norm{\half{(B \sumform C)_\infty} \psi}^2
        \end{align*}
        Thus $A \sumform C \loewner B \sumform C$.

        \item Let $a,b$ be the forms that correspond to $A_\infty, B_\infty$ and $A \loewner B$. That means $a\leq b$, i.e., $\domain{b} \subseteq \domain{a}$ and $a[x] \leq b[x]$ for all $x \in \domain{b}$ by Proposition \ref{prop:orderEquivalence}.
        Then $\domain{b_M} = \{ \psi \in \hilbert \mid M\psi \in \domain{b}\} \subseteq \{ \psi \in \hilbert \mid M\psi \in \domain{a}\} = \domain{a_M}$.
        For all $\psi \in \domain{b_M}$ it is $a_M[\psi] = a[M\psi] \leq b[M\psi] = b_M[\psi]$, thus $a_M \leq b_M$ and thus $\formsandwich{M}{A} \loewner \formsandwich{M}{B}$ by Proposition \ref{prop:orderEquivalence}.\qedhere

    \end{itemize}
\end{proof}

Proof of Lemma \ref{lem:expectAlternativeDef}:
\begin{proof}
    In this proof we use $\expect{A}{\rho}$ to denote the expectation value defined via the spectral measure in Definition \ref{def:expectationValue} and $\expect{A_1}{\rho}$ and $\expect{A_2}{\rho}$ to denote the two equations in this lemma.

    We first show the statement for pure states $\rho = \ket{\psi}\bra{\psi}$.

    It is $\domain{\half{A_\infty}}=\{ \ket{\psi} \mid \int \abs{\lambda} d (tr(E_\lambda \ket{\psi}\bra{\psi})) < \infty\}$ by \cite[Sec. 10.2, Eq. 10.7]{unboundedSelfAdjointBook}.

    If $\ket{\psi} \not \in \domain{\half{A_\infty}} = \domain{a}$, then $\int \abs{\lambda} d (tr(E_\lambda \ket{\psi}\bra{\psi})) = \infty$ and thus $\expect{A}{\ket{\psi}\bra{\psi}} = \infty$ in all cases.

    If $\ket{\psi} \in \domain{\half{A_\infty}}$, then $\int \abs{\lambda} d (tr(E_\lambda \ket{\psi}\bra{\psi}))= \int \abs{\lambda} d (\bra{\psi}E_\lambda \ket{\psi})< \infty$. First we show equivalence between $\expect{A}{\ket{\psi}\bra{\psi}}$ and $\expect{A_1}{\ket{\psi}\bra{\psi}}$:

    We define $I_n(x)=\begin{cases}
        x & \text{if } x \leq n \\
        0 & \text{otherwise}
    \end{cases}$ for all $n >0$. Then $I_n(A_\infty),I_n(\half{A_\infty})$ is a bounded operator for all $n >0$ and $I_n(A_\infty) \to A_\infty, I_n(\half{A_\infty}) \to \half{A_\infty}$ strongly.
    \allowdisplaybreaks
    \begin{align*}
        & \norm{\half{A_\infty} \ket{\psi}}^2 && (I_n(\half{A_\infty}) \to \half{A_\infty}\text{ strongly})\\
        =& \lim_{n\to \infty} \norm{I_n(\half{A_\infty}) \ket{\psi}}^2 \\
        =& \lim_{n\to \infty} \norm{\half{(I_{n^2}(A_\infty))} \ket{\psi}}^2 && (\text{def. norm})\\
        =& \lim_{n\to \infty} \bra{\psi} (\half{(I_{n^2}(A_\infty))})^2 \ket{\psi} \\
        =& \lim_{n\to \infty} \bra{\psi} I_{n^2}(A_\infty) \ket{\psi} \\
        =& \lim_{n\to \infty} \int I_{n^2}(\lambda) d (\bra{\psi}E_\lambda \ket{\psi}) && (\text{dominated convergence theorem})\\
        =& \int \lim_{n\to \infty} I_{n^2}(\lambda) d (\bra{\psi}E_\lambda \ket{\psi}) \\
        =& \int \lambda d (\bra{\psi}E_\lambda \ket{\psi}) \\
        =& \int \lambda d (tr(E_\lambda \ket{\psi}\bra{\psi}))
    \end{align*}
    The dominated convergence theorem applies as $\abs{I_{n^2}(x)} \leq x$ for all $n>0, x$ because $\lim_{n \to \infty} I_{n^2}(x) = x$ for all $x$. Also $\int \abs{\lambda} d (\bra{\psi}E_\lambda \ket{\psi}) < \infty$ because $\ket{\psi} \in \domain{\half{A_\infty}}$.

    The equivalence between $\expect{A_1}{\ket{\psi}\bra{\psi}}$ and $\expect{A_2}{\ket{\psi}\bra{\psi}}$ follows from the definition of the form $a$ that corresponds to $A_\infty$ as $a[\ket{\psi}] = \norm{\half{A_\infty} \ket{\psi}}^2$ for all $\ket{\psi} \in \domain{a}$.

    Now we consider the case of mixed states $\rho = \sum_i p_i \ket{\psi_i}\bra{\psi_i}$ with $p_i > 0$.
    If any $\ket{\psi_i}$ is not in $\domain{\half{A_\infty}}$, then $\int \abs{\lambda} d (tr(E_\lambda \ket{\psi_i}\bra{\psi_i})) = \infty$ and thus $\expect{A}{\rho} = \infty$ in all cases.

   If $\ket{\psi_i} \in \domain{\half{A_\infty}} = \domain{a}$ for all $i$, then by \cite[Eq. 3.3]{ExpectationValuesKraus}
   \begin{align*}
   \expect{A}{\rho} = \sum_i p_i \expect{A}{\ket{\psi_i}\bra{\psi_i}} = \sum_i p_i \int \lambda d (tr(E_\lambda \ket{\psi_i}\bra{\psi_i}))
   \end{align*}
   As both $\expect{A_1}{\rho}$ and $\expect{A_2}{\rho}$ are defined as $\sum_i p_i \expect{A_1}{\ket{\psi_i}\bra{\psi_i}}$ and $\sum_i p_i \expect{A_2}{\ket{\psi_i}\bra{\psi_i}}$ respectively, and we have already shown the equivalence for pure states, we are done.
\end{proof}

Proof of Lemma \ref{lem:orderExpectation}:
\begin{proof}
    $\Rightarrow$: Assume $A \loewner B$. Then $\domain{\half{B_\infty}} \subseteq \domain{\half{A_\infty}}$ and $\norm{\half{A_\infty} \psi}^2 \leq \norm{\half{B_\infty} \psi}^2$ for all $\ket{\psi} \in \domain{\half{B_\infty}}$.
    Let $\rho \in \density$ be a density operator with decomposition $\rho = \sum_i p_i \ket{\psi_i}\bra{\psi_i}$.
    If there is any $\ket{\psi_i} \not \in \domain{\half{B_\infty}}$, then $\expect{B}{\rho} = \infty$ and thus $\expect{A}{\rho} \leq \expect{B}{\rho}$.
    Otherwise, it is
    \begin{align*}
        \expect{A}{\rho} =& \sum_i p_i \expect{A}{\ket{\psi_i}\bra{\psi_i}} = \sum_i p_i \norm{\half{A_\infty} \ket{\psi_i}}^2 \\
        \leq& \sum_i p_i \norm{\half{B_\infty} \ket{\psi_i}}^2 = \sum_i p_i \expect{B}{\ket{\psi_i}\bra{\psi_i}} = \expect{B}{\rho}
    \end{align*}
    $\Leftarrow$: Assume that for all $\rho \in \density$ it is $\expect{A}{\rho} \leq \expect{B}{\rho}$.
    Let $\ket{\psi} \in \domain{\half{B_\infty}}$ and consider the density operator $\rho = \ket{\psi}\bra{\psi}$.
    Then it is \begin{align*}
        \expect{A}{\ket{\psi}\bra{\psi}} \leq \expect{B}{\ket{\psi}\bra{\psi}}
    \end{align*}
    by assumption.
    As $\ket{\psi} \in \domain{\half{B_\infty}}$, it is $\expect{B}{\ket{\psi}\bra{\psi}} = \norm{\half{B_\infty} \ket{\psi}}^2 < \infty$, thus $\expect{A}{\ket{\psi}\bra{\psi}} < \infty$ and thus $\ket{\psi} \in \domain{\half{A_\infty}}$, i.e., $\domain{\half{B_\infty}} \subseteq \domain{\half{A_\infty}}$.
    Furthermore, it is \begin{align*}
        \norm{\half{A_\infty} \ket{\psi}}^2 = \expect{A}{\ket{\psi}\bra{\psi}} \leq \expect{B}{\ket{\psi}\bra{\psi}} = \norm{\half{B_\infty} \ket{\psi}}^2
    \end{align*}
    Overall, we have $\domain{\half{B_\infty}} \subseteq \domain{\half{A_\infty}}$ and $\norm{\half{A_\infty} \psi}^2 \leq \norm{\half{B_\infty} \psi}^2$ for all $\ket{\psi} \in \domain{\half{B_\infty}}$, thus $A \loewner B$.
\end{proof}

Proof of Proposition \ref{prop:expectationProperties}:
\begin{proof}
    \textbf{Linearity:}
    First we consider the sum of two expectations. Let $A,B \in \predicate$ and $\rho= \sum_i p_i \ket{\psi_i}\bra{\psi_i} \in \density$.
    We distinguish two cases:
    \begin{enumerate}
        \item There is at least one $\ket{\psi_i} \not \in \domain{\half{(A \sumform B)_\infty}}$.
        If there is any $\ket{\psi_i} \not \in \domain{\half{(A \sumform B)_\infty}}$, then $\expect{A\sumform B}{\rho} = \infty$. $\ket{\psi_i} \not \in \domain{\half{(A\sumform B)_\infty}} = \domain{\half{A_\infty}}\cap \domain{\half{B_\infty}}$ implies $\ket{\psi_i} \not \in \domain{\half{A_\infty}}$ or $\ket{\psi_i} \not \in \domain{\half{B_\infty}}$ and thus $\expect{A}{\rho} = \infty$ or $\expect{B}{\rho} = \infty$. Thus $\expect{A}{\rho} + \expect{B}{\rho} = \infty$.
        \item All $\ket{\psi_i} \in \domain{\half{(A \sumform B)_\infty}}$.
        We consider the forms $a,b$ that correspond to $A_\infty, B_\infty$ and $c$ for $ (A \sumform B)_\infty$.
        Then it is $\ket{\psi_i} \in \domain{a}$ and $\ket{\psi_i} \in \domain{b}$ for all $i$ and thus
    \begin{align*}
        \expect{A\sumform B}{\rho} =& \sum_i p_i \expect{A\sumform B}{\ket{\psi_i}\bra{\psi_i}} = \sum_i p_i c[\psi_i] = \sum_i p_i \left( a[\psi_i] + b[\psi_i] \right) \\
        =& \sum_i p_i a[\psi_i] + \sum_i p_i b[\psi_i] = \expect{A}{\rho} + \expect{B}{\rho}
    \end{align*}
    \end{enumerate}

    The linearity w.r.t. states follow from the definition (more details in the proof of Lemma \ref{lem:expectAlternativeDef}).

    \textbf{Scalar multiplication:}
    For states it follows again from the definition of the expectation value (more details in the proof of Lemma \ref{lem:expectAlternativeDef}).

    For the expectation, we distinguish two cases again. Let $A \in \predicate$ with corresponding form $a$ and $cA$ the form corresponding to $cA$ and $\rho= \sum_i p_i \ket{\psi_i}\bra{\psi_i} \in \density$.
    \begin{enumerate}
        \item There is at least one $\ket{\psi_i} \not \in \domain{\half{A_\infty}} = \domain{a} = \domain{ca} = \domain{\half{cA_\infty}}$.
        Then $\expect{cA}{\rho} = \infty = \expect{A}{\rho}$.
        \item All $\ket{\psi_i} \in \domain{\half{A_\infty}} = \domain{\half{cA_\infty}}$.
        Then it is
    \begin{align*}
        \expect{cA}{\rho} =& \sum_i p_i \cdot \expect{cA}{\ket{\psi_i}\bra{\psi_i}}
        = \sum_i p_i \cdot ca[\psi_i]  = \sum_i p_i c\cdot a[\psi_i] \\
        =& \sum_i p_i \cdot c \cdot \expect{A}{\ket{\psi_i}\bra{\psi_i}} \\
        =& \text{ } c \cdot \sum_i p_i \cdot \expect{A}{\ket{\psi_i}\bra{\psi_i}} = c \cdot \expect{A}{\rho}
    \end{align*}
    \end{enumerate}

    \textbf{Continuity:}
    We start with the sequence of expectations:
    Let $\{A_i\}_i \subseteq \predicate$ be an increasing chain of expectations, $\rho \in \density$ and $A = \bigvee_i A_i$ (existence follows by Theorem \ref{thm:omegaCompleteness}).
    Let $\rho = \sum_j p_j \ket{\psi_j}\bra{\psi_j}$ be a decomposition of $\rho$ with $p_j > 0$.
    We distinguish two cases:
    \begin{enumerate}
        \item There is at least one $\ket{\psi_j} \not \in \domain{\half{A_\infty}} = \domain{a}$. Then $\expect{A}{\rho} = \infty$. It also means that $a[\psi_j] = \lim_i a_i[\psi_j] = \infty$ as otherwise $\ket{\psi_j} \in \domain{a}$.
        If $\ket{\psi_j}$ not in $\domain{a_i}$ for at least one $i$, then $\expect{A_i}{\rho} = \infty$ and thus $\bigvee_i \expect{A_i}{\rho} = \infty$. Otherwise, it is $\ket{\psi_j} \in \domain{a_i}$ for all $i$ and thus $\expect{A_i}{\ket{\psi_j}\bra{\psi_j}} = a_i[\psi_j] \not = \infty$ for all $i$. As $\lim_i a_i[\psi_j] = \infty$, it is $\bigvee_i p_j \expect{A_i}{\ket{\psi_j}\bra{\psi_j}} = \infty$ and thus $\bigvee_i \expect{A_i}{\rho} \geq \bigvee_i p_j \expect{A_i}{\ket{\psi_j}\bra{\psi_j}} = \infty$.
        \item All $\ket{\psi_j} \in \domain{\half{A_\infty}}$.
        Then it is
    \begin{align*}
        \expect{A}{\rho} \overset{Lem. \ref{lem:expectAlternativeDef}}{=}& \sum_j p_j a[\psi_j] \\
        \overset{Def}{=}& \sum_j p_j \lim_i a_i[\psi_j]\\
         =& \sum_j \lim_i p_j a_i[\psi_j]\\
        \overset{*}{=}& \lim_i \sum_j p_j  a_i[\psi_j] \\
        \overset{Lem. \ref{lem:expectAlternativeDef}}{=}& \lim_i \sum_j p_j \expect{A_i}{\ket{\psi_j}\bra{\psi_j}} \\
        \overset{Incr. Seq.}{=}& \bigvee_i \sum_j p_j \expect{A_i}{\ket{\psi_j}\bra{\psi_j}} \\
        \overset{Lem. \ref{lem:expectAlternativeDef}}{=}& \bigvee_i \expect{A_i}{\rho}
    \end{align*}
    where the step marked with $*$ follows by the monotone convergence theorem as supremum and limit coincide here because $a_i[\psi_j]\leq a_{i+1}[\psi_j]$ for all $i,j$.
    \end{enumerate}

    Now we consider the sequence of states:
    Let $\{ \rho_i \}_i \subseteq \density$ be an increasing chain of density operators, $A \in \predicate$ and $\rho = \bigvee_i \rho_i$ (existence follows by \cite[Prop. 8.2.1]{YingPredicateTranserformerSemantics}).
    To show: $\expect{A}{\rho} = \bigvee_i \expect{A}{\rho_i}$.

    For each $n$, define $A_n = \int_0^n \lambda d E_\lambda$ where $E_\lambda$ is the spectral measure of $A_\infty$. Then $A_n$ is a bounded operator for all $n$ and $A_n \to A_\infty$ strongly. We have $\expect{A_n}{\rho} = tr(A_n \rho)$ and $\expect{A_n}{\rho_i} = tr(A_n \rho_i)$ for all $i,n$. This will be denoted by (1).
    As $\rho_i \to \rho$ in trace norm, it is $tr(A_n \rho_i) \to tr(A_n \rho)$ for all $n$, denoted by (2)

    As the expectation value is continuous in the expectation, it is $\expect{A}{\rho} = \bigvee_n \expect{A_n}{\rho}$ and $\expect{A}{\rho_i} = \bigvee_n \expect{A_n}{\rho_i}$ for all $i$. This will be denoted by (3). Overall we have
    \begin{align*}
        \expect{A}{\bigvee_i \rho_i} \overset{(3)}{=}& \bigvee_n \expect{A_n}{\bigvee_i \rho_i} \\
        \overset{(1)}{=}& \bigvee_n tr(A_n \bigvee_i \rho_i) \\
        \overset{(2)}{=}& \bigvee_n \bigvee_i tr(A_n \rho_i) \\
        =& \bigvee_i \bigvee_n tr(A_n \rho_i) \\
        \overset{(1)}{=}& \bigvee_i \bigvee_n \expect{A_n}{\rho_i} \\
        \overset{(3)}{=}& \bigvee_i \expect{A}{\rho_i}
    \end{align*}
    Exchanging the suprema is possible as both sequences (fixed i, varying n and fixed n, varying i) are increasing.
\end{proof}

Proof of Lemma \ref{lem:infSumLinExpectation}:
\begin{proof}
    We show the first statement. We consider finite sums $\sum_{n=0}^N A_n$, then $\sum_n A_n$ is $\bigvee_N \sum_{n=0}^N A_n$.
    \begin{align*}
        \expect{\sum_n A_n}{\rho} =&\text{ } \expect{\bigvee_N \sum_{n=0}^N A_n}{\rho} \\
        =& \bigvee_N \expect{\sum_{n=0}^N A_n}{\rho} && (\text{continuity of the expectation})\\
        =& \bigvee_N \sum_{n=0}^N \expect{A_n}{\rho} && (\text{linearity of the expectation})\\
        =& \sum_n \expect{A_n}{\rho}
    \end{align*}

    The second statement follows analogously with finite sums $\sum_{n=0}^N \rho_n$.
\end{proof}

Proof of Proposition \ref{prop:boundedOpinExpectation}:
\begin{proof}
    Let $a$ be the form-expectation of $A$ and $a_M$ the one of $\formsandwich{M}{A}$.
    Let $\rho = \sum_i p_i \ket{\psi_i}\bra{\psi_i}$ be a decomposition with $p_i > 0$.

    If $\ket{\psi_i} \not \in \domain{a_M}$ for at least one $i$, then $\expect{\formsandwich{M}{A}}{\rho} = \infty$. As $\ket{\psi_i} \not \in \domain{a_M} = \{ \psi \in \hilbert \mid M\psi \in \domain{a}\}$, it is $M\ket{\psi_i} \not \in \domain{a}$ and thus $\expect{A}{M\rho M^\dagger} = \infty$. Thus $\expect{\formsandwich{M}{A}}{\rho} = \expect{A}{M\rho M^\dagger}$.

    Otherwise, $\expect{\formsandwich{M}{A}}{\rho} = \sum_i p_i a_M[\psi_i] = \sum_i p_i a[M\psi_i] = \expect{A}{M \rho M^\dagger}$.
\end{proof}

\section{Details of Section \ref{sec:programs} - A Quantum Programming Language with Rewards}
\label{sec:appendix_pl}

Proof of Proposition \ref{prop:supCPmaps}:
\begin{proof}
    We show that each $g_i$ corresponds to a positive, symmetric and closed form $a_i$ and then we can write $g$ as a least upper bound of the $a_i$ and thus it is also a positive, symmetric and closed form by Proposition \ref{prop:OpstoForms} and Theorem \ref{thm:omegaCompleteness}.

    It is $a_i[\ket{\psi}] = g_i(\ket{\psi}\bra{\psi})$ for all $\ket{\psi} \in \hilbert$ (using $\domain{a_i} = \hilbert$ as $g_i$, and thus $a$, is bounded).
    Positivity and symmetry follow immediately from the co-domain of $g_i$.
    It remains to show that $a_i$ is closed. Let $\ket{\psi_n}$ be a sequence in $\hilbert$ such that $\lim_{n \to \infty} \ket{\psi_n} = \ket{\psi}$ and $\lim_{n,m \to \infty}a_i[\ket{\psi_n} - \ket{\psi_m}] = 0$. As $a_i$ is bounded, $\ket{\psi} \in \hilbert = \domain{a_i}$. It remains to show that $\lim_{n\to\infty} a_i[\ket{\psi_n} - \ket{\psi}] = 0$:
     \begin{align*}
        & \lim_{n\to\infty}\ket{\psi_n} = \ket{\psi}\\
        \Rightarrow & \lim_{n\to\infty} \ket{\psi_n} - \ket{\psi} = 0\\
        \Rightarrow & \lim_{n\to \infty} (\ket{\psi_n} - \ket{\psi})(\ket{\psi_n} - \ket{\psi})^\dagger = 0\\
        \Rightarrow & \lim_{n\to \infty} g_i((\ket{\psi_n} - \ket{\psi})(\ket{\psi_n} - \ket{\psi})^\dagger) = g_i(0)\\
        \Rightarrow & \lim_{n\to \infty} a_i[\ket{\psi_n} - \ket{\psi}] = 0
    \end{align*}
    because $g_i$ is continuous as it is bounded and positive.
\end{proof}

Proof of Proposition \ref{prop:semanticsProperties}:
\begin{proof}
    \begin{enumerate}
        \item See \cite[Lem. 5.1]{floydHoareLogic} as $\rewardR$ behaves like $\skipbf$.
        \item See \cite[Prop. 5.2]{floydHoareLogic} as $\rewardR$ behaves like $\skipbf$.
        \item We already know that $\semanticsD{S}$ is linear and trace non-increasing. Now we do an induction over the structure to show that it can be represented in Kraus form, which implies complete positivity.
        \begin{itemize}
            \item For $\skipbf, \qzero, \Uq$ and $\rewardR$, the semantics have Kraus form directly.
            \item For $\concat$, if $\semanticsD{S_1}$ and $\semanticsD{S_2}$ can be represented in Kraus form with operators $E_i$ and $F_j$ respectively, then $\semanticsD{S_2}(\semanticsD{S_1}(\rho)) = \sum_{i,j} F_j E_i \rho E_i^\dagger F_j^\dagger$ is a Kraus representation of $\semantics{\concat}$.
            \item For $\measure$, if $\semanticsD{S_m}$ can be represented in Kraus form with operators $E_{m,i}$ for all $m$, then we obtain a Kraus representation by
            \begin{equation*}
            \semanticsD{\measure}(\rho) = \sum_{m,i} E_{m,i} M_m \rho M_m^\dagger E_{m,i}^\dagger.
            \end{equation*}

            \item For $\while$, we define $\mathcal{M}_0(\rho) = M_0 \rho M_0^\dagger$ and $\mathcal{M}_1(\rho) = M_1 \rho M_1^\dagger$ for all $\rho \in \density$. Then we define
            \begin{align*}
                W_0 =& \text{ }\Omega = 0\\
                W_{i+1} =& \text{ }\semanticsD{\textbf{if }M[\overline{q}]=1 \textbf{ then } S; W_i} \\
                =& \text{ }\mathcal{M}_1 \circ \semanticsD{S} \circ W_i + \mathcal{M}_0
            \end{align*}
            Intuitively, $W_i$ corresponds to the semantics of the program that executes at most $i$ iterations of the loop and then terminates.
            Then we can show $W_i = \sum_{k=0}^{i-1} (\mathcal{M}_1 \circ \semanticsD{S})^k \circ \mathcal{M}_0$ by induction over $i$.
            \begin{itemize}
                \item $i=1$: $W_1 = \mathcal{M}_1 \circ \semanticsD{S} \circ W_0 + \mathcal{M}_0  = \mathcal{M}_0 = \sum_{k=0}^0 (\mathcal{M}_1 \circ \semanticsD{S})^k \circ \mathcal{M}_0$.
                \item $i \to i+1$: We have
            \begin{align*}
                W_{i+1} =& \text{ }\mathcal{M}_1 \circ \semanticsD{S} \circ W_i + \mathcal{M}_0 \\
                =& \text{ }\mathcal{M}_1 \circ \semanticsD{S} \circ \sum_{k=0}^{i-1} (\mathcal{M}_1 \circ \semanticsD{S})^k \circ \mathcal{M}_0 + \mathcal{M}_0 \\
                 =& \sum_{k=1}^{i} (\mathcal{M}_1 \circ \semanticsD{S})^k \circ \mathcal{M}_0 + \mathcal{M}_0 \\
                 =& \sum_{k=0}^{i} (\mathcal{M}_1 \circ \semanticsD{S})^k \circ \mathcal{M}_0
             \end{align*}
             \end{itemize}
             As $\mathcal{M}_1$ and $\semanticsD{S}$ can be represented in Kraus form, $W_i$ can be represented in Kraus form as well by the previous cases, i.e., it is completely positive.
             We can show that $W_{i+1}-W_i$ can be represented in Kraus form as well:
             \begin{align*}
                 W_{i+1} - W_i =& \sum_{k=0}^{i} (\mathcal{M}_1 \circ \semanticsD{S})^k \circ \mathcal{M}_0 - \sum_{k=0}^{i-1} (\mathcal{M}_1 \circ \semanticsD{S})^k \circ \mathcal{M}_0 \\
                 =& \text{ }(\mathcal{M}_1 \circ \semanticsD{S})^i \circ \mathcal{M}_0
             \end{align*}
             which can be represented in Kraus form by the previous cases as it is finite concatenation of Kraus-representable maps.

             Then we have $\bigvee_{i=0}^\infty W_i = \sum_{i=0}^{\infty} W_{i+1}-W_i$ where each summand $W_{i+1}-W_i$ can be represented in Kraus form, which implies that $\sum_{i=0}^\infty W_{i+1}-W_i$ can be represented in Kraus form as well by taking the union of the Kraus operators for all $i$ \cite{isabelleproofKraus}.

         \end{itemize}

         \item $\semanticsR{S}(\rho)$ is linear in $\rho$ can be shown by induction over the structure of $S$. Let $\rho, \sigma \in \density$ and $\alpha, \beta \in \C$ such that $\alpha \rho, \beta \sigma, \alpha \rho + \beta \sigma \in \density$.
        \begin{itemize}
            \item For $\skipbf, \qzero, \Uq$ it follows immediately from the definition ($=0$ always).
            \item For $\rewardR$, we have
            \begin{align*}
                & \text{ }\semanticsR{\rewardR}(\alpha \rho + \beta \sigma) \\
                =&  \text{ }c \cdot tr(\alpha \rho + \beta \sigma) \\
                =& \text{ }c\cdot (\alpha tr(\rho) + \beta tr(\sigma)) \\
                = & \text{ }\alpha \cdot \semanticsR{\rewardR}(\rho) + \beta \cdot \semanticsR{\rewardR}(\sigma)
            \end{align*}
            as the trace is linear.
            \item For $\concat$, we have
            \begin{align*}
                &\text{ }\semanticsR{\concat}(\alpha \rho + \beta \sigma) \\
                =& \text{ }\semanticsR{S_2}(\semanticsD{S_1}(\alpha \rho + \beta \sigma)) + \semanticsR{S_1}(\alpha \rho + \beta \sigma) \\
                =& \text{ }\semanticsR{S_2}(\alpha \semanticsD{S_1}(\rho) + \beta \semanticsD{S_1}( \sigma)) + \semanticsR{S_1}(\alpha \rho + \beta \sigma) \\
                =& \text{ }\alpha\semanticsR{S_2}( \semanticsD{S_1}(\rho)) + \beta \semanticsR{S_2}(\semanticsD{S_1}( \sigma)) + \alpha \semanticsR{S_1}(\rho) + \beta \semanticsR{S_1}(\sigma) \\
                =& \text{ }\alpha \cdot \semanticsR{S_2}(\semanticsD{S_1}(\rho),\semanticsR{S_1}(\rho)) + \beta \cdot  \semanticsR{S_2}(\semanticsD{S_1}( \sigma),\semanticsR{S_1}(\sigma)) \\
                =& \text{ }\alpha \cdot \semanticsR{\concat}(\rho) + \beta \cdot \semanticsR{\concat}(\sigma)
            \end{align*}
            \item For $\measure$, we have
            \begin{align*}
                &\semanticsR{\measure}(\alpha \rho + \beta \sigma) \\
                =& \sum_m \semanticsR{S_m}(M_m (\alpha \rho + \beta \sigma) M_m^\dagger) \\
                =& \sum_m \semanticsR{S_m}(\alpha M_m \rho M_m^\dagger + \beta M_m \sigma M_m^\dagger) \\
                =& \sum_m [\alpha\semanticsR{S_m}(M_m \rho M_m^\dagger) + \beta\semanticsR{S_m}(M_m \sigma M_m^\dagger)]\\
                =&\text{ } \alpha\sum_m\semanticsR{S_m}(M_m \rho M_m^\dagger) + \beta\sum_m\semanticsR{S_m}(M_m \sigma M_m^\dagger)\\
                =& \text{ } \alpha\cdot\semanticsR{\measure}(\rho) + \beta\cdot\semanticsR{\measure}(\sigma)
            \end{align*}
            \item For $\while$, we have
            \begin{align*}
                &\semanticsR{\while}(\alpha \rho + \beta \sigma) \\
                =& \lim_{i \to \infty} \semanticsR{\while^i}(\alpha \rho + \beta \sigma) \\
                =& \lim_{i \to \infty} \alpha\semanticsR{\while^i}(\rho) + \beta\semanticsR{\while^i}(\sigma) \\
                =& \text{ }\alpha\cdot\semanticsR{\while}(\rho) + \beta\cdot\semanticsR{\while}(\sigma)
            \end{align*}
            by induction hypothesis (for $\while^i$) and the fact that limits preserve linearity for positive sequences.
            \end{itemize}
            \item $a[\ket{\psi}] := \semanticsR{S}(\ket{\psi} \bra{\psi})$ for all $\ket{\psi} \in \domain{a} = \{\ket{\psi} \in \hilbert \mid \semanticsR{S}(\ket{\psi} \bra{\psi}) < \infty \}$ is a symmetric, positive and closed form. Positivity follows immediately from the previous points and symmetry follows from the co-domain of $\semanticsR{S}$. We show closeness by induction over the structure of $S$:
            \begin{itemize}
                \item For $\skipbf, \qzero$ and $\Uq$ is $a$ exactly the zero form, which is closed.

                \item For $\rewardR$, we distinguish between $c < \infty$ and $c = \infty$. If $c < \infty$, then $a$ is defined on $\hilbert$ and is closed (e.g., by Proposition \ref{prop:OpstoForms}as it corresponds to the self-adjoint operator $c\identityOp$).
                If $c = \infty$, then $a$ is the infinity form that is zero on the zero vector and $\infty$ everywhere else, which is also closed.

                \item For $\concat$, let $a_1$ and $a_2$ be the (closed, positive, symmetric) forms corresponding to $S_1$ and $S_2$ respectively. We know that $\semanticsD{S_1}$ has a Kraus representation by the previous points and thus we can write $\semanticsD{S_1}(\ket{\psi}\bra{\psi}) = \sum_i E_i \ket{\psi}\bra{\psi} E_i^\dagger$. Then we have $a[\ket{\psi}] = \semanticsR{S_1}(\ket{\psi}\bra{\psi}) + \semanticsR{S_2}(\semanticsD{S_1}(\ket{\psi}\bra{\psi})) = a_1[\ket{\psi}] + \sum_i a_2[E_i \ket{\psi}]$. Then $a$ is closed by Proposition \ref{prop:FormAddMult} and Lemma \ref{lem:infSumExpect}.

                \item For $\measure$, let $a_m$ be the (closed, positive, symmetric) form corresponding to $S_m$ for all $m$. We have $a[\ket{\psi}] = \sum_m a_m[M_m \ket{\psi}]$. Then $a$ is closed by Proposition \ref{prop:FormAddMult} and Lemma \ref{lem:infSumExpect}.

                \item For $\while$, we have $a[\ket{\psi}] = \lim_{i \to \infty} a_i[\ket{\psi}]$ where $a_i$ is the form corresponding to $\while^i$.
                We can write $a_i$ as a sum of forms corresponding to sequences of program statements (by induction hypothesis)without loops and thus we can write $a_i$ as a least upper bound of an increasing sequence of closed forms by Proposition \ref{prop:OpstoForms} and Theorem \ref{thm:omegaCompleteness}. \qedhere
            \end{itemize}
    \end{enumerate}
\end{proof}

\section{Details of Section \ref{sec:wp} - Weakest Pre-Expectations}
\label{sec:appendix_wp}

Proof of Proposition \ref{prop:hoareProperties}:
\begin{proof} \text{  }
    \begin{enumerate}
        \item By definition and because $\semanticsD{S}$ is trace non-increasing (thus the term in square brackets is non-negative) and for reward-free programs is $\semanticsR{S}(\rho)= 0$.
        \item $\expect{\zeroOp}{\rho} = 0 \leq \expect{A}{\semanticsD{S}(\rho)} + \semanticsR{S}(\rho)$ for all $\rho$.
        \item \begin{align*}
            &\expect{A}{\rho} \leq \expect{\identityOp}{\rho} = tr(\rho) \\
            \Rightarrow \text{ }&\expect{A}{\rho} \leq tr(\rho) = tr(\semanticsD{S}(\rho)) + [tr(\rho) - tr(\semanticsD{S}(\rho))] \\
            \Rightarrow \text{ }&\expect{A}{\rho} \leq \expect{\identityOp}{\semanticsD{S}(\rho)} + [tr(\rho) - tr(\semanticsD{S}(\rho))]\qedhere
        \end{align*}
    \end{enumerate}
\end{proof}

Proof of Proposition \ref{prop:wpProperties}:
\begin{proof}
    \begin{itemize}
        \item It is $\expect{\zeroOpVar{\hilbert}}{\semanticsD{S}(\rho)} = 0 = \expect{\qwp{S}{\zeroOpVar{\hilbert}}}{\rho}$ for all $\rho$, so $\qwp{S}{\zeroOpVar{\hilbert}} = \zeroOpVar{\hilbert}$ as $\semanticsR{S}(\rho) = 0$ for reward-free programs.

        \item Follows by the linearity of expectation values (Proposition \ref{prop:expectationProperties}):
        \begin{align*}
            \expect{\qwp{S}{c B}}{\rho} &= \expect{c B}{\semanticsD{S}(\rho)} = c \cdot \expect{B}{\semanticsD{S}(\rho)} \\
            &= c \cdot \expect{\qwp{S}{B}}{\rho} = \expect{c \cdot \qwp{S}{B}}{\rho} \\
            \expect{\qwp{S}{B_1 \sumform B_2}}{\rho} &= \expect{B_1 \sumform B_2}{\semanticsD{S}(\rho)} = \expect{B_1}{\semanticsD{S}(\rho)} + \expect{B_2}{\semanticsD{S}(\rho)} \\
            &= \expect{\qwp{S}{B_1}}{\rho} + \expect{\qwp{S}{B_2}}{\rho} \\
            & = \expect{\qwp{S}{B_1} \sumform \qwp{S}{B_2}}{\rho}
        \end{align*}
        This holds for all $\rho \in \density$ and we conclude that $\qwp{S}{c B} = c \cdot \qwp{S}{B}$ and $\qwp{S}{B_1 \sumform B_2} = \qwp{S}{B_1} \sumform \qwp{S}{B_2}$ by Lemma \ref{lem:orderExpectation} and because $\loewner$ is a partial order and thus in particular antisymmetric. Note that this property does not hold for programs with rewards as the reward term $\semanticsR{S}(\rho)$ is not $0$ then.

        \item Follows by Lemma \ref{lem:orderExpectation}:
        \begin{align*}
            &B_1 \loewner B_2 \\
            \Rightarrow &\text{ } \forall \rho \in \density: \expect{B_1}{\semanticsD{S}(\rho)} \leq \expect{B_2}{\semanticsD{S}(\rho)} \\
            \Rightarrow & \text{ }\forall \rho \in \density: \expect{\qwp{S}{B_1}}{\rho} + \semanticsR{S}(\rho) \leq \expect{\qwp{S}{B_2}}{\rho} + \semanticsR{S}(\rho)\\
            \Rightarrow & \text{ }\forall \rho \in \density: \expect{\qwp{S}{B_1}}{\rho}  \leq \expect{\qwp{S}{B_2}}{\rho}\\
            \Rightarrow & \text{ }\qwp{S}{B_1} \loewner \qwp{S}{B_2}
        \end{align*}

        \item The existence of both suprema follows by Theorem \ref{thm:omegaCompleteness} and the monotonicity. The equivalence follows by continuity of expectation values (Proposition \ref{prop:expectationProperties}):
        \begin{align*}
            \expect{\qwp{S}{\bigvee_n B_n}}{\rho} &= \expect{\bigvee_n B_n}{\semanticsD{S}(\rho)} + \semanticsR{S}(\rho) \\
            &= \bigvee_n \expect{B_n}{\semanticsD{S}(\rho)} + \semanticsR{S}(\rho) \\
            &= \bigvee_n \expect{\qwp{S}{B_n}}{\rho} = \expect{\bigvee_n \qwp{S}{B_n}}{\rho}
        \end{align*}
        and the same argument as given in the proof of linearity. \qedhere
    \end{itemize}
\end{proof}

Proof of Theorem \ref{thm:existenceFormsWp}:
\begin{proof}
    In the following we will use $a$ (and later $A$) for the precondition and $b$ (and $B$) for the postcondition.

    By Definition \ref{def:generalSemantics}, we can use the Kraus representation of $\semanticsD{S}$ and write $\semanticsD{S}(\rho) = \sum_i E_i \rho E_i^\dagger$ as a sum for some bounded operators $E_i$. They are not unique, but for now we fix one family $E_i$ (later we will show that it does not matter which representation we choose).
    Then we define
        \begin{align*}
            a_i &= \lambda \psi. b_{E_i}[\ket{\psi}]+ \semanticsR{S}(\ket{\psi}\bra{\psi})\\
            \hat{a_j} &= \sum_{i\leq j} a_i\\
            a &= \bigvee_j \hat{a_j}
        \end{align*}
    As $b\in \formPredicate$, $b_{E_i} \in \formPredicate$ by Proposition \ref{prop:FormAddMult} as it is just applying a bounded operator $E_i$. $\semanticsR{S}(\ket{\psi}\bra{\psi})$ corresponds also to a form-expectation by Definition \ref{def:generalSemantics}. Thus $a_i$ (as the sum of two form-expectations) is a form-expectation by Proposition \ref{prop:FormAddMult}.

    $\hat{a_j}$ is a finite sum, we can apply Proposition \ref{prop:FormAddMult} again and obtain that $\hat{a_j} \in \formPredicate$ for all $j$ again.
    As $a_i$ is positive for all $i$ , $\hat{a_j}$ is increasing in $j$.

    Then we can also consider the corresponding expectations $\hat{A}_j$ for each $j$. They form an increasing chain of expectations by Proposition\ref{prop:orderEquivalence}, thus the corresponding least upper bound $\bigvee_j \hat{A}_j$ exists by Theorem \ref{thm:omegaCompleteness} and is an expectation and $a$ is its corresponding form, so $a \in \formPredicate$.

    Now we show that $\qwp{S}{B}$ is the corresponding expectation of $a = \qwpForm{S}{b}$. By the above part, we know that there exists a corresponding expectation $A$ to $a$.

    First we show that $\expect{B}{\semantics{S}(\ket{\psi}\bra{\psi})} + \semanticsR{S}(\ket{\psi}\bra{\psi})= \expect{A}{\ket{\psi}\bra{\psi}} $ holds for all $\ket{\psi}$.
    Let $\ket{\psi}\in \hilbert$. Then $\semantics{S}(\ket{\psi}\bra{\psi}) = \sum_i E_i \ket{\psi}\bra{\psi} E_i^\dagger$.

    If there is at least one $E_i \ket{\psi} \not \in \domain{B}$ then $\ket{\psi} \not \in \domain{a}$ and both sides are $\infty$.

    If all $E_i \ket{\psi} \in \domain{B}$, then
    \begin{align*}
        &\text{ } \expect{B}{\semantics{S}(\ket{\psi}\bra{\psi})} + \semanticsR{S}(\ket{\psi}\bra{\psi})\\
        =& \text{ }\expect{B}{\sum_i E_i \ket{\psi} \bra{\psi}E_i^\dagger} + \semanticsR{S}(\ket{\psi}\bra{\psi})\\
        =& \sum_i \expect{B}{E_i \ket{\psi} \bra{\psi}E_i^\dagger} + \semanticsR{S}(\ket{\psi}\bra{\psi})\\
        =& \sum_i b[E_i \ket{\psi}]+ \semanticsR{S}(\ket{\psi}\bra{\psi})\\
        & = \sum_i a_i[\ket{\psi}] = a[\ket{\psi}] =\expect{A}{\ket{\psi}\bra{\psi}}
    \end{align*}

    We can conclude that $\expect{A}{\rho} = \expect{B}{\semanticsD{S}(\rho)} + \semanticsR{S}(\rho)$ for all $\rho = \sum_i p_i \ket{\psi_i}\bra{\psi_i}$: \begin{align*}
        \expect{B}{\semanticsD{S}(\rho)} + \semanticsR{S}(\rho) &= \expect{B}{\semanticsD{S}(\sum_i p_i \ket{\psi_i}\bra{\psi_i})} +  \semanticsR{S}(\sum_i p_i\ket{\psi_i}\bra{\psi_i})\\
        &= \expect{B}{\sum_i p_i \semanticsD{S}(\ket{\psi_i}\bra{\psi_i})}+ \sum_i p_i \semanticsR{S}(\ket{\psi_i}\bra{\psi_i})\\
        &= \sum_i p_i \expect{B}{\semanticsD{S}(\ket{\psi_i}\bra{\psi_i})} + \sum_i p_i \semanticsR{S}(\ket{\psi_i}\bra{\psi_i})\\
        &= \sum_i p_i \left(\expect{B}{\semanticsD{S}(\ket{\psi_i}\bra{\psi_i})} + \semanticsR{S}(\ket{\psi_i}\bra{\psi_i})\right)\\
        &= \sum_i p_i \expect{A}{\ket{\psi_i}\bra{\psi_i}} =  \expect{A}{\sum_i p_i \ket{\psi_i}\bra{\psi_i}} \\
        &= \expect{A}{\rho}
    \end{align*} as $\semanticsD{S}, \semanticsR{S}$ and $\mathbb{E}$ are linear.

    In the beginning, we fixed one Kraus representation $E_i$. Now lets assume there is another one $F_j$ with $\semanticsD{S}^F(\rho) := \sum_j F_j \rho F_j^\dagger =\semanticsD{S}(\rho)= \sum_i E_i \rho E_i^\dagger =: \semanticsD{S}^E(\rho)$ for all $\rho$.
    Let $A_F$ and $A_E$ be the corresponding expectations to $\qwpForm{S_F}{a}$ and $\qwpForm{S_E}{a}$ respectively.
    Then for all $\rho$
    \begin{align*}
        \expect{A_F}{\rho} = \expect{B}{\semanticsD{S}^F(\rho)} + \semanticsR{S}(\rho) = \expect{B}{\semanticsD{S}^E(\rho)} + \semanticsR{S}(\rho) = \expect{A_E}{\rho}.
    \end{align*}
    Together with the fact that $\loewner$ is a partial order and thus antisymmetric by Proposition \ref{prop:orderproperties} and Lemma \ref{lem:orderExpectation}, we conclude $\qwpForm{S_F}{a} = \qwpForm{S_E}{a}$ and thus the definition does not depend on the choice of Kraus operators.
\end{proof}

Proof of Proposition \ref{prop:wpDefinition}:
\begin{proof}
    We have to check two things:
    \begin{enumerate}
        \item $\qwp{S}{B}$ is an expectation, and
        \item $\expect{\qwp{S}{B}}{\rho} = \expect{B}{\semanticsD{S}(\rho)} + \semanticsR{S}(\rho)$ for all $\rho \in \density$.
    \end{enumerate}
    We show both statements together by an induction on the structure of $S$.
    Let $\rho = \sum_i p_i \ket{\psi_i}\bra{\psi_i} \in \density$ with $p_i>0$ and $B\in \predicate$ with corresponding form $b \in \formPredicate$.
    \begin{itemize}
        \item $S = \skipbf$: Trivial as $\qwp{\skipbf}{B} = B$ and $\semanticsD{\skipbf}(\rho)=\rho$, $\semanticsR{\skipbf}(\rho) = 0$ for all $\rho \in \density$.
        \item $S = \qzero$:
        \begin{enumerate}
        \item $\qwp{\qzero}{B} \in \predicate$ by Definition \ref{def:FormAddMult} as $\ket{n}\bra{0}$ are bounded operators. The sum is infinite and by Lemma \ref{lem:infSumExpect} in $\predicate$.

        \item It is \begin{align*}
            \expect{\qwp{\qzero}{B}}{\rho} =& \text{ }\expect{\sum_n \ket{n}\bra{0} \odot B \odot \ket{0}\bra{n}}{\rho} \\
            =& \sum_n \expect{\ket{n}\bra{0} \odot B \odot \ket{0}\bra{n}}{\rho} \\
            =& \sum_n \expect{B}{\ket{0}\bra{n} \rho \ket{n}\bra{0}} \\
            =& \text{ }\expect{B}{\sum_n \ket{0}\bra{n} \rho \ket{n}\bra{0}} \\
            =& \text{ }\expect{B}{\semanticsD{\qzero}(\rho)}
        \end{align*}
        as $\mathbb{E}$ is linear in both the state and the expectation (Lemma \ref{lem:infSumLinExpectation}) and we can move bounded operators from the expectation to the density operator by Proposition \ref{prop:boundedOpinExpectation}. As $\semanticsR{\qzero}(\rho) = 0$ for all $\rho$, we have $\expect{\qwp{\qzero}{B}}{\rho} = \expect{B}{\semanticsD{\qzero}(\rho)} + \semanticsR{\qzero}(\rho)$.

        \end{enumerate}

        \item $S = \Uq$:
        \begin{enumerate}
        \item $\qwp{\Uq}{B}\in \predicate$ by Definition \ref{def:FormAddMult} as $U$ is a bounded operator.
        \item Let $wp$ the corresponding form of $\qwp{\Uq}{B}_\infty$. It is $\domain{wp} = \{\ket{\psi} \in \hilbert \mid U\ket{\psi} \in \domain{b}\}$.
        \begin{enumerate}
            \item If there is (at least) one $\ket{\psi_i}\not \in \domain{wp}$, then $\expect{\qwp{S}{B}}{\rho} = \infty$.
            By the definition of $\domain{wp}$, we have $\expect{B}{\semanticsD{\Uq}(\ket{\psi_i}\bra{\psi_i})} = \infty$.
            Overall it is $\expect{\qwp{\Uq}{B}}{\rho} =\infty = \expect{B}{\semanticsD{\Uq}(\rho)} $ by linearity.
            \item If $\ket{\psi_i} \in \domain{wp}$ for all $i$,
            then \begin{align*}
            \expect{\qwp{\Uq}{B}}{\rho} & = \sum_i p_i wp[\ket{\psi_i}]\\
            &= \sum_i p_i b[U\ket{\psi_i}] \\
            &= \sum_i p_i \expect{B}{\semanticsD{\Uq}(\ket{\psi_i}\bra{\psi_i})} \\
            &= \expect{B}{\semanticsD{\Uq}(\rho)}
            \end{align*}
        \end{enumerate}
        Overall we have $\expect{\qwp{\Uq}{B}}{\rho} = \expect{B}{\semanticsD{\Uq}(\rho)} + \semanticsR{\Uq}(\rho)$ as $\semanticsR{\Uq}(\rho) = 0$ for all $\rho$.
        \end{enumerate}

            \item $S = \rewardR$:
            \begin{enumerate}
            \item $\qwp{\rewardR}{B} = B \sumform c \cdot \identityOp$ is in $\predicate$ by Definition \ref{def:FormAddMult}.
            \item It is $\expect{\qwp{\rewardR}{B}}{\rho} = \expect{B}{\semanticsD{\rewardR}(\rho)} + \semanticsR{\rewardR}(\rho)$ as
            \begin{align*}
                \expect{\qwp{\rewardR}{B}}{\rho} &=
                \expect{B \sumform c \cdot \identityOp}{\rho} \\
                & = \expect{B}{\semanticsD{\rewardR}(\rho)} + \expect{c \cdot \identityOp}{\rho} \\
                &= \expect{B}{\semanticsD{\rewardR}(\rho)} + c \cdot tr(\rho) \\
                &= \expect{B}{\semanticsD{\rewardR}(\rho)} + \semanticsR{\rewardR}(\rho)
                \end{align*}
                by linearity of expectation values (Proposition \ref{prop:expectationProperties}).
            \end{enumerate}

        \item $S = \concat$:
        \begin{enumerate}
        \item $\qwp{S_1}{\qwp{S_2}{B}} \in \predicate$ follows immediately from the induction hypothesis.
        \item \begin{align*}
            \expect{\qwp{\concat}{B}}{\rho} & = \expect{\qwp{S_1}{\qwp{S_2}{B}}}{\rho} \\
            &\overset{IH}{=} \expect{\qwp{S_2}{B}}{\semanticsD{S_1}(\rho)} + \semanticsR{S_1}(\rho) \\
            &\overset{IH}{=} \expect{B}{\semanticsD{S_2}(\semanticsD{S_1}(\rho))} + \semanticsR{S_2}(\semanticsD{S_1}(\rho)) + \semanticsR{S_1}(\rho) \\
            &= \expect{B}{\semanticsD{\concat}(\rho)} + \semanticsR{\concat}(\rho)
        \end{align*}
        \end{enumerate}

        \item $S = \measurePrime$:
        \begin{enumerate}
        \item By induction hypothesis, $\qwp{S_m'}{B} \in \predicate$ for all $m$. Then by Definition \ref{def:FormAddMult}, $\formsandwich{M_m}{\qwp{S_m'}{B}} \in \predicate$ for all $m$ as $M_m$ are bounded operators. We can define
        $\hat{f}_n = \sum_{m\leq n} \formsandwich{M_m}{\qwp{S_m'}{B}}$. As each $\formsandwich{M_m}{\qwp{S_m'}{B}} \in \predicate$, $\hat{f}_n \in \predicate$ for all $n$ as it is a finite sum by Definition \ref{def:FormAddMult}.
        As $\hat{f}_n$ is an increasing chain (each $\formsandwich{M_m}{\qwp{S_m'}{B}}$ is positive), by Theorem \ref{thm:omegaCompleteness}, the least upper bound $\bigvee_n \hat{A}_n$ exists and corresponds to $\sum_m \formsandwich{M_m}{\qwp{S_m'}{B}}$.

        \item Let $wp$ be the form-expectation of $\qwp{\measurePrime}{B}$ and $wp_m$ the one of $\qwp{S_m'}{B}$. It is $\domain{wp} = \bigcap_m \{\ket{\psi}\mid M_m \ket{\psi} \in \domain{wp_m}\}$.
        \begin{enumerate}
            \item If there is an $i$ such that $\ket{\psi_i}\not \in \domain{wp}$, then there must be a $j$ such that $M_j \ket{\psi_i} \not \in \domain{wp_j}$ and thus
            \begin{align*}
                \infty =& \expect{\qwp{S_j'}{B}}{\ket{M_j \psi_i}\bra{M_j \psi_i}} \\
                \overset{IH}{=}&\expect{B}{\semanticsD{S_j'}(\ket{M_j \psi_i}\bra{M_j \psi_i})} + \semanticsR{S_j'}(\ket{M_j \psi_i}\bra{M_j \psi_i})
            \end{align*}
            By linearity of semantics (Proposition \ref{prop:semanticsProperties}) and expectation values (Proposition \ref{prop:expectationProperties}, Lemma \ref{lem:infSumLinExpectation}):
            \begin{align*}
            &\expect{B}{\semanticsD{\measurePrime}(\rho)} + \semanticsR{\measurePrime}(\rho)\\
            =& \infty = \expect{\qwp{\measure}{B}}{\rho}.
            \end{align*}
            \item Assume $\ket{\psi_i} \in \domain{wp}$ for all $i$, then $M_m \ket{\psi_i} \in \domain{wp_m}$ for all $i,m$.
            Then
            \begin{align*}
                &\expect{\qwp{\measure}{B}}{\rho} \\
                 =& \sum_{m} \expect{\formsandwich{M_m}{\qwp{S_m'}{B}}}{\rho} \\
                = & \sum_{m} \sum_i p_i wp_m[M_m \ket{\psi_i}] \\
                \overset{IH}{=} &\sum_{m} \sum_i p_i \expect{B}{\semanticsD{S_m'}(\ket{M_m \psi_i}\bra{M_m \psi_i})} + \semanticsR{S_m'}(\ket{M_j \psi_i}\bra{M_j \psi_i}) \\
                 = &\sum_i p_i \sum_{m}  \expect{B}{\semanticsD{S_m'}(\ket{M_m \psi_i}\bra{M_m \psi_i})} + \semanticsR{S_m'}(\ket{M_j \psi_i}\bra{M_j \psi_i})\\
                 =& \sum_i p_i \expect{B}{\semanticsD{\measurePrime}(\ket{\psi_i}\bra{\psi_i})} \\
                & +\semanticsR{\measurePrime}(\ket{\psi_i}\bra{\psi_i}) \\
                =& \expect{B}{\semanticsD{\measurePrime}(\rho)} + \semanticsR{\measurePrime}(\rho)
            \end{align*}
              As $\expect{B}{\ket{\psi}\bra{\psi}}, \semanticsR{S}(\ket{\psi}\bra{\psi})$ and $p_i$ are non-negative for all $\ket{\psi}, S$, we can switch the sums in the forth step due to Fubini's theorem for infinite series. We also use linearity of semantics (Proposition \ref{prop:semanticsProperties}) and expectation values (Proposition \ref{prop:expectationProperties}, Lemma \ref{lem:infSumLinExpectation}).
        \end{enumerate}

        \end{enumerate}

        \item $S = \while'$: In this case we prove both parts together. First of all, we show $\qwp{S'}{A_n} \in \predicate$ for all $n$.
        Each $A_n \in \predicate$ because it is defined via sums and applications of bounded operators on expectations (see Definition \ref{def:FormAddMult}). Now it reminds to show that $\{A_n\}_n$ is an increasing chain, then the least upper bound $\bigvee_n A_n$ exists by Theorem \ref{thm:omegaCompleteness} and is in $\predicate$.

        First we show
        \begin{align} \label{align:expectAn}
            \expect{A_n}{\rho} = \expect{B}{\semanticsD{\while^n}(\rho)} + \semanticsR{\while^n}(\rho)
        \end{align}
        by induction on $n$:
        \begin{itemize}
            \item $n=0$:
            \begin{align*}
                \expect{A_0}{\rho} =& \text{ }\expect{\zeroOp}{\rho} = 0 \\
                =& \text{ }\expect{B}{\semanticsD{\Omega}(\rho)} + \semanticsR{\Omega}(\rho) \\
                =& \text{ }\expect{B}{\semantics{\while^0}(\rho)} + \semanticsR{\while^0}(\rho)
            \end{align*}
            \item $n \to n+1$:
            Let $wp_n$ be the form corresponding to $\qwp{S'}{A_n}$, $b$ to $B_\infty$ and $a_{n}$ the form corresponding to ${A_n}_\infty$.
            Then $\domain{a_{n+1}} = \{\psi \mid M_0 \psi \in \domain{b}\} \cap \{\psi \mid M_1 \psi \in \domain{wp_n}\}$.
            Let $\rho = \sum_i p_i \ket{\psi_i}\bra{\psi_i}$ be a decomposition of $\rho$ into pure states. We show the statement for each $\ket{\psi_i}\bra{\psi_i}$, then we can conclude it for $\rho$ by linearity of expectation values and the semantics (Proposition \ref{prop:expectationProperties}, Lemma \ref{lem:infSumLinExpectation} and Proposition \ref{prop:semanticsProperties}).
            We distinguish two cases:
            \begin{enumerate}
                \item Assume there is an $i$ such that $\ket{\psi_i} \not \in \domain{a_{n+1}}$, then $\expect{A_{n+1}}{\rho} = \infty$.

                If $M_0 \ket{\psi_i} \not \in \domain{b}$, then $\expect{B}{M_0 \ket{\psi_i}\bra{\psi_i} M_0^\dagger} = \infty$ and thus by linearity of expectation values and semantics (Proposition \ref{prop:expectationProperties}, Lemma \ref{lem:infSumLinExpectation}, Proposition \ref{prop:semanticsProperties}), we get $\expect{B}{\semanticsD{\while^{n+1}}(\rho)} = \infty$, thus the sum with $\semanticsR{\while^{n+1}}(\rho)$ is $\infty$ as well.

                If $M_1 \ket{\psi_i} \not \in \domain{wp_n}$, then
                \begin{align*}
                    \infty &= \expect{\qwp{S'}{A_n}}{\ket{M_1 \psi_i}\bra{M_1 \psi_i}} \\
                    & = \expect{A_n}{\semanticsD{S'}(\ket{M_1 \psi_i}\bra{M_1 \psi_i})} \\
                    &= \expect{B}{\semanticsD{\while^n}(\semanticsD{S'}(\ket{M_1 \psi_i}\bra{M_1 \psi_i}))}
                \end{align*}
                By linearity (Proposition \ref{prop:expectationProperties}, Lemma \ref{lem:infSumLinExpectation} and Proposition \ref{prop:semanticsProperties}), we have \\
                $\expect{B}{\semanticsD{\while^{n+1}}(\rho)} = \infty$ and thus the same for the addition with $\semanticsR{\while^{n+1}}(\rho)$.

                \item Assume $\ket{\psi_i} \in \domain{a_{n+1}}$ for all $i$.
                \begin{align*}
                    \expect{A_{n+1}}{\ket{\psi_i}\bra{\psi_i}} =& \text{ } \expect{\formsandwich{M_0}{B}}{\ket{\psi_i}\bra{\psi_i}} + \expect{\formsandwich{M_1}{\qwp{S'}{A_n}}}{\ket{\psi_i}\bra{\psi_i}} \\
                    =&\text{ } b[M_0 \ket{\psi_i}] + wp_n[M_1 \ket{\psi_i}] \\
                    =&\text{ } \expect{B}{\ket{M_0 \psi_i}\bra{M_0 \psi_i}} + \expect{\qwp{S'}{A_n}}{\ket{M_1 \psi_i}\bra{M_1 \psi_i}} \\
                    =&\text{ } \expect{B}{\ket{M_0 \psi_i}\bra{M_0 \psi_i}} \\
                    &+ \expect{A_n}{\semanticsD{S'}(\ket{M_1 \psi_i}\bra{M_1 \psi_i})} + \semanticsR{S'}(\ket{M_1 \psi_i} \bra{M_1 \psi_i})\\
                    =&\text{ } \expect{B}{\ket{M_0 \psi_i}\bra{M_0 \psi_i}} \\
                    &+ \expect{B}{\semanticsD{\while^n}(\semanticsD{S'}(\ket{M_1 \psi_i}\bra{M_1 \psi_i}))} \\
                    &+ \semanticsR{\while^n}(\semanticsD{S'}(\ket{M_1 \psi_i}\bra{M_1 \psi_i})) \\
                    & + \semanticsR{S'}(\ket{M_1 \psi_i} \bra{M_1 \psi_i}) \\
                    =&\text{ } \expect{B}{\semanticsD{\skipbf}(\ket{M_0 \psi_i}\bra{M_0 \psi_i})} \\
                    &+ \expect{B}{\semanticsD{S'; \while^n}(\ket{M_1 \psi_i}\bra{M_1 \psi_i})} \\
                    &+ \semanticsR{\skipbf}(\ket{M_0 \psi_i}\bra{M_0 \psi_i}) \\
                    &+ \semanticsR{S'; \while^n}(\ket{M_1 \psi_i}\bra{M_1 \psi_i}) \\
                    =&\text{ } \expect{B}{\semanticsD{\while^n}(\ket{\psi_i}\bra{ \psi_i})} \\
                    &+ \semanticsR{\while^n}(\ket{ \psi_i}\bra{ \psi_i})
                \end{align*}
            \end{enumerate}

        Now we can show that $\{A_n\}_n$ is an increasing chain:
        \begin{itemize}
            \item $n=0$: $A_0 = \zeroOp \loewner A_1$ and $\domain{A_1} \subseteq \hilbert =\domain{\zeroOp}$
            \item $n \to n+1$:  $A_{n+1} = \formsandwich{M_0}{B} + \formsandwich{M_1}{\qwp{S'}{A_n}}$
            \begin{align*}
                & \text{ }A_{n-1} \loewner A_n \\
                \Rightarrow & \text{ }\qwp{S'}{A_{n-1}} \loewner \qwp{S'}{A_n} \\
                \Rightarrow & \text{ } \formsandwich{M_1}{\qwp{S'}{A_{n-1}}} \loewner \formsandwich{M_1}{\qwp{S'}{A_n}} \\
                \Rightarrow & \text{ } \formsandwich{M_0}{B} + \formsandwich{M_1}{\qwp{S'}{A_{n-1}}}\\
                & \loewner \formsandwich{M_0}{B} + \formsandwich{M_1}{\qwp{S'}{A_n}} \\
                \Rightarrow & \text{ } A_n \loewner A_{n+1}
            \end{align*}
        \end{itemize} The first step holds by the monotonicity of $wp$ (shown in Proposition \ref{prop:wpProperties}).
        The next two steps hold by Proposition \ref{prop:orderproperties}. Thus we have shown $\bigvee_n A_n \in \predicate$.

        Finally, we can conclude the second part of the proof based on the previous results:
        \begin{align*}
            & \expect{\qwp{\while}{B}}{\rho} \\
            & = \expect{\bigvee_n A_n}{\rho}\\
             &= \bigvee_n \expect{A_n}{\rho} \\
             &\overset{Eq. \ref{align:expectAn}}{=} \bigvee_n [\expect{B}{\semanticsD{\while^n}(\rho)} + \semanticsR{\while^n}(\rho)]\\
             & = \bigvee_n \expect{B}{\semanticsD{\while^n}(\rho)} + \lim_n \semanticsR{\while^n}(\rho)\\
             &= \expect{B}{\semanticsD{\while}(\rho)} + \semanticsR{\while}(\rho)
        \end{align*}
        because both $\expect{B}{\semanticsD{\while^n}(\rho)}$ and $\semanticsR{\while^n}(\rho)$ are increasing in $n$. \qedhere
    \end{itemize}
    \end{itemize}
\end{proof}

Proof of Proposition \ref{prop:wpBoundedEquiv}:
\begin{proof}\text{ }
    \begin{itemize}
        \item First of all note that if $A \loewner \identityOp$ is bounded, then $\expect{A}{\rho} = tr(A\rho)$ for all partial density operators $\rho \in \density$ as $\domain{A} = \hilbert$ for bounded operators.
        \item Then the definition of total and partial correctness are exactly the same as in \cite{floydHoareLogic} for reward-free programs $S$ as $\semanticsR{S}(\rho) = 0$ for all $\rho$.
        \item The definition of weakest pre-expectation is also the same as in \cite{DHondtWeakestPreconditions}.
        \item Finally, the concrete representation of the weakest pre-expectation transformer $wp$ differs only by the addition and multiplication but for bounded operators they coincide by Proposition \ref{prop:BoundedAddition}. \qedhere
    \end{itemize}
\end{proof}

Proof of Proposition \ref{prop:park}:
\begin{proof}
    As a technical result, we can show that the weakest pre-expectation of a loop is equivalent to the least upper bound of its characteristic function:
    Let $\while$ be a loop in qrWhile and $B \in \predicate$ be a post-expectation.
    Then $\qwp{\while}{B} = \bigvee_n \psi_B^n(0)$.

    To show this equivalence, we show $\psi_B^n(0) = A_n$ for all $n$ where $A_n$ is defined as in Table \ref{tab:wp}:
    \begin{itemize}
        \item $n=0$: $\psi_B^0(0) = 0 = A_0$
        \item $n \to n+1$: \\
        \begin{align*}
            \psi_B^{n+1}(0) &= \psi_B(\psi_B^n(0)) \\
            &= \formsandwich{M_0}{B} \sumform \formsandwich{M_1}{\qwp{S}{\psi_B^n(0)}} \\
            &= \formsandwich{M_0}{B} \sumform \formsandwich{M_1}{\qwp{S}{A_n}} \\
            &= A_{n+1}
        \end{align*}
    \end{itemize}

    Now we continue by showing that $\psi_B$ is monotone: Let $A_1, A_2 \in \predicate$ with $A_1 \loewner A_2$, then
    \begin{align*}
       A_1 &\loewner A_2 \\
        \overset{Prop. \ref{prop:wpProperties}}{\Rightarrow} \qwp{S}{A_1} &\loewner \qwp{S}{A_2} \\
        \overset{Prop. \ref{prop:orderproperties}}{\Rightarrow} \formsandwich{M_1}{\qwp{S}{A_1}} &\loewner \formsandwich{M_1}{\qwp{S}{A_2}} \\
        \overset{Prop. \ref{prop:orderproperties}}{\Rightarrow} \formsandwich{M_0}{B} \sumform \formsandwich{M_1}{\qwp{S}{A_1}} &\loewner \formsandwich{M_0}{B} \sumform \formsandwich{M_1}{\qwp{S}{A_2}} \\
        \Rightarrow \psi_B(A_1) &\loewner \psi_B(A_2)
    \end{align*}

    We can show that $\psi_B(A)\loewner A$ implies $\psi^n_B(0) \loewner A$ for all $n\in \mathbb{N}$ by induction over $n$:
    \begin{itemize}
        \item $n=0$: $\psi^0_B(0) = 0 \loewner A$ as $A$ is positive.
        \item $n \to n+1$:
        \begin{align*}
            \psi_B(A) &\loewner A \\
            \overset{IH}{\Rightarrow} \psi^n_B(0) &\loewner A \\
            \Rightarrow \psi_B(\psi^n_B(0)) &\loewner \psi_B(A) \\
            \Rightarrow \psi^{n+1}_B(0) &\loewner \psi_B(A) \\
            \Rightarrow \psi^{n+1}_B(0) &\loewner A
        \end{align*}
        where the second step holds by monotonicity of $\psi_B$ and the last step holds by the assumption $\psi_B(A) \loewner A$ and the transitivity of $\loewner$ (Proposition \ref{prop:orderproperties}).
    \end{itemize}

     Finally, we can show that whole statement:
    \begin{align*}
        \psi_B(A)\loewner A &\Rightarrow \forall n\in \mathbb{N}: \psi^n_B(0) \loewner A \\
        & \Rightarrow \bigvee_n \psi^n_B(0) \loewner A \\
        & \Rightarrow \qwp{\while}{B} \loewner A
    \end{align*}
    where the first step holds by the above induction, the second one by the definition of $\bigvee$ and the last one by the technical result at the beginning of the proof.
\end{proof}

The proof of Proposition \ref{prop:parkWLP} works completely analogous.

\section{Details of Section \ref{sec:ertRewards} - Expected Runtime using Rewards}
\label{sec:appendix_ert_rewards}

Proof of Proposition \ref{prop:AstPastUsingWp}:
\begin{proof}
    \begin{itemize}
        \item $\forall \rho \in \density$ is $\semanticsR{S}(\rho) = 0$, thus
\begin{align*}
    tr(\semanticsD{S}(\rho)) = tr(\identityOp \semanticsD{S}(\rho)) = tr(\qwp{S}{\identityOp} \rho) = tr(\identityOp \rho) = tr(\rho).
\end{align*}
        \item $\domain{\qwp{\transform{S}}{\zeroOp}} = \hilbert$ means that $\expect{\qwp{\transform{S}}{\zeroOp}}{\rho} < \infty$ for all $\rho \in \density$, thus it is $\semanticsR{\transform{S}}(\rho) < \infty$ for all $\rho \in \density$ and thus $S$ is PAST. Conversely, if $S$ is PAST, then $\semanticsR{\transform{S}}(\rho) < \infty$ for all $\rho \in \density$, thus $\expect{\qwp{\transform{S}}{\zeroOp}}{\rho} < \infty$ for all $\rho \in \density$ and thus $\domain{\qwp{\transform{S}}{\zeroOp}} = \hilbert$. \qedhere
    \end{itemize}
\end{proof}

\section{Details of Section \ref{sec:ertCalculus} - ERT Calculus}
\label{sec:appendix_ert_calculus}

Proof of Proposition \ref{prop:ERTProperties}:
\begin{proof}
Let $S$ be a reward-free qrWhile program, $A,B \in \predicate$ and $c\geq 0$.
\begin{enumerate}
    \item[1. - 3.] Well-definedness \& Connection to wp \& Monotonicity:
    We show this together by induction over the structure of $S$. The connection to wp is shown by induction over the structure of $S$ as well, and monotonicity follows immediately from the connection to wp and the monotonicity of wp:
    Let $A \loewner B$. We show $\ert{S}{A} \loewner \ert{S}{B}$ based on 2. and the monotonicity of wp (Proposition \ref{prop:wpProperties}):
    \begin{align*}
        &\ert{S}{A} \\
        &\overset{2.}{=} \qwp{S}{A} \sumform \ert{S}{\zeroOp} \\
        &\loewner \qwp{S}{B} \sumform \ert{S}{\zeroOp} \\
        &\overset{2.}{=} \ert{S}{B}
    \end{align*}
    Induction over the structure of $S$:
    \begin{itemize}
        \item $S = \skipbf$: Clearly $\ert{\skipbf}{A} \in \predicate$.
        \begin{align*}
            \ert{\skipbf}{A} = A = A + \zeroOp = \qwp{\skipbf}{A} \sumform \ert{\skipbf}{\zeroOp}
        \end{align*}

        \item $S = \qzero$: It is $\ert{\qzero}{A} \in \predicate$ by Definition \ref{def:FormAddMult} as $\ket{n}\bra{0}$ are bounded operators. The sum is infinite and by Lemma \ref{lem:infSumExpect} in $\predicate$.
        \begin{align*}
            &\ert{\qzero}{A} \\
            &= \identityOp \sumform \sum_{n=0}^\infty \ket{n}\bra{0} \odot A \odot \ket{0}\bra{n} \\
            &= \identityOp \sumform \qwp{\qzero}{A} \\
            &= \qwp{\qzero}{A} \sumform \ert{\qzero}{\zeroOp}
        \end{align*}

        \item $S = \Uq$: It is $\ert{\Uq}{A} \in \predicate$ as $\predicate$ is closed under $\sumform$ and bounded-operator application.
        \begin{align*}
            \ert{\Uq}{A} = \formsandwich{U}{A} \sumform \identityOp = \qwp{\Uq}{A} \sumform \ert{\Uq}{\zeroOp}
        \end{align*}

        \item $S = \concat$:
        It is $\ert{\concat}{A}\in \predicate$ by induction hypothesis.
        \begin{align*}
            \ert{\concat}{A} &= \ert{S_1}{\ert{S_2}{A}} \\
            &\overset{IH}{=} \ert{S_1}{\qwp{S_2}{A} \sumform \ert{S_2}{\zeroOp}} \\
            &\overset{IH}{=} \qwp{S_1}{\qwp{S_2}{A} \sumform \ert{S_2}{\zeroOp}} \sumform \ert{S_1}{\zeroOp} \\
            &= \qwp{S_1}{\qwp{S_2}{A}} \sumform \qwp{S_1}{\ert{S_2}{\zeroOp}} \sumform \ert{S_1}{\zeroOp} \\
            &\overset{IH}{=}\qwp{S_1}{\qwp{S_2}{A}} \sumform \ert{S_1}{\ert{S_2}{\zeroOp}}\\
            &= \qwp{\concat}{A} \sumform \ert{ \concat}{\zeroOp}
        \end{align*}
        \item $S = \measurePrime$:

        To show $\ert{\measurePrime}{A} \in \predicate$:
        It is $\formsandwich{M_m}{\ert{S_m'}{A}} \in \predicate$ for all $m$ as $\predicate$ is closed under bounded-operator application and $\ert{S_m'}{A} \in \predicate$ by induction hypothesis.
        The sum is infinite and in $\predicate$ by Lemma \ref{lem:infSumExpect}. Adding $\identityOp$ is also fine as $\predicate$ is closed under $\sumform$.
        \begin{align*}
            &\ert{\measurePrime}{A} \\
            &= \identityOp \sumform \left(\sum_{m} \formsandwich{M_m}{\ert{S_m'}{A}} \right) \\
            &\overset{IH}{=} \identityOp \sumform \left(\sum_{m} \formsandwich{M_m}{\left(\qwp{S_m'}{A} \sumform \ert{S_m'}{\zeroOp}\right)} \right) \\
            &= \identityOp \sumform \left(\sum_{m} \formsandwich{M_m}{\qwp{S_m'}{A}} \right) \sumform \left(\sum_{m} \formsandwich{M_m}{\ert{S_m'}{\zeroOp}} \right) \\
            &= \qwp{\measurePrime}{A} \sumform \ert{\measurePrime}{\zeroOp}
        \end{align*}
        Distributivity for infinite sums holds by Lemma \ref{lem:distrInfSum} (and for finite ones by Proposition \ref{prop:ContinuityDistributivityFormSum}).
        \item $S = \while'$: We consider three different chains:
        \begin{align*}
            \ert{\while'}{A} &= \bigvee_n F_n \quad \text{where }\\
            F_0 &= \zeroOp,\\
            F_{n+1} &= \formsandwich{M_0}{A} \sumform \formsandwich{M_1}{\ert{S'}{F_n}} \sumform \identityOp\\
            \ert{\while'}{\zeroOp} &= \bigvee_n G_n \quad \text{where }\\
            G_0 &= \zeroOp,\\
            G_{n+1} &= \formsandwich{M_0}{\zeroOp} \sumform \formsandwich{M_1}{\ert{S'}{G_n}} \sumform \identityOp\\
            &= \formsandwich{M_1}{\ert{S'}{G_n}} \sumform \identityOp\\
            \qwp{\while'}{A} &= \bigvee_n H_n \quad \text{where }\\
            H_0 &= \zeroOp,\\
            H_{n+1} &= \formsandwich{M_0}{A} \sumform \formsandwich{M_1}{\qwp{S'}{H_n}}
        \end{align*}
        First of all, each $F_n\in \predicate$ by induction over $n$ as $\predicate$ is closed under $\sumform$ and $\formsandwich{M}{\cdot}$ for any $M$. We can also show that $F_n$ is increasing in $n$ by induction over $n$ as well:
        \begin{itemize}
            \item $n=0$: $F_0 = \zeroOp \loewner F_1$.
            \item $n \to n+1$:
            \begin{align*}
                F_{n} &= \formsandwich{M_0}{A} \sumform \formsandwich{M_1}{\ert{S'}{F_{n-1}}} \sumform \identityOp \\
                &\loewner \formsandwich{M_0}{A} \sumform \formsandwich{M_1}{\ert{S'}{F_{n}}} \sumform \identityOp = F_{n+1}
             \end{align*}
             by Proposition \ref{prop:orderproperties} and because $\ert{S'}{F_{n-1}} \loewner \ert{S'}{F_{n}}$ holds by the monotonicity of $\ert{S'}{\cdot}$ which follows by induction hypothesis.
        \end{itemize}
        Thus $F_n$ is an increasing sequence in $n$ and thus $\bigvee_n F_n$ is well-defined. (Analogously, we can show the same for $G_n$ and $H_n$.)

        We show by induction over $n$ that $F_n = H_n \sumform G_n$ for all $n$:
        \begin{itemize}
            \item $n=0$: $F_0 = \zeroOp = H_0 \sumform G_0$.
            \item $n \to n+1$:
            \begin{align*}
                F_{n+1} &= \formsandwich{M_0}{A} \sumform \formsandwich{M_1}{\ert{S'}{F_n}} \sumform \identityOp \\
                &\overset{IH}{=} \formsandwich{M_0}{A} \sumform \formsandwich{M_1}{\ert{S'}{H_n \sumform G_n}} \sumform \identityOp \\
                &\overset{IH}{=} \formsandwich{M_0}{A} \sumform \formsandwich{M_1}{\left(\qwp{S'}{H_n \sumform G_n} \sumform \ert{S'}{\zeroOp}\right)} \sumform \identityOp \\
                &= \formsandwich{M_0}{A} \sumform \formsandwich{M_1}{\left(\qwp{S'}{H_n} \sumform \qwp{S'}{G_n} \sumform \ert{S'}{\zeroOp}\right)} \sumform \identityOp \\
                &\overset{IH}{=} \formsandwich{M_0}{A} \sumform \formsandwich{M_1}{\left(\qwp{S'}{H_n} \sumform \ert{S'}{G_n}\right)} \sumform \identityOp \\
                &= \formsandwich{M_0}{A} \sumform \formsandwich{M_1}{\qwp{S'}{H_n}} \sumform \formsandwich{M_1}{\ert{S'}{G_n}} \sumform \identityOp \\
                &= H_{n+1} \sumform G_{n+1}
             \end{align*}
            \end{itemize}

            Thus we have $F_n = H_n \sumform G_n$ for all $n$. Finally, we can show the connection to wp:
            \begin{align*}
                &\ert{\while'}{A} = \bigvee_n F_n \\
                =& \bigvee_n (H_n \sumform G_n) = \bigvee_n H_n \sumform \bigvee_n G_n \\
                =& \text{ }\qwp{\while'}{A} \sumform \ert{\while'}{\zeroOp}
            \end{align*}
    \end{itemize}

    \item[4.] Continuity: Let $\{A_i\}_i$ be an increasing chain of expectations. We show $\ert{S}{\bigvee_i A_i} = \bigvee_i \ert{S}{A_i}$:
    \begin{align*}
        &\text{ } \ert{S}{\bigvee_i A_i}\\
        =& \text{ }\qwp{S}{\bigvee_i A_i} \sumform \ert{S}{\zeroOp} \\
        =& \bigvee_i \qwp{S}{A_i} \sumform \ert{S}{\zeroOp} \\
        =& \bigvee_i \left(\qwp{S}{A_i} \sumform \ert{S}{\zeroOp}\right) \\
        =& \bigvee_i \ert{S}{A_i}
     \end{align*}
     where the first step holds by 2., the second one by the continuity of $\qwp{S}{\cdot}$ (Proposition \ref{prop:wpProperties}) and the third one by the continuity of $\sumform$ (Proposition \ref{prop:ContinuityDistributivityFormSum}).

    \item[5.] Constant propagation: We show $\ert{S}{c\cdot A} = c\cdot A \sumform \ert{S}{\zeroOp}$
    \begin{align*}
        &\ert{S}{c\cdot A} \\
        &\overset{2.}{=} \qwp{S}{c\cdot A} \sumform \ert{S}{\zeroOp} \\
        &= c\cdot \qwp{S}{A} \sumform \ert{S}{\zeroOp} \\
        &\overset{2.}{=} c\cdot A \sumform \ert{S}{\zeroOp}
     \end{align*}
     where the second step holds as wp is linear by Proposition \ref{prop:wpProperties}.

     \item[6.] Preservation of infinity: It is $\ert{S}{\inftyOp} = \inftyOp$ because $\formsandwich{M}{\inftyOp} = \inftyOp$ for any (linear bounded) $M\neq \zeroOp$, $\inftyOp \sumform A = \inftyOp$ for any (linear) $A$ and $\bigvee \inftyOp = \inftyOp$, thus it follows by induction over the structure of $S$ that $\ert{S}{\inftyOp} = \inftyOp$. \qedhere
\end{enumerate}
\end{proof}
\newpage
Proof of Theorem \ref{thm:equivErtQwp}:
\begin{proof}
    We show the statement by induction over the structure of $S$:
    \begin{itemize}
        \item $S = \skipbf$: \begin{align*}
            \ert{\skipbf}{A} = A = \qwp{\skipbf}{A} = \qwp{\transform{\skipbf}}{A}
        \end{align*}

        \item $S = \qzero$: \begin{align*}
            &\text{ }\ert{\qzero}{A} \\
            =& \text{ }\identityOp \sumform \sum_{n=0}^\infty \ket{n}\bra{0} \odot A \odot \ket{0}\bra{n} \\
            =& \text{ }\identityOp \sumform \qwp{\qzero}{A} \\
            =& \text{ }\qwp{\reward{1}}{\qwp{\qzero}{A}}\\
            =& \text{ }\qwp{\reward{1}; \qzero}{A}\\
             =& \text{ }\qwp{\transform{\qzero}}{A}
        \end{align*}

        \item $S = \Uq$: \begin{align*}
            &\text{ } \ert{\Uq}{A} = \identityOp \sumform \formsandwich{U}{A}  \\
            =& \text{ }\identityOp \sumform \qwp{\Uq}{A} = \qwp{\reward{1}}{\qwp{\Uq}{A}} \\
            =& \text{ }\qwp{\transform{\Uq}}{A}
        \end{align*}

        \item $S = \concat$:
        \begin{align*}
            &\ert{\concat}{A} = \ert{S_1}{\ert{S_2}{A}} \\
            \overset{IH}{=} &\qwp{\transform{S_1}}{\qwp{\transform{S_2}}{A}}  = \qwp{\transform{\concat}}{A}
        \end{align*}

        \item $S = \measurePrime$:
        \begin{align*}
            &\text{ }\ert{\measurePrime}{A} \\
            =& \text{ }\identityOp \sumform \left(\sum_{m} \formsandwich{M_m}{\ert{S_m'}{A}} \right) \\
            \overset{IH}{=} &\identityOp \sumform \left(\sum_{m} \formsandwich{M_m}{\qwp{\transform{S_m'}}{A}} \right) \\
            =& \text{ }\identityOp \sumform \qwp{\measurePrime}{A} \\
            =& \text{ }\qwp{\reward{1}}{\qwp{\measurePrime}{A}} \\
            =& \text{ }\qwp{\transform{\measurePrime}}{A}
        \end{align*}

        \item $S = \while'$:
        For this case, we need an auxiliary result:
        \begin{lemma}
            Let $S$ be a reward-free qrWhile program and $A \in \predicate$. Then $\ert{\while}{A} = \bigvee F'_n$ where
            \begin{align*}
                F'_0 &= \identityOp\\
                F'_{n+1} &= \formsandwich{M_0}{A} \sumform \formsandwich{M_1}{\ert{S}{F'_n}} \sumform \identityOp.
            \end{align*}
        \end{lemma}
        \begin{proof}
            We have to show that $\bigvee F'_n$ and $\bigvee F_n$ are equal where $F_n$ is defined as in Table \ref{tab:ert} for $\while'$, i.e.
            \begin{align*}
                F_0 &= \zeroOp\\
                F_{n+1} &= \formsandwich{M_0}{A} \sumform \formsandwich{M_1}{\ert{S'}{F_n}} \sumform \identityOp
            \end{align*}
            The only difference is the definition of $F_0$.

            Clearly $F'_n$ is an increasing sequence as well, so we show by induction over $n$ that $F'_n \loewner F_{n+1}$ and $F_n \loewner F'_{n}$ using the induction hypothesis, monotonicity of $\ert{S'}{\cdot}$ and Proposition \ref{prop:orderproperties}. This implies $\bigvee_n F'_n = \bigvee_n F_n$.
            \begin{itemize}
            \item $F'_n \loewner F_{n+1}$:
            \begin{align*}
                F'_0  =& \text{ }\identityOp \loewner  \identityOp \sumform \formsandwich{M_0}{A} \sumform \formsandwich{M_1}{\ert{S'}{\zeroOp}} = F_1\\
                F'_{n+1} =& \text{ }\formsandwich{M_0}{A} \sumform \formsandwich{M_1}{\ert{S'}{F'_n}} \sumform \identityOp  \\
                \loewner& \text{ }\formsandwich{M_0}{A} \sumform \formsandwich{M_1}{\ert{S'}{F_{n+1}}} \sumform \identityOp = F_{n+2}
            \end{align*}
            \item $F_n \loewner F'_{n}$
            \begin{align*}
                F_0 = &\text{ } \zeroOp \loewner \identityOp = F'_0\\
                F_{n+1} =& \text{ }\formsandwich{M_0}{A} \sumform \formsandwich{M_1}{\ert{S'}{F_n}} \sumform \identityOp \\
                \loewner& \text{ }\formsandwich{M_0}{A} \sumform \formsandwich{M_1}{\ert{S'}{F'_n}} \sumform \identityOp = F'_{n+1} \qedhere
            \end{align*}
            \end{itemize}
        \end{proof}

        Using the auxiliary result, we can now show the statement for the while case. Recall that $\qwp{\transform{\while'}}{A} = \bigvee A_n$ where
        \begin{align*}
            A_0 &= \zeroOp,\\
            A_{n+1} &= \formsandwich{M_0}{A} \sumform \formsandwich{M_1}{\qwp{\transform{S'};\reward{1}}{A_n}}
        \end{align*}
        We show by induction over $n$ that $F'_n = A_n \sumform \identityOp$ for all $n$ which implies $\bigvee_n F'_n = \identityOp \sumform \bigvee_n A_n$.
        \begin{itemize}
            \item $n = 0$: \begin{align*}
                F'_0 &= \identityOp \\
                A_0 &= \zeroOp \\
                F'_0 &= A_0 \sumform \identityOp
            \end{align*}
            \item $n \to n+1$: \begin{align*}
                F'_{n+1} &= \formsandwich{M_0}{A} \sumform \formsandwich{M_1}{\ert{S'}{F'_n}} \sumform \identityOp \\
                &\overset{IH}{=} \formsandwich{M_0}{A} \sumform \formsandwich{M_1}{\ert{S'}{A_n \sumform \identityOp}} \sumform \identityOp \\
                &\overset{IH}{=} \formsandwich{M_0}{A} \sumform \formsandwich{M_1}{\qwp{\transform{S'}}{A_n \sumform \identityOp}} \sumform \identityOp \\
                &= \formsandwich{M_0}{A} \sumform \formsandwich{M_1}{\qwp{\transform{S'};\reward{1}}{A_n}} \sumform \identityOp \\
                &= A_{n+1} \sumform \identityOp \qedhere
        \end{align*}
        \end{itemize}

    \end{itemize}

\end{proof}

Proof of Theorem \ref{thm:ertEquivERT}:
\begin{proof}
    We show this by induction:
    \begin{itemize}
        \item $S = \skipbf$: $\expect{\ert{\skipbf}{\zeroOp}}{\rho} = \expect{\zeroOp}{\rho} = 0 = \ERT{\skipbf}{\rho}$.

        \item $S = \qzero$:
        \begin{align*}
            \text{ }\expect{\ert{\qzero}{\zeroOp}}{\rho}
=& \text{ }\expect{\identityOp \sumform \sum_{n=0}^\infty \ket{n}\bra{0} \odot \zeroOp \odot \ket{0}\bra{n}}{\rho} \\
=& \text{ }\expect{\identityOp}{\rho} + \expect{\sum_{n=0}^\infty \ket{n}\bra{0} \odot \zeroOp \odot \ket{0}\bra{n}}{\rho} \\
=& \text{ }tr(\rho) + \expect{\zeroOp}{\rho} = tr(\rho) = \ERT{\qzero}{\rho}
        \end{align*}

        \item $S = \Uq$: $\expect{\ert{\Uq}{\zeroOp}}{\rho} = \expect{\formsandwich{U}{\zeroOp} \sumform \identityOp}{\rho} = \expect{\identityOp}{\rho}= tr(\rho) = \ERT{\Uq}{\rho}$.

        \item $S = \concat$: We use the connection to $wp$ (Proposition \ref{prop:ERTProperties})
        \begin{align*}
            \expect{\ert{\concat}{\zeroOp}}{\rho}
            =&\text{ }\expect{\ert{S_1}{\ert{S_2}{\zeroOp}}}{\rho} \\
    =& \text{ }\expect{\qwp{S_1}{\ert{S_2}{\zeroOp}} \sumform \ert{S_1}{\zeroOp}}{\rho} \\
=& \text{ }\expect{\qwp{S_1}{\ert{S_2}{\zeroOp}}}{\rho} + \expect{\ert{S_1}{\zeroOp}}{\rho} \\
=& \text{ }\expect{\ert{S_2}{\zeroOp}}{\semantics{S_1}(\rho)} + \expect{\ert{S_1}{\zeroOp}}{\rho} \\
=& \text{ }\ERT{S_2}{\semantics{S_1}(\rho)} + \ERT{S_1}{\rho} \\
=& \text{ }\ERT{\concat}{\rho}
        \end{align*}

        \item $S = \measurePrime$:
        \begin{align*} &\text{ }\expect{\ert{\measurePrime}{\zeroOp}}{\rho}\\
    = &\text{ }\expect{\identityOp \sumform \left(\sum_{m} \formsandwich{M_m}{\ert{S_m'}{\zeroOp}} \right)}{\rho} \\
=& \text{ }\expect{\identityOp}{\rho} + \expect{\sum_{m} \formsandwich{M_m}{\ert{S_m'}{\zeroOp}}}{\rho} \\
=& \text{ }tr(\rho) + \sum_{m} \expect{\formsandwich{M_m}{\ert{S_m'}{\zeroOp}}}{\rho} \\
=& \text{ }tr(\rho) + \sum_{m} \expect{\ert{S_m'}{\zeroOp}}{M_m{\rho} M_m^\dagger} \\
=& \text{ }tr(\rho) + \sum_{m} \ERT{S_m'}{M_m \rho M_m^\dagger} \\
=& \text{ }\ERT{\measurePrime}{\rho}
        \end{align*}

        \item $S= \while'$: We show $\expect{F_n}{\rho} = \ERT{while^{[n]}}{\rho}$
        \begin{align*}
            \expect{F_0}{\rho} =& \text{ }\expect{\zeroOp}{\rho} = 0 = \ERT{while^{[0]}}{\rho} \\
            \expect{F_{n+1}}{\rho} =&\text{ }  \expect{\formsandwich{M_0}{\zeroOp} \sumform \formsandwich{M_1}{\ert{S'}{F_n}} \sumform \identityOp}{\rho} \\
=& \text{ } \expect{\zeroOp}{\rho} + \expect{\formsandwich{M_1}{\ert{S'}{F_n}}}{\rho} + \expect{\identityOp}{\rho} \\
=& \text{ } 0 + \expect{\ert{S'}{F_n}}{M_1 \rho M_1^\dagger} + tr(\rho) \\
=& \text{ }\expect{\qwp{S'}{F_n} + \ert{S'}{\zeroOp} }{M_1 \rho M_1^\dagger} + tr(\rho) \\
=& \text{ }\expect{\qwp{S'}{F_n}}{M_1 \rho M_1^\dagger} + \expect{\ert{S'}{\zeroOp}}{M_1 \rho M_1^\dagger} + tr(\rho) \\
=& \text{ }\expect{F_n}{\semantics{S'}(M_1 \rho M_1^\dagger)} + \expect{\ert{S'}{\zeroOp}}{M_1 \rho M_1^\dagger} + tr(\rho) \\
=& \text{ }0 + \ERT{while^{[n]}}{\semantics{S'}(M_1 \rho M_1^\dagger)} + \ERT{S'}{M_1 \rho M_1^\dagger} + tr(\rho) \\
=& \text{ }\ERT{\skipbf}{M_0 \rho M_0^\dagger} + \ERT{S';while^{[n]}}{M_1 \rho M_1^\dagger} + tr(\rho)\\
=& \text{ }\ERT{while^{[n+1]}}{\rho}
        \end{align*}
        Then \begin{align*}
            &\text{ } \expect{\ert{\while'}{\zeroOp}}{\rho} \\
            =& \text{ }\expect{\bigvee_n F_n}{\rho} = \sup_n \expect{F_n}{\rho} \\
            =& \sup_n \ERT{while^{[n]}}{\rho} \\
            =& \lim_n \ERT{while^{[n]}}{\rho} \\
            =& \text{ } \ERT{\while'}{\rho}
        \end{align*}
        where supremum and limit coincide since we have increasing sequences.
    \end{itemize}
\end{proof}

Proof of Lemma \ref{lem:PastImpliesAst}:
\begin{proof}
    Let $S$ be reward-free and $\rho \in \density$. Assume $S$ is PAST on $\rho$, that means $\semanticsR{\transform{S}}(\rho) < \infty$ by definition. We show by induction that $tr(\semanticsD{S}(\rho)) = tr(\rho)$, that means $S$ is AST on $\rho$.
    \begin{itemize}
        \item $S = \skipbf$: $tr(\semanticsD{\skipbf}(\rho)) = tr(\rho)$ by definition.
        \item $S = \qzero$:
        \begin{align*}
            &tr(\semanticsD{\qzero}(\rho)) = tr(\sum_n \ket{0}\bra{n} \rho \ket{n}\bra{0}) \\
            =& \sum_n tr( \ket{n}\bra{0} \ket{0}\bra{n} \rho ) = tr(\sum_n \ket{n}\bra{n} \rho ) = tr(\rho).
        \end{align*}
        \item $S = \Uq$: $tr(\semanticsD{\Uq}(\rho)) = tr(U \rho U^\dagger) =tr(U^\dagger U \rho )= tr(\rho)$ as $U$ is unitary.
        \item $S = \concat$: By $\semanticsR{\transform{\concat}}(\rho) = \semanticsR{\transform{S_1}}(\rho) + \semanticsR{\transform{S_2}}(\semanticsD{S_1}(\rho)) < \infty$, we have
        \begin{align*}
            \semanticsR{\transform{S_1}}(\rho) < \infty \quad \text{and} \quad \semanticsR{\transform{S_2}}(\semanticsD{S_1}(\rho)) < \infty.
        \end{align*}
        Thus $tr(\semanticsD{S_1}(\rho)) = tr(\rho)$ and $tr(\semanticsD{S_2}(\semanticsD{S_1}(\rho))) = tr(\semanticsD{S_1}(\rho))$ by induction hypothesis. Overall it is
        $tr(\semanticsD{\concat}(\rho)) = tr(\semanticsD{S_2}(\semanticsD{S_1}(\rho))) = tr(\semanticsD{S_1}(\rho)) = tr(\rho)$.
        \item $S = \measurePrime$: It is $\semanticsR{\transform{\measurePrime}}(\rho) = \sum_m \semanticsR{\transform{S_m'}}(M_m \rho M_m^\dagger) < \infty$, thus $\semanticsR{\transform{S_m'}}(M_m \rho M_m^\dagger) < \infty$ for all $m$. Then $tr(\semanticsD{S_m'}(M_m \rho M_m^\dagger)) = tr(M_m \rho M_m^\dagger)$ for all $m$ by induction hypothesis. Overall it is
        \begin{align*}
            tr(\semanticsD{\measurePrime}(\rho)) &= tr(\sum_m \semanticsD{S_m'}(M_m \rho M_m^\dagger)) \\
            &= \sum_m tr(\semanticsD{S_m'}(M_m \rho M_m^\dagger)) \\
            &\overset{IH}{=} \sum_m tr(M_m \rho M_m^\dagger) \\
            &= tr(\sum_m M_m^\dagger M_m \rho) \\
            &= tr(\rho)
        \end{align*}
        \item $S = \while'$:
        By $\ERT{\while'}{\rho}< \infty$, we have
        \begin{align*}
            \forall n:&\text{ }  \ERT{(\while')^{[n]}}{\rho}\\
             \leq &\lim_{n \to \infty} \ERT{(\while')^{[n]}}{\rho} \\
             = & \text{ } \ERT{\while'}{\rho} < \infty
        \end{align*}
        and thus by induction hypothesis, it is $tr(\semanticsD{(\while')^{[n]}} (\rho)) = tr(\rho)$ for all $n$.

        We define $\sigma_0 = \rho$ and $\sigma_{n+1} = \semanticsD{S'}(M_1 \sigma_n M_1^\dagger)$ for all $n$. Intuitively, $\sigma_n$ is the state entering the $n$-th iteration of the loop.
        We have
        \begin{align*}
            \infty > & \text{ }\ERT{(\while')^{[n]}}{\rho} \\
            = & \text{ } tr(\rho) + \ERT{\skipbf}{M_0 \rho M_0^\dagger} + \ERT{S'; (\while')^{[n-1]}}{M_1 \rho M_1^\dagger} \\
            =& \text{ }tr(\rho) + \ERT{S'}{M_1 \rho M_1^\dagger} + \ERT{(\while')^{[n-1]}}{\semanticsD{S'}(M_1 \rho M_1^\dagger)}\\
            \geq & \text{ }tr(\sigma_0) + \ERT{(\while')^{[n-1]}}{\sigma_1}\\
            \geq &\text{ }  tr(\sigma_0) + tr(\sigma_1) + \ERT{(\while')^{[n-2]}}{\sigma_2} \\
            \geq & \dots \\
            \geq & \text{ } tr(\sigma_0) + tr(\sigma_1) + \ldots + \ERT{(\while')^{[1]}}{\sigma_{n-1}} \geq \sum_{i=0}^{{n-1}} tr(\sigma_i)
        \end{align*}
        Since $\sum_{i=0}^{n} tr(\sigma_i) < \infty$ for all $n$, $\lim_{n \to \infty} tr(\sigma_n) = 0$.

        From the equation above we get
        $\ERT{(\while')^{[1]}}{\sigma_n} < \infty$ for all $n$, thus
        \begin{align*}
            \infty > & \text{ }\ERT{(\while')^{[1]}}{\sigma_n} \\
            =& \text{ }tr(\sigma_n) + \ERT{S'; \skipbf}{M_1 \sigma_n M_1^\dagger} + \ERT{\skipbf}{M_0 \sigma_n M_0^\dagger} \\
            \geq &\text{ } \ERT{S'}{M_1 \sigma_n M_1^\dagger}
        \end{align*}
        That means $S'$ is PAST on $M_1 \sigma_n M_1^\dagger$ for all $n$. By induction hypothesis for $S'$, we have $tr(\semanticsD{S'}(M_1 \sigma_n M_1^\dagger)) = tr(M_1 \sigma_n M_1^\dagger) $ for all $n$.

        Then we show that $tr(\rho) - tr(\semanticsD{(\while')^{n}} (\rho)) = tr(\sigma_n)$ by induction over $n$:
        \begin{itemize}
            \item $n=0$: $tr(\rho) - tr(\semanticsD{(\while')^{0}} (\rho)) = tr(\rho) - tr(0) = tr(\rho) = tr(\sigma_0)$
            \item $n \to n+1$: It is
            \begin{align*}
                &\text{ }  tr(\rho) - tr(\semanticsD{(\while')^{n+1}} (\rho)) \\
                = & \text{ } tr(\rho) - tr(\semanticsD{\skipbf}(M_0 \rho M_0^\dagger) + \semanticsD{S'; (\while')^{n}}(M_1 \rho M_1^\dagger)) \\
                = & \text{ } tr(\rho) - tr(\semanticsD{\skipbf}(M_0 \rho M_0^\dagger)) - tr(\semanticsD{S'; (\while')^{n}}(M_1 \rho M_1^\dagger)) \\
                = & \text{ } tr(\rho) - tr(M_0 \rho M_0^\dagger) - tr(\semanticsD{ (\while')^{n}}(\semanticsD{S'}(M_1 \rho M_1^\dagger))) \\
                & + tr(\semanticsD{S'}(M_1 \rho M_1^\dagger)) - tr(\semanticsD{S'}(M_1 \rho M_1^\dagger))\\
                =& \text{ } tr(\rho) - tr(M_0 \rho M_0^\dagger) - tr(M_1 \rho M_1^\dagger) + \\
                & \left(tr(\semanticsD{S'}(M_1 \rho M_1^\dagger))  - tr(\semanticsD{ (\while')^{n}}(\semanticsD{S'}(M_1 \rho M_1^\dagger)))\right) \\
                \overset{I.H.}{=}& 0+ tr(\sigma_{n+1})
            \end{align*}
            because $\rho = \sigma_0$, thus $tr(\semanticsD{S'}(M_1 \rho M_1^\dagger)) = tr(M_1 \rho M_1^\dagger)$ by induction hypothesis for $S'$. Also $tr(\rho) - tr(M_0 \rho M_0^\dagger) - tr(M_1 \rho M_1^\dagger) = 0$ because $M_0^\dagger M_0 + M_1^\dagger M_1 = \identityOp$.
        \end{itemize}
        Overall we have shown that $tr(\rho) - tr(\semanticsD{(\while')^{n}} (\rho)) = tr(\sigma_n)$ for all $n$ and as $\lim_{n \to \infty} tr(\sigma_n) = 0$, $\lim_{n \to \infty} tr(\semanticsD{(\while')^{n}} (\rho)) = tr(\rho)$ and thus $tr(\semanticsD{\while'} (\rho)) = tr(\rho)$ (limit and $\bigvee$ coincide as we have an increasing sequence of real numbers) and $S$ is AST on $\rho$.
    \end{itemize}
\end{proof}

Proof of Proposition \ref{prop:ertPark}:
\begin{proof}
    First we show an auxiliary result that $\ert{\while}{B} = \bigvee_n \Psi_B^n(0)$ where $\Psi_B^n$ is the $n$-fold composition of $\Psi_B$:

    We show by induction over $n$ that $\Psi_B^n(0) = F_n$ where $F_n$ is defined as in the definition of $\ert{\while}{B}$ in Table \ref{tab:ert}:
    \begin{itemize}
        \item $n=0$: $\Psi_B^0(0) = 0 = F_0$.
        \item $n \to n+1$:
        \begin{align*}
            \Psi_B^{n+1}(0) =&\text{ }  \Psi_B(\Psi_B^n(0)) \\
=& \text{ } \identityOp \sumform \formsandwich{M_0}{B} \sumform \formsandwich{M_1}{\ert{S}{\Psi_B^n(0)}}  \\
=& \text{ } \identityOp \sumform \formsandwich{M_0}{B} \sumform \formsandwich{M_1}{\ert{S}{F_n}} \\
=& \text{ } F_{n+1}
        \end{align*}
    \end{itemize}
    Thus we have $\Psi_B^n(0) = F_n$ for all $n$ and thus $\bigvee_n \Psi_B^n(0) = \bigvee_n F_n = \ert{\while}{B}$.

    The rest of the proof is now completely analogous to the proof of Park induction for wp (Proposition \ref{prop:park}) and thus we only give a sketch here:
    $\Psi_B$ is monotonous ($A_1 \loewner A_2 \Rightarrow \Psi_B(A_1) \loewner \Psi_B(A_2)$) by the monotonicity of ert and by Proposition \ref{prop:orderproperties}:
    \begin{align*}
        A_1 \loewner A_2 &\Rightarrow \formsandwich{M_1}{\ert{S}{A_1}} \loewner \formsandwich{M_1}{\ert{S}{A_2}} \\
        &\Rightarrow \identityOp \sumform \formsandwich{M_0}{B} \sumform \formsandwich{M_1}{\ert{S}{A_1}} \\
        & \loewner \identityOp \sumform \formsandwich{M_0}{B} \sumform \formsandwich{ M_1}{\ert{S}{A_2}} \\
        &\Rightarrow \Psi_B(A_1) \loewner \Psi_B(A_2)
    \end{align*}
    Then $\Psi_B(A) \loewner A$ implies $\Psi_B^n(0) \loewner A$ completely analogous as we have the same needed properties of $\Psi_B$ as in the wp case. Finally, we have $\ert{\while}{B} = \bigvee_n \Psi_B^n(0) \loewner A$.
\end{proof}

\section{Details of Section \ref{sec:examples} - Examples}
\label{sec:appendix_examples}

In this section, we provide additional details and calculations for the examples presented in Section~\ref{sec:examples}.

\subsection{Example 1: Probabilistic PAST vs. Quantum PAST}
\label{sec:appendix_example1}
To start the analysis, we can compute $\ert{\while_2}{\zeroOp}= \bigvee_n P_n$ of the second while loop. Intuitively, this loop decrements the quantum integer $q_1$ until it reaches $\ket{0}$, i.e., its runtime is the current value of $q_1$. We have:
\begin{align*}
    P_0 &= \zeroOp \\
    P_{n+1} &= M_0^\dagger \zeroOp M_0 + M_1^\dagger \ert{q_1 := U_{-1} q_1}{P_n} M_1 + \identityOp \\
    &= M_1^\dagger (U_{-1}^\dagger P_n U_{-1}+\identityOp) M_1 + \identityOp
\end{align*}
Iterating this a couple of times gives:
\begin{align*}
    P_1 &= 2\identityOp - \ket{0}\bra{0} \\
    P_2 &= 3\identityOp - \ket{1}\bra{1} - 2\ket{0}\bra{0} \\
    P_3 &= 4\identityOp - \ket{2}\bra{2} - 2\ket{1}\bra{1} - 3\ket{0}\bra{0}
\end{align*}
We can see that $P_n = (n+1)\identityOp - \sum_{k=0}^{n-1} (n-k)\ket{k}\bra{k}$. Thus, we have:
\begin{equation*}
    \ert{\while_2}{\zeroOp} = \bigvee_{n=0}^\infty P_n = \lim_{n\to\infty} \left((n+1)\identityOp - \sum_{k=0}^{n-1} (n-k)\ket{k}\bra{k}\right):= W_1
\end{equation*}
To compute this limit, we can look at the expectation value on a basis state $\ket{m}$:
\begin{align*}
    \expect{\ert{\while_2}{\zeroOp}}{\ket{m}\bra{m}} &= \lim_{n\to\infty} \left((n+1) - (n-m)\right) \\
    &= \lim_{n\to\infty} (m+1) = m+1
\end{align*}
That means the expected runtime of the second while loop on input state $\ket{m}$ is $m+1$ as expected. Note that this is only about variable $q_1$ here, the other variable do not matter for this loop, i.e., we have $\identityOp$ on the other variable.

Now we can analyze the first while loop. Let $\ert{while_1}{W_1\otimes \identityOp} = \bigvee_n Q_n$ with
\begin{align*}
    Q_0 =&\text{ } \zeroOp \\
    Q_{n+1} =& \text{ } (\identityOp \otimes M_0^\dagger) (W_1 \otimes \identityOp) (\identityOp \otimes M_0) \\
    & + (\identityOp \otimes M_1^\dagger) \ert{q_2:= Hq_2; q_1 := U_{2} q_1}{Q_n} (\identityOp \otimes M_1) + \identityOp \\
     =& \text{ } W_1 \otimes (\identityOp - \ket{1}\bra{1}) \\
    & + (\identityOp \otimes \ket{1}\bra{1}) \ert{q_2:= Hq_2; q_1 := U_{2} q_1}{Q_n}(\identityOp \otimes \ket{1}\bra{1}) + \identityOp
\end{align*}
with
\begin{align*}
    &\ert{q_2:= Hq_2; q_1 := U_{2} q_1}{Q_n} \\
    = & \text{ } (\identityOp \otimes H^\dagger) ((U_2^\dagger \otimes \identityOp) Q_n (U_2 \otimes \identityOp)+ \identityOp)( \identityOp \otimes H) + \identityOp\\
    =& \text{ } (U_2^\dagger \otimes H^\dagger) Q_n (U_2 \otimes H) + 2\identityOp
\end{align*}
Overall we have \begin{align*}
    Q_{n+1} &= W_1 \otimes (\identityOp - \ket{1}\bra{1}) +  (\identityOp \otimes \ket{1}\bra{1}) ((U_2^\dagger \otimes H^\dagger) Q_n (U_2 \otimes H) + 2\identityOp) (\identityOp \otimes \ket{1}\bra{1}) + \identityOp \\
    &= W_1 \otimes (\identityOp - \ket{1}\bra{1}) +  (\identityOp \otimes \ket{1}\bra{1}) (U_2^\dagger \otimes H^\dagger) Q_n (U_2 \otimes H) (\identityOp \otimes \ket{1}\bra{1}) + \identityOp + 2 \identityOp \otimes \ket{1}\bra{1} \\
\end{align*}

Thus we have:
\begin{align*}
    Q_0 =& \text{ } \zeroOp \\
    Q_1 =& \text{ }W_1 \otimes (\identityOp - \ket{1}\bra{1}) + (\identityOp \otimes \ket{1}\bra{1}) (U_2^\dagger \otimes H^\dagger) \zeroOp (U_2 \otimes H) (\identityOp \otimes \ket{1}\bra{1}) + \identityOp + 2 \identityOp \otimes \ket{1}\bra{1} \\
    =& \text{ } W_1 \otimes (\identityOp - \ket{1}\bra{1}) + \identityOp + 2 \identityOp \otimes \ket{1}\bra{1}\\
    Q_2 =& \text{ } (\identityOp \otimes \ket{1}\bra{1}) (U_2^\dagger \otimes H^\dagger) \left(W_1 \otimes (\identityOp - \ket{1}\bra{1}) + \identityOp + 2 \identityOp \otimes \ket{1}\bra{1}\right) (U_2 \otimes H) (\identityOp \otimes \ket{1}\bra{1}) \\
    & + \identityOp + 2 \identityOp \otimes \ket{1}\bra{1} + W_1 \otimes (\identityOp - \ket{1}\bra{1}) \\
    =& \text{ } (\identityOp \otimes \ket{1}\bra{1}) (U_2^\dagger \otimes H^\dagger) \left(W_1 \otimes (\identityOp - \ket{1}\bra{1}) \right) (U_2 \otimes H) (\identityOp \otimes \ket{1}\bra{1}) \\
    & + \identityOp + (2+1+2\cdot \frac{1}{2}) \identityOp \otimes \ket{1}\bra{1} + W_1 \otimes (\identityOp - \ket{1}\bra{1}) \\
    =& \text{ } U_2^\dagger W_1 U_2 \otimes \left(\ket{1}\bra{1} H^\dagger(\identityOp - \ket{1}\bra{1}) H \ket{1}\bra{1}\right) \\
    & + \identityOp + 4 \identityOp \otimes \ket{1}\bra{1} + W_1 \otimes (\identityOp - \ket{1}\bra{1}) \\
    =& \text{ } U_2^\dagger W_1 U_2 \otimes \left(\frac{1}{2} \ket{1}\bra{1}\right) + \identityOp + 4 \identityOp \otimes \ket{1}\bra{1} + W_1 \otimes (\identityOp - \ket{1}\bra{1}) \\
    =&  \left( 4 \identityOp + \frac{1}{2} U_2^\dagger W_1 U_2 \right)\otimes \ket{1}\bra{1} + \identityOp + W_1 \otimes (\identityOp - \ket{1}\bra{1}) \\
    Q_3 =& \text{ } (\identityOp \otimes \ket{1}\bra{1}) (U_2^\dagger \otimes H^\dagger) \left(\left( 4 \identityOp + \frac{1}{2} U_2^\dagger W_1 U_2 \right)\otimes \ket{1}\bra{1} + \identityOp + W_1 \otimes (\identityOp - \ket{1}\bra{1}) \right) (U_2 \otimes H) (\identityOp \otimes \ket{1}\bra{1}) \\
    & + \identityOp + 2 \identityOp \otimes \ket{1}\bra{1} + W_1 \otimes (\identityOp - \ket{1}\bra{1}) \\
    =&\text{ }  U_2^\dagger\left( 4 \identityOp + \frac{1}{2} U_2^\dagger W_1 U_2 \right) U_2 \otimes \frac{1}{2}\ket{1}\bra{1} + U_2^\dagger W_1 U_2 \otimes \ket{1}\bra{1} H^\dagger\left(\identityOp - \ket{1}\bra{1}\right)H\ket{1}\bra{1} \\
    & + \identityOp + 3 \identityOp \otimes \ket{1}\bra{1} + W_1 \otimes (\identityOp - \ket{1}\bra{1}) \\
    =& \text{ } U_2^\dagger\left( 4 \identityOp + \frac{1}{2} U_2^\dagger W_1 U_2 \right) U_2 \otimes \frac{1}{2}\ket{1}\bra{1} + U_2^\dagger W_1 U_2 \otimes \frac{1}{2}\ket{1}\bra{1} + \identityOp + 3 \identityOp \otimes \ket{1}\bra{1} + W_1 \otimes (\identityOp - \ket{1}\bra{1}) \\
    =& \left(5 \identityOp + \frac{1}{2}U_2^\dagger\left( \frac{1}{2} U_2^\dagger W_1 U_2 +W_1 \right)  U_2\right) \otimes \ket{1}\bra{1} + \identityOp + W_1 \otimes (\identityOp - \ket{1}\bra{1}) \\
    Q_4 =& (\identityOp \otimes \ket{1}\bra{1}) (U_2^\dagger \otimes H^\dagger) \\
    & \left(\left(5 \identityOp + \frac{1}{2}U_2^\dagger\left( \frac{1}{2} U_2^\dagger W_1 U_2 +W_1 \right)  U_2\right) \otimes \ket{1}\bra{1} + \identityOp + W_1 \otimes (\identityOp - \ket{1}\bra{1})\right) \\
    &(U_2 \otimes H) (\identityOp \otimes \ket{1}\bra{1}) \\
    & + \identityOp + 2 \identityOp \otimes \ket{1}\bra{1} + W_1 \otimes (\identityOp - \ket{1}\bra{1}) \\
    =& \text{ } U_2^\dagger \left(5 \identityOp + \frac{1}{2}U_2^\dagger\left( \frac{1}{2} U_2^\dagger W_1 U_2 +W_1 \right)  U_2\right) U_2 \otimes \frac{1}{2}\ket{1}\bra{1} + U_2^\dagger W_1 U_2 \otimes \frac{1}{2}\ket{1}\bra{1} \\
    & + \identityOp + 3 \identityOp \otimes \ket{1}\bra{1} + W_1 \otimes (\identityOp - \ket{1}\bra{1}) \\
    =& \left(5.5 \identityOp + \frac{1}{2} U_2^\dagger \left(\frac{1}{2}U_2^\dagger\left( \frac{1}{2} U_2^\dagger W_1 U_2 +W_1 \right)  U_2 + W_1\right) U_2 \right)\otimes \ket{1}\bra{1} + \identityOp +  W_1 \otimes (\identityOp - \ket{1}\bra{1}) \\
    =& \left(5.5 \identityOp + \sum_{k=1}^3 (\frac{1}{2} U_2^\dagger)^k W_1 U_2 \right)\otimes \ket{1}\bra{1} + \identityOp +  W_1 \otimes (\identityOp - \ket{1}\bra{1})
\end{align*}

We can see a pattern emerging here. For $n>0$, we have:
\begin{align*}
    Q_n = \left(a_n \identityOp + \sum_{k=1}^{n-1} (\frac{1}{2} U_2^\dagger)^k W_1 U_2^k \right)\otimes \ket{1}\bra{1} + \identityOp +  W_1 \otimes (\identityOp - \ket{1}\bra{1})
\end{align*}
with $a_1 = 2$, $a_2 = 4$, $a_3 = 5$ $a_4 = 5.5$. The sequence continues as $a_n = 3 + \frac{1}{2}a_{n-1}$, i.e., $a_n = 6 - \frac{8}{2^n}$. Thus, we have:
\begin{align*}
    &\ert{\while_1}{W_1 \otimes \identityOp} = \bigvee_{n=0}^\infty Q_n \\
    =& \lim_{n\to\infty} \left( \left(\left(6 - \frac{8}{2^n} \right) \identityOp + \sum_{k=1}^{n-1} (\frac{1}{2} U_2^\dagger)^k W_1 U_2^k \right)\otimes \ket{1}\bra{1} + \identityOp +  W_1 \otimes (\identityOp - \ket{1}\bra{1}) \right) := W_2
\end{align*}
To compute the expected value of $W_2$ on a basis state $\ket{m}\otimes \ket{b}$, we need to compute the expected value of $\sum_{k=1}^{n-1} (\frac{1}{2} U_2^\dagger)^k W_1 U_2^k$ on $\ket{1}\bra{1}$:
\begin{align*}
    & \expect{\sum_{k=1}^{n-1} (\frac{1}{2} U_2^\dagger)^k W_1 U_2^k}{\ket{1}\bra{1}}\\
    =& \sum_{k=1}^{n-1} \expect{(\frac{1}{2} U_2^\dagger)^k W_1 U_2^k}{\ket{1}\bra{1}} \\
    =& \sum_{k=1}^{n-1} \expect{W_1}{(\frac{1}{2} U_2)^k \ket{1}\bra{1} (U_2^\dagger)^k} \\
    =& \sum_{k=1}^{n-1} \frac{1}{2^k} \expect{W_1}{\ket{2^{k}}\bra{2^{k}}} \\
    =& \sum_{k=1}^{n-1} \frac{1}{2^k} (2^k -1) \\
    =& \sum_{k=1}^{n-1} (1 - \frac{1}{2^k}) = n-1 - (1 - \frac{1}{2^{n-1}}) = n-2 + \frac{1}{2^{n-1}}
\end{align*}
Overall we have:
\begin{align*}
    &\text{ } \expect{\ert{L_1}{\zeroOp}}{\ket{1}\bra{1} \otimes \ket{1}\bra{1}} \\
    =& \text{ } \expect{W_2}{\ket{1}\bra{1} \otimes \ket {1}\bra{1}} \\
    =& \text{ } \expect{\lim_{n\to\infty} \left( \left(\left(6 - \frac{8}{2^n} \right) \identityOp + \sum_{k=1}^{n-1} (\frac{1}{2} U_2^\dagger)^k W_1 U_2^k \right)\otimes \ket{1}\bra{1} + \identityOp +  W_1 \otimes (\identityOp - \ket{1}\bra{1}) \right)}{\ket{1}\bra{1} \otimes \ket {1}\bra{1}} \\
    =& \lim_{n\to\infty} \expect{\left(6 - \frac{8}{2^n} \right) \identityOp }{\ket{1}\bra{1} \otimes \ket{1}\bra{1}}+ \expect{\left( \sum_{k=1}^{n-1} (\frac{1}{2} U_2^\dagger)^k W_1 U_2^k \right)\otimes \ket{1}\bra{1}}{\ket{1}\bra{1} \otimes \ket{1}\bra{1}} \\
    & + \expect{\identityOp}{\ket{1}\bra{1} \otimes \ket {1}\bra{1}}  + \expect{W_1 \otimes (\identityOp - \ket{1}\bra{1}) }{\ket{1}\bra{1} \otimes \ket {1}\bra{1}} \\
    =& \lim_{n\to\infty} 6 - \frac{8}{2^n} + \expect{\sum_{k=1}^{n-1} (\frac{1}{2} U_2^\dagger)^k W_1 U_2^k}{\ket{1}\bra{1}} + 1 \\
    =& \lim_{n\to\infty} 7 - \frac{8}{2^n} + n-2 + \frac{1}{2^{n-1}}  = \infty.
\end{align*}

For the computation of AST, we first consider the second loop. We can compute $\qwp{\while_2}{\identityOp} = \bigvee_n P_n$ with
\begin{align*}
    P_0 &= \zeroOp \\
    P_{n+1} &= M_0^\dagger \identityOp M_0 + M_1^\dagger \qwp{S}{P_n} M_1 \\
    &= \ket{0}\bra{0} + (\identityOp - \ket{0}\bra{0}) \qwp{q_1 := U_{-1} q_1}{P_n} (\identityOp - \ket{0}\bra{0}) \\
    &= \ket{0}\bra{0} + (\identityOp - \ket{0}\bra{0}) U_{-1}^\dagger P_n U_{-1} (\identityOp - \ket{0}\bra{0})
\end{align*}
Then we have
\begin{align*}
    P_0 &= \zeroOp \\
    P_1 &= \ket{0}\bra{0} \\
    P_2 &= \ket{0}\bra{0} + (\identityOp - \ket{0}\bra{0}) U_{-1}^\dagger \ket{0}\bra{0} U_{-1} (\identityOp - \ket{0}\bra{0}) \\
    &= \ket{0}\bra{0} + \ket{1}\bra{1} \\
    P_3 &= \ket{0}\bra{0} + (\identityOp - \ket{0}\bra{0}) U_{-1}^\dagger (\ket{0}\bra{0} + \ket{1}\bra{1}) U_{-1} (\identityOp - \ket{0}\bra{0})\\
     &= \ket{0}\bra{0} + \ket{1}\bra{1} + \ket{2}\bra{2} \\
    &\dots \\
    P_n &= \ket{0}\bra{0} + \ket{1}\bra{1} + \dots + \ket{n-1}\bra{n-1}
\end{align*}
and overall we have $\qwp{\while_2}{\identityOp} = \bigvee_n P_n = \identityOp$.

Now we can analyze the first while loop. We have $\qwp{\while_1}{\identityOp \otimes \identityOp} = \bigvee_n P_n$ with
\begin{align*}
    P_0 &= \zeroOp \\
    P_{n+1} &= M_0^\dagger \identityOp M_0 + M_1^\dagger \qwp{S}{P_n} M_1 \\
    &= \identityOp \otimes \ket{0}\bra{0} + (\identityOp \otimes \ket{1}\bra{1}) \qwp{q := H q; q_1 := U_2 q_1}{P_n} (\identityOp \otimes \ket{1}\bra{1}) \\
    &= \identityOp \otimes \ket{0}\bra{0} + (\identityOp \otimes \ket{1}\bra{1}) (U_2^\dagger \otimes H^\dagger) P_n (U_2 \otimes H) (\identityOp \otimes \ket{1}\bra{1})
\end{align*}
Iterating this, we have
\begin{align*}
    P_0 &= \zeroOp \\
    P_1 &= \identityOp \otimes \ket{0}\bra{0} \\
    P_2 &= \identityOp \otimes \ket{0}\bra{0} + (\identityOp \otimes \ket{1}\bra{1}) (U_2^\dagger \otimes H^\dagger) (\identityOp \otimes \ket{0}\bra{0}) (U_2 \otimes H) (\identityOp \otimes \ket{1}\bra{1})  \\
    &= \identityOp \otimes \left(\ket{0} \bra{0} + \frac{1}{2} \ket{1}\bra{1}\right) \\
    P_3 &= \identityOp \otimes \ket{0}\bra{0} + (\identityOp \otimes \ket{1}\bra{1}) (U_2^\dagger \otimes H^\dagger) \left(\identityOp \otimes (\ket{0} \bra{0} + \frac{1}{2} \ket{1}\bra{1}) \right) (U_2\otimes H) (\identityOp \otimes \ket{1}\bra{1}) \\
    & = \identityOp \otimes \ket{0}\bra{0} + \identityOp \otimes \ket{1}\bra{1} H^\dagger \left((\ket{0} \bra{0} + \frac{1}{2} \ket{1}\bra{1}) \right) H \ket{1}\bra{1}  \\
    &= \identityOp \otimes\left(\ket{0}\bra{0} + \frac{3}{4} \ket{1}\bra{1}\right)  \\
    &\dots \\
    P_n &= \identityOp \otimes \left(\ket{0}\bra{0} + \left(1 - \frac{1}{2^n}\right) \ket{1}\bra{1}\right)
\end{align*}
and overall we have $\qwp{\while_1}{\identityOp} = \bigvee_n P_n = \identityOp$.

\subsection{Example 2: One-Sided Quantum Walk}
\label{sec:appendix_example2}
For the calculation of AST, we show that $A = (\zeroOp_{<0} + \identityOp_{\geq 0}) \otimes \identityOp = \identityOp_{\geq 0} \otimes \identityOp$ is a fix-point of $\psi_\identityOp$:
\begin{align*}
    \psi_\identityOp(\identityOp_{\geq 0}\otimes \identityOp) =& \ket{0}\bra{0} \otimes \identityOp + ((\identityOp - \ket{0}\bra{0})\otimes H)S^\dagger ( \identityOp_{\geq 0} \otimes \identityOp) S  ((\identityOp - \ket{0}\bra{0}) \otimes H)\\
    =& \ket{0}\bra{0} \otimes \identityOp + ((\identityOp - \ket{0}\bra{0})\otimes H) S^\dagger( \sum_{n=0}^{\infty} \ket{n}\bra{n} \otimes \ket{0}\bra{0} +\sum_{n=0}^{\infty} \ket{n}\bra{n} \otimes \ket{1}\bra{1} )S\\
    &((\identityOp - \ket{0}\bra{0}) \otimes H)\\
    =& \ket{0}\bra{0} \otimes \identityOp + ((\identityOp - \ket{0}\bra{0})\otimes H) ( \sum_{n=0}^{\infty} \ket{n+1}\bra{n+1} \otimes \ket{0}\bra{0} +\sum_{n=0}^{\infty} \ket{n}\bra{n} \otimes \ket{1}\bra{1} )\\
    &((\identityOp - \ket{0}\bra{0}) \otimes H)\\
    =& \ket{0}\bra{0} \otimes \identityOp + ((\identityOp - \ket{0}\bra{0})\otimes H)\\
    & ( \sum_{n=1}^{\infty} \ket{n}\bra{n} \otimes \ket{0}\bra{0} +\sum_{n=1}^{\infty} \ket{n}\bra{n} \otimes \ket{1}\bra{1} + \ket{0}\bra{0} \otimes \ket{1}\bra{1})\\
    &((\identityOp - \ket{0}\bra{0}) \otimes H)\\
    =& \ket{0}\bra{0} \otimes \identityOp + ((\identityOp - \ket{0}\bra{0})\otimes H) ( \sum_{n=1}^{\infty} \ket{n}\bra{n} \otimes \identityOp + \ket{0}\bra{0} \otimes \ket{1}\bra{1})((\identityOp - \ket{0}\bra{0}) \otimes H)\\
    =& \ket{0}\bra{0} \otimes \identityOp +  \sum_{n=1}^{\infty} \ket{n}\bra{n} \otimes \identityOp = \identityOp_{\geq 0} \otimes \identityOp
\end{align*}

For finding the invariant for the calculation of ERT, we used the ansatz \begin{align*}
A = \sum_{n=0}^\infty a_n \ket{n}\bra{n} \otimes \ket{0}\bra{0} + \sum_{n=0}^\infty b_n \ket{n}\bra{n} \otimes \ket{1}\bra{1} + \zeroOp|_{\geq 0}.
\end{align*}
This ansatz is now used to find the coefficients $a_n \geq 0$ and $b_n \geq 0$. We obtain:
\begin{align*}
    \Psi_\zeroOp(A) =&  \sum_{n=0}^{\infty} \ket{n}\bra{n} \otimes T_n + \zeroOp|_{\geq 0} \\
\end{align*}
with \begin{align*}
    T_n = \begin{cases}
        \identityOp & n=0 \\
        (a_{n-1}+3) \ket{+}\bra{+} + (b_n+3) \ket{-}\bra{-} & n>0
    \end{cases}
\end{align*}

To ensure that $\Psi_\zeroOp(A) \sqsubseteq A$, we need to ensure that $T_n \sqsubseteq a_n \ket{0}\bra{0} + b_n \ket{1}\bra{1}$ for all $n$. Thus, we obtain the following result:
\begin{lemma}
    For a given $a_n$, the smallest $a_{n+1}$ such that $T_{n+1} \sqsubseteq a_{n+1} \ket{0}\bra{0} + b_{n+1} \ket{1}\bra{1}$ has a solution is $a_{n+1} = a_n + 12 + 6\sqrt{2}$. For that $a_n$, there is only one solution for $b_{n+1}$, which is $b_{n+1} = a_n + 6 + 3\sqrt{2}$.
\end{lemma}
The proof of this lemma is done in Lean and can be downloaded by clicking here: \textattachfile[color=0 0 1]{LemmaF1.lean}{\textcolor{blue}{\faDownload\ Download file}}
or by looking in the embedded files pane, depending on your PDF viewer.

\embedfile[desc={Proof of Lemma F.1}]{LemmaF1.lean}

Solving this recurrence gives
\begin{align*}
    a_n &= 1 + 2n(6+3\sqrt{2}) \\
    b_n &= \begin{cases} 1 + (2n-1)(6+3\sqrt{2}) & n>0 \\ 1 & n=0 \end{cases}
\end{align*}

Overall, we have the following invariant:
\begin{align*}
    A = \sum_{n=1}^\infty 2n(6+3\sqrt{2})\ket{n}\bra{n} \otimes \ket{0}\bra{0} + \sum_{n=1}^\infty (2n-1)(6+3 \sqrt{2}) \ket{n}\bra{n} \otimes \ket{1}\bra{1} + \identityOp + \zeroOp|_{\geq 0}.
\end{align*}

We can also show explicitly that $A$ is an upper-invariant of $\psi_\zeroOp$, i.e., $\psi_\zeroOp(A) \sqsubseteq A$:
\begin{align*}
    \Psi_\zeroOp(A) =& \sum_{n=0}^{\infty} \ket{n}\bra{n} \otimes \identityOp + \sum_{n=1}^\infty 2 \ket{n}\bra{n} \otimes \identityOp + ((\identityOp - \ket{0}\bra{0})\otimes H)S^\dagger \\
    & \left( \sum_{n=1}^\infty 2n(6+3\sqrt{2})\ket{n}\bra{n} \otimes \ket{0}\bra{0} + \sum_{n=1}^\infty (2n-1)(6+3 \sqrt{2}) \ket{n}\bra{n} \otimes \ket{1}\bra{1} + \identityOp + \zeroOp|_{\geq 0} \right) \\
    &S  ((\identityOp - \ket{0}\bra{0}) \otimes H) \\
    =& \sum_{n=0}^{\infty} \ket{n}\bra{n} \otimes \identityOp + \sum_{n=1}^\infty 2 \ket{n}\bra{n} \otimes \identityOp + ((\identityOp - \ket{0}\bra{0})\otimes H) \\
    &\left( \sum_{n=1}^\infty 2n(6+3\sqrt{2})\ket{n+1}\bra{n+1} \otimes \ket{0}\bra{0} + \sum_{n=1}^\infty (2n-1)(6+3 \sqrt{2}) \ket{n}\bra{n} \otimes \ket{1}\bra{1} + \identityOp \right) \\
    &((\identityOp - \ket{0}\bra{0}) \otimes H) +\zeroOp|_{\geq 0} \\
    =& \sum_{n=0}^{\infty} \ket{n}\bra{n} \otimes \identityOp + \sum_{n=1}^\infty 2 \ket{n}\bra{n} \otimes \identityOp + \sum_{n =1}^\infty\ket{n}\bra{n} \otimes \identityOp + \\
    &\sum_{n = 2}^\infty 2(n-1)(6+3\sqrt{2})\ket{n}\bra{n} \otimes \ket{+}\bra{+} + \sum_{n=1}^\infty (2n-1)(6+3 \sqrt{2}) \ket{n}\bra{n} \otimes \ket{-}\bra{-} + \zeroOp|_{\geq 0}\\
    =& \text{ } \identityOp + \sum_{n=1}^\infty 3 \ket{n}\bra{n} \otimes \identityOp - \sum_{n = 2}^\infty (6+3\sqrt{2})\ket{n}\bra{n} \otimes \ket{+}\bra{+} + (6+3 \sqrt{2}) \ket{1}\bra{1} \otimes \ket{-}\bra{-} +\\
    &\sum_{n = 2}^\infty (2n-1)(6+3\sqrt{2})\ket{n}\bra{n} \otimes \ket{+}\bra{+} + \sum_{n=2}^\infty (2n-1)(6+3 \sqrt{2}) \ket{n}\bra{n} \otimes \ket{-}\bra{-} +\zeroOp|_{\geq 0} \\
    =& \text{ } \identityOp + \sum_{n=1}^\infty 3 \ket{n}\bra{n} \otimes \identityOp - \sum_{n = 2}^\infty (6+3\sqrt{2})\ket{n}\bra{n} \otimes \ket{+}\bra{+} + (6+3 \sqrt{2}) \ket{1}\bra{1} \otimes \ket{-}\bra{-} +\\
    &\sum_{n=2}^\infty (2n-1)(6+3 \sqrt{2}) \ket{n}\bra{n} \otimes \identityOp  + \zeroOp|_{\geq 0}\\
    =& \text{ } \identityOp + \sum_{n=1}^\infty 3 \ket{n}\bra{n} \otimes \identityOp - \sum_{n = 2}^\infty (6+3\sqrt{2})\ket{n}\bra{n} \otimes \ket{+}\bra{+} + (6+3 \sqrt{2}) \ket{1}\bra{1} \otimes \ket{-}\bra{-} +\\
    &\sum_{n=2}^\infty (2n-1)(6+3 \sqrt{2}) \ket{n}\bra{n} \otimes \ket{0}\bra{0}+ \sum_{n=2}^\infty (2n-1)(6+3 \sqrt{2}) \ket{n}\bra{n} \otimes \ket{1}\bra{1} +\zeroOp|_{\geq 0}  \\
    =& \text{ } \identityOp + \sum_{n=1}^\infty 3 \ket{n}\bra{n} \otimes \identityOp - \sum_{n = 2}^\infty (6+3\sqrt{2})\ket{n}\bra{n} \otimes \left(\ket{+}\bra{+} + \ket{0}\bra{0}\right) + \\
    & (6+3 \sqrt{2}) \ket{1}\bra{1} \otimes \ket{-}\bra{-} + \sum_{n=2}^\infty 2n(6+3 \sqrt{2}) \ket{n}\bra{n} \otimes \ket{0}\bra{0}+ \\
    &\sum_{n=2}^\infty (2n-1)(6+3 \sqrt{2}) \ket{n}\bra{n} \otimes \ket{1}\bra{1} +\zeroOp|_{\geq 0}
\end{align*}
Now we show that $A - \Psi_\zeroOp(A) \sqsupseteq 0$:
\begin{align*}
    &A - \Psi_\zeroOp(A) \\
    =& \sum_{n=1}^\infty 2n(6+3\sqrt{2})\ket{n}\bra{n} \otimes \ket{0}\bra{0} + \sum_{n=1}^\infty (2n-1)(6+3 \sqrt{2}) \ket{n}\bra{n} \otimes \ket{1}\bra{1} + \identityOp + \zeroOp|_{\geq 0}- \\
    & \identityOp - \sum_{n=1}^\infty 3 \ket{n}\bra{n} \otimes \identityOp + \sum_{n = 2}^\infty (6+3\sqrt{2})\ket{n}\bra{n} \otimes (\ket{+}\bra{+} + \ket{0}\bra{0}) - (6+3 \sqrt{2}) \ket{1}\bra{1} \otimes \ket{-}\bra{-} - \\
    & \sum_{n=2}^\infty 2n(6+3 \sqrt{2}) \ket{n}\bra{n} \otimes \ket{0}\bra{0}- \sum_{n=2}^\infty (2n-1)(6+3 \sqrt{2}) \ket{n}\bra{n} \otimes \ket{1}\bra{1} - \zeroOp|_{\geq 0} \\
    =& \text{ } 2(6+3\sqrt{2})\ket{1}\bra{1} \otimes \ket{0}\bra{0} + (6+3 \sqrt{2}) \ket{1}\bra{1} \otimes \ket{1}\bra{1} + \\
    & \sum_{n = 2}^\infty (6+3\sqrt{2})\ket{n}\bra{n} \otimes (\ket{+}\bra{+} + \ket{0}\bra{0}) - \sum_{n=1}^\infty 3 \ket{n}\bra{n} \otimes \identityOp - (6+3 \sqrt{2}) \ket{1}\bra{1} \otimes \ket{-}\bra{-}\\
    =& \text{ } (6+3\sqrt{2})\ket{1}\bra{1} \otimes \ket{0}\bra{0} + (6+3 \sqrt{2}) \ket{1}\bra{1} \otimes \identityOp + \sum_{n = 2}^\infty (6+3\sqrt{2})\ket{n}\bra{n} \otimes (\ket{+}\bra{+} + \ket{0}\bra{0}) - \\
    & \sum_{n=1}^\infty 3 \ket{n}\bra{n} \otimes \identityOp - (6+3 \sqrt{2}) \ket{1}\bra{1} \otimes \ket{-}\bra{-} \\
    =& \text{ } (6+3\sqrt{2})\ket{1}\bra{1} \otimes \ket{0}\bra{0} + (6+3 \sqrt{2}) \ket{1}\bra{1} \otimes \ket{+}\bra{+} + \\
    & \sum_{n = 2}^\infty (6+3\sqrt{2})\ket{n}\bra{n} \otimes (\ket{+}\bra{+} + \ket{0}\bra{0}) - \sum_{n=1}^\infty 3 \ket{n}\bra{n} \otimes \identityOp \\
    =& \sum_{n = 1}^\infty (6+3\sqrt{2})\ket{n}\bra{n} \otimes (\ket{+}\bra{+} + \ket{0}\bra{0}) - \sum_{n=1}^\infty 3 \ket{n}\bra{n} \otimes \identityOp \\
    =& \sum_{n = 1}^\infty 3 \ket{n}\bra{n} \otimes \left((2+\sqrt{2})(\ket{+}\bra{+} + \ket{0}\bra{0}) - \identityOp \right) \sqsupseteq 0
\end{align*}

We have shown that $A$ is an upper-invariant of $\psi_\zeroOp$. Thus, we have $\ert{\while}{\zeroOp} \sqsubseteq A$ by Park induction. To get a bound on the expected runtime of the whole program, we have to consider the initialization of $c:= \ket{0}$:
By definition and monotonicity of $ert$, we have \begin{align*}
    &\ert{c:= \ket{0}; \while}{\zeroOp} \\
    =& \text{ } \ert{c:= \ket{0}}{\ert{\while}{\zeroOp}} \\
    \sqsubseteq& \text{ } \ert{c:= \ket{0}}{A} \\
    =&\sum_{n=1}^\infty 2n(6+3\sqrt{2})\ket{n}\bra{n} \otimes \identityOp + \zeroOp|_{\geq 0} + \identityOp
\end{align*}
thus we have $\expect{\ert{P_1}{\zeroOp}}{\ket{n} \otimes \psi} \leq 2n(6+3\sqrt{2}) +1$ for $n\geq 0$ and any state $\psi$.

Additionally, we show that $\ert{P_1}{\zeroOp}$ is not bounded from above by proving the lower bound $B = \sum_{n=0}^\infty n \ket{n}\bra{n} \otimes \identityOp$. This is done by considering $\ert{\while}{\zeroOp} = \bigvee_n F_n$ with
\begin{align*}
    F_0 &= \zeroOp \\
    F_{n+1} &= M_0^\dagger \zeroOp M_0 + M_1^\dagger \ert{S}{F_n} M_1 + \identityOp \\
    &= 3 \identityOp - 2 \ket{0}\bra{0} \otimes \identityOp + ((\identityOp - \ket{0}\bra{0})\otimes H)S^\dagger F_n S  ((\identityOp - \ket{0}\bra{0}) \otimes H).
\end{align*}

We define $G_k = \sum_{n=0}^{\infty} \min{\{n,k-1\}} \ket{n}\bra{n} \otimes \identityOp$ and show that $G_k \sqsubseteq F_k$ for all $k$ by induction:
\begin{itemize}
    \item $k=0$: $G_0 = \zeroOp \sqsubseteq F_0 = \zeroOp$
    \item $k \to k+1$: Assume $G_k \sqsubseteq F_k$. Then we have
    \begin{align*}
        F_{k+1} \sqsupseteq &\text{ } \identityOp + ((\identityOp - \ket{0}\bra{0})\otimes H)S^\dagger G_k S  ((\identityOp - \ket{0}\bra{0}) \otimes H) \\
        =& \text{ }\identityOp + \sum_{n=0}^{\infty} \min{\{n,k-1\}} \ket{n+1}\bra{n+1} \otimes \ket{+}\bra{+} + \sum_{n=1}^{\infty} \min{\{n,k-1\}} \ket{n}\bra{n} \otimes \ket{-}\bra{-} \\
        =&\text{ } \identityOp + \sum_{n=1}^{\infty} \min{\{n-1,k-1\}} \ket{n}\bra{n} \otimes \ket{+}\bra{+} + \sum_{n=1}^{\infty} \min{\{n,k-1\}} \ket{n}\bra{n} \otimes \ket{-}\bra{-} \\
        \sqsupseteq & \text{ } \identityOp + \sum_{n=1}^{\infty} \min{\{n-1,k-1\}} \ket{n}\bra{n} \otimes \identityOp \\
        \sqsupseteq & \sum_{n=0}^{\infty} \min{\{n,k\}} \ket{n}\bra{n} \otimes \identityOp  = G_{k+1}\\
    \end{align*}
\end{itemize}
As $G_k \sqsubseteq F_k$ for all $k$, we have $ B = \sum_{n=0}^{\infty} n \ket{n}\bra{n} \otimes \identityOp= \bigvee_{k=0}^\infty G_k \sqsubseteq \bigvee_{k=0}^\infty F_k = \ert{\while}{\zeroOp}$.
Again, we need to consider the initialization of $c:= \ket{0}$:
\begin{align*}
    &\ert{c:= \ket{0}; \while}{\zeroOp} \\
     =&\text{ }  \ert{c:= \ket{0}}{\ert{\while}{\zeroOp}} \\
    \sqsupseteq& \text{ } \ert{c:= \ket{0}}{B} = \sum_{n=1}^\infty n\ket{n}\bra{n} \otimes \identityOp + \identityOp \sqsupseteq B
\end{align*}

%% file: arxiv.bbl
\begin{thebibliography}{10}
\providecommand{\url}[1]{\texttt{#1}}
\providecommand{\urlprefix}{URL }
\providecommand{\doi}[1]{https://doi.org/#1}

\bibitem{expectationTransformer}
Avanzini, M., Moser, G., Pechoux, R., Perdrix, S., Zamdzhiev, V.: Quantum expectation transformers for cost analysis. In: Proceedings of the 37th Annual ACM/IEEE Symposium on Logic in Computer Science. LICS '22, Association for Computing Machinery, New York, NY, USA (2022). \doi{10.1145/3531130.3533332}

\bibitem{hardnessQuantitative}
Avanzini, M., Moser, G., Péchoux, R., Perdrix, S.: On the hardness of analyzing quantum programs quantitatively (2023), \url{https://arxiv.org/abs/2312.13657}

\bibitem{LicsNewPaper}
Barthe, G., Gao, M., Khan, J.K.A., Muis, M., Renison, I., Sakabe, K., Walter, M., Xu, Y., Yu, T., Zhou, L.: {Complete Relational Logic for Infinite-Dimensional Quantum Programs with Unbounded Assertions}. In: Faggian, C., Katoen, J.P. (eds.) 41st Annual Symposium on Logic in Computer Science (LICS 2026). Leibniz International Proceedings in Informatics (LIPIcs), vol.~380, pp. 15:1--15:28. Schloss Dagstuhl -- Leibniz-Zentrum f{\"u}r Informatik, Dagstuhl, Germany (2026). \doi{10.4230/LIPIcs.LICS.2026.15}

\bibitem{dualityTheorem}
Barthe, G., Gao, M., Khan, J.K.A., Muis, M., Renison, I., Sakabe, K., Walter, M., Xu, Y., Zhou, L.: {A Duality Theorem for Classical-Quantum States with Applications to Complete Relational Program Logics}  (10 2025), \url{https://arxiv.org/abs/2510.07051v1}

\bibitem{relationalHoareLogicTransport}
Barthe, G., Gao, M., Wang, T., Zhou, L.: Complete quantum relational hoare logics from optimal transport duality. In: 2025 40th Annual ACM/IEEE Symposium on Logic in Computer Science (LICS). pp. 884--925 (2025). \doi{10.1109/LICS65433.2025.00072}

\bibitem{lenaPOPL2023}
Batz, K., Kaminski, B.L., Katoen, J.P., Matheja, C., Verscht, L.: A calculus for amortized expected runtimes. Proc. ACM Program. Lang.  \textbf{7}(POPL) (Jan 2023). \doi{10.1145/3571260}

\bibitem{boyd2004convex}
Boyd, S., Vandenberghe, L.: Convex Optimization. Cambridge University Press (2004). \doi{10.1017/CBO9780511804441.016}

\bibitem{conway2000courseOperator}
Conway, J.B.: A Course in Operator Theory. Graduate Studies in Mathematics, American Mathematical Society (2000)

\bibitem{conway1994}
Conway, J.B.: A Course in Functional Analysis. Graduate Texts in Mathematics, Springer New York (2007), \url{https://doi.org/10.1007/978-1-4757-4383-8}

\bibitem{DENG202273}
Deng, Y., Feng, Y.: Formal semantics of a classical-quantum language. Theoretical Computer Science  \textbf{913},  73--93 (2022), \url{https://doi.org/10.1016/j.tcs.2022.02.017}

\bibitem{DHondtWeakestPreconditions}
D'Hondt, E., Panangaden, P.: Quantum weakest preconditions p. 429–451 (2006). \doi{10.1017/S0960129506005251}

\bibitem{Dijkstra75}
Dijkstra, E.W.: Guarded commands, nondeterminacy and formal derivation of programs. Commun. ACM  \textbf{18}(8),  453–457 (1975). \doi{10.1145/360933.360975}

\bibitem{Dijkstra76}
Dijkstra, E.W.: A Discipline of Programming. Prentice-Hall (1976)

\bibitem{FengNondeterministicQuantumVerification}
Feng, Y., Xu, Y.: Verification of nondeterministic quantum programs. In: Proceedings of the 28th ACM International Conference on Architectural Support for Programming Languages and Operating Systems, Volume 3. p. 789–805. ASPLOS 2023, Association for Computing Machinery, New York, NY, USA (2023). \doi{10.1145/3582016.3582039}

\bibitem{FengQHLClassicalVars}
Feng, Y., Ying, M.: Quantum {H}oare logic with classical variables. ACM Transactions on Quantum Computing  \textbf{2}(4) (2021). \doi{10.1145/3456877}

\bibitem{krausChoi}
Friedland, S.: Infinite dimensional generalizations of {C}hoi’s theorem. Special Matrices  \textbf{7},  67--77 (07 2019). \doi{10.1515/spma-2019-0006}

\bibitem{bayesianInf}
Gehnen, C., Unruh, D., Katoen, J.P.: {Bayesian Inference in Quantum Programs}. In: Censor-Hillel, K., Grandoni, F., Ouaknine, J., Puppis, G. (eds.) 52nd International Colloquium on Automata, Languages, and Programming (ICALP 2025). Leibniz International Proceedings in Informatics (LIPIcs), vol.~334, pp. 157:1--157:18. Schloss Dagstuhl -- Leibniz-Zentrum f{\"u}r Informatik, Dagstuhl, Germany (2025). \doi{10.4230/LIPIcs.ICALP.2025.157}

\bibitem{hall2013}
Hall, B.C.: Quantum Theory for Mathematicians. No.~267 in Graduate Texts in Mathematics, Springer New York, \url{https://doi.org/10.1007/978-1-4614-7116-5}

\bibitem{hoare69}
{H}oare, C.A.R.: An axiomatic basis for computer programming. Commun. ACM  \textbf{12}(10),  576–580 (Oct 1969). \doi{10.1145/363235.363259}

\bibitem{probRuntime}
Kaminski, B.L., Katoen, J.P., Matheja, C., Olmedo, F.: Weakest precondition reasoning for expected run--times of probabilistic programs. In: Thiemann, P. (ed.) Programming Languages and Systems. pp. 364--389. Springer Berlin Heidelberg, Berlin, Heidelberg (2016), \url{https://doi.org/10.1007/978-3-662-49498-1_15}

\bibitem{ProbERTJournal}
Kaminski, B.L., Katoen, J.P., Matheja, C., Olmedo, F.: Weakest precondition reasoning for expected runtimes of randomized algorithms. J. ACM  \textbf{65}(5) (Aug 2018). \doi{10.1145/3208102}

\bibitem{kato}
Kat{\=o}, T.: Perturbation Theory for Linear Operators. Die Grundlehren der mathematischen Wissenschaften in Einzeldarstellungen, Springer-Verlag (1966), \url{https://books.google.de/books?id=N_HysgEACAAJ}

\bibitem{KOZEN1985162}
Kozen, D.: A probabilistic {PDL}. Journal of Computer and System Sciences  \textbf{30}(2),  162--178 (1985). \doi{https://doi.org/10.1016/0022-0000(85)90012-1}

\bibitem{ExpectationValuesKraus}
Kraus, K., Schr{\"o}ter, J.: Expectation values of unbounded observables. International Journal of Theoretical Physics  \textbf{7}(6),  431--442 (1973). \doi{10.1007/BF00713245}, \url{https://doi.org/10.1007/BF00713245}

\bibitem{LiuRuntime}
Liu, J., Zhou, L., Barthe, G., Ying, M.: Quantum weakest preconditions for reasoning about expected runtimes of quantum programs. In: Proceedings of the 37th Annual ACM/IEEE Symposium on Logic in Computer Science. LICS '22, Association for Computing Machinery, New York, NY, USA (2022). \doi{10.1145/3531130.3533327}

\bibitem{McIverWpProb}
McIver, A., Morgan, C.: Abstraction, Refinement and Proof for Probabilistic Systems. Monographs in Computer Science, Springer (2005). \doi{10.1007/b138392}

\bibitem{Nielsen_Chuang_2010}
Nielsen, M.A., Chuang, I.L.: Quantum Computation and Quantum Information: 10th Anniversary Edition. Cambridge University Press (2010). \doi{10.1017/CBO9780511976667}

\bibitem{olmedoRuntime}
Olmedo, F., Díaz-Caro, A.: Runtime analysis of quantum programs: A formal approach (2019), \url{https://arxiv.org/abs/1911.11247}

\bibitem{park}
Park, D.: Fixpoint induction and proofs of program properties. Machine Intelligence (5) (1969)

\bibitem{reed1980methods}
Reed, M., Simon, B.: Methods of Modern Mathematical Physics: Functional analysis. No. Bd. 1 in Methods of Modern Mathematical Physics, Academic Press (1980), \url{https://books.google.de/books?id=bvuRuwuFBWwC}

\bibitem{AST78}
Saheb-Djahromi, N.: Probabilistic lcf. In: Winkowski, J. (ed.) Mathematical Foundations of Computer Science 1978. pp. 442--451. Springer Berlin Heidelberg, Berlin, Heidelberg (1978), \url{https://doi.org/10.1007/3-540-08921-7_92}

\bibitem{unboundedSelfAdjointBook}
Schmüdgen, K.: Unbounded Self-adjoint Operators on Hilbert Space. Springer Dordrecht (2012). \doi{10.1007/978-94-007-4753-1}

\bibitem{philippHigherMoments}
Schr{\"o}er, P., Katoen, J.P.: Highly incremental: A simple programmatic approach for many objectives. In: Sampaio, A., Stoelinga, M. (eds.) Formal Methods. pp. 559--577. Springer Nature Switzerland, Cham (2026), \url{https://doi.org/10.1007/978-3-032-26204-2_29}

\bibitem{takesaki1979theory}
Takesaki, M.: Theory of Operator Algebras I. No. Bd. 1 in Encyclopaedia of Mathematical Sciences, Springer New York (1979), \url{https://doi.org/10.1007/978-1-4612-6188-9}

\bibitem{isabelleproofTraceclassSum}
Unruh, D.: The tensor product on {H}ilbert spaces. Arch. Formal Proofs  \textbf{2024}, \url{https://www.isa-afp.org/entries/Hilbert\_Space\_Tensor\_Product.html}, {L}emma Trace\_Class.trace\_class\_decomp\_4pos

\bibitem{heisenbergdualityUnruh}
Unruh, D.: Quantum references  (2024), \url{https://arxiv.org/abs/2105.10914}

\bibitem{isabelleproofKraus}
Unruh, D.: Kraus maps. Arch. Formal Proofs  (June 2025), \url{https://isa-afp.org/entries/Kraus\_Maps.html}, {L}emma kraus-map-infsum

\bibitem{floydHoareLogic}
Ying, M.: Floyd--{H}oare logic for quantum programs. ACM Trans. Program. Lang. Syst.  (2012). \doi{10.1145/2049706.2049708}

\bibitem{YingPredicateTranserformerSemantics}
Ying, M., Duan, R., Feng, Y., Ji, Z.: Predicate transformer semantics of quantum programs. Semantic Techniques in Quantum Computation  (2010). \doi{10.1017/CBO9781139193313.009}

\bibitem{ZhouAppliedQHL}
Zhou, L., Yu, N., Ying, M.: An applied quantum {H}oare logic. In: Proceedings of the 40th ACM SIGPLAN Conference on Programming Language Design and Implementation. p. 1149–1162. PLDI 2019, Association for Computing Machinery, New York, NY, USA (2019). \doi{10.1145/3314221.3314584}

\end{thebibliography}
